\documentclass[11pt]{article}
\usepackage[margin=1in]{geometry}

\usepackage{authblk}
\usepackage[numbers]{natbib}

\usepackage{parskip}

\usepackage{hyperref}       
\usepackage{url}            
\usepackage{booktabs}       
\usepackage{amsfonts}       
\usepackage{nicefrac}       
\usepackage{microtype}      
\usepackage{graphics}
\usepackage{caption}
\usepackage{subcaption}

\usepackage{amssymb,amsthm}
\usepackage{amsfonts}
\usepackage{latexsym,mathtools,enumitem,ifthen}

\usepackage{commands}

\begin{document}

\title{Exploiting Neighborhood Interference with Low Order Interactions under Unit Randomized Design}

\author[1]{Mayleen Cortez}
\author[1]{Matthew Eichhorn}
\author[2]{Christina Lee Yu}
\affil[2]{School of Operations Research and Information Engineering, Cornell University}
\affil[1]{Center for Applied Mathematics, Cornell University}

\maketitle

\begin{abstract}%
Network interference, where the outcome of an individual is affected by the treatment assignment of those in their social network, is pervasive in real-world settings. However, it poses a challenge to estimating causal effects. We consider the task of estimating the total treatment effect (TTE), or the difference between the average outcomes of the population when everyone is treated versus when no one is, under network interference. Under a Bernoulli randomized design, we provide an unbiased estimator for the TTE when network interference effects are constrained to low order interactions among neighbors of an individual. We make no assumptions on the graph other than bounded degree, allowing for well-connected networks that may not be easily clustered. We derive a bound on the variance of our estimator and show in simulated experiments that it performs well compared with standard estimators for the TTE. We also derive a minimax lower bound on the mean squared error of our estimator which suggests that the difficulty of estimation can be characterized by the degree of interactions in the potential outcomes model. We also prove that our estimator is asymptotically normal under boundedness conditions on the network degree and potential outcomes model. Central to our contribution is a new framework for balancing model flexibility and statistical complexity as captured by this \textit{low order interactions} structure.
\end{abstract}

\section{Introduction}

Accurately estimating causal effects is relevant in numerous applications, from pharmaceutical companies researching the efficacy of a new medication, to policy makers understanding the impact of social welfare programs, to social media companies evaluating the impact of different recommendation algorithms on user engagement across their platforms. Often, the entity interested in understanding a causal effect will design an experiment where they randomly assign subsets of the population to treatment (e.g. new medication) and to control (e.g. a placebo) and draw conclusions based on the observed outcomes of the participants (e.g. health outcomes). Other times, the entity may need to determine causal effects from observational data accrued in a previous study where they did not have full control over the treatment assignment mechanism. 

Our work focuses on estimating the \textit{total treatment effect} (TTE), or the difference between the average outcomes of the population when everyone is treated versus when no one is treated, given data collected from a \textit{randomized experiment}. This estimand is sometimes referred to as the global average treatment effect (GATE), as in \cite{UganderYin2020}. The TTE is a quantity of interest in scenarios where the decision maker must choose between adopting the proposed treatment or sticking with the status quo. For example, a social media company may develop a new recommendation algorithm for suggesting content to their users, and they want to decide whether or not to roll out this new algorithm across their platform. As another example, suppose that a pharmaceutical company develops a vaccine for some infectious disease. Then, public health experts and officials must decide whether the vaccine is safe and efficacious enough to warrant its recommendation to the general population. As the side effects of the new treatment are unknown, the goal is to determine the efficacy of the treatment relative to the status quo baseline by running a budgeted randomized trial, where the number of treated individuals in the trial is limited for safety reasons.

The techniques and guarantees for estimating causal effects in classical causal inference heavily rely upon the stable unit treatment value assumption (SUTVA), which posits that the outcome of each individual is independent of the treatment assignment of all other individuals \cite{Rubin1980SUTVA}. Unfortunately, SUTVA is violated in all the above mentioned applications because people influence and are impacted by their peers. In the presence of this \textit{network interference}, an individual's outcome is affected by the treatment assignment of others in their social network, and SUTVA no longer holds. Distinguishing between the direct effect of treatment on an individual and the network effect of others' treatment on the individual can be challenging.
This has resulted in a growing literature on causal inference in the presence of network (interference) effects, sometimes referred to as spillover or peer influence effects. In this work, we consider the task of estimating the TTE from unit randomized trials under \textit{neighborhood interference}, when an individual is affected by the treatment of its direct neighbors but is otherwise unaffected by the treatment of the individuals outside their neighborhood. Furthermore, we focus on \textit{unit randomized designs}, wherein individuals are independently assigned to either treatment or control in a randomized experiment. This is in contrast to cluster randomized designs, which has been proposed as an approach to address network interference for randomized experimental design but which may not be feasible in practice due to an incompatibility with existing experimental platforms.

Estimating a causal effect from randomized experiments involves two decisions. First, we must decide what kind of experiment to run, i.e. how we choose the individuals to participate in the study, and how we determine which individuals are assigned to receive the new treatment. In this paper we focus on randomization inference, in which we assume the entire population participates in the study, and the randomness arises from the assignment of treatment or control. Second, after the experiment is conducted, we must decide how to analyze the data and construct an estimate from the measured observations. The literature addressing network interference can largely be categorized as either constructing clever randomized designs to exploit clustering structure in the network, or designing new estimators that exploit structure in the potential outcomes model. Without making assumptions on either the network or the potential outcomes model, it is impossible to estimate the TTE under arbitrary interference \cite{AronowSamii17,Manski13,BasseAiroldi17}.

Table \ref{tab:literature} summarizes how different assumptions on the potential outcomes model or network structure lead to different types of solutions, with a focus on unbiased estimators and neighborhood interference models where the network is known. The columns correspond to assumptions on network structure (in order of increasing generality from left to right), while the rows correspond to assumptions on the structure in the potential outcomes model (in order of increasing generality from top to bottom). The body of the graph lists proposed solutions for estimating the TTE under the corresponding assumptions on network/model structure. For example, under a fully general network structure and a linear potential outcomes model, \cite{ToulisKao13,GuiXuBhasinHan15,BasseAiroldi15,cai2015social,parker2016optimal,chin2019regression} propose using an ordinary least squares (OLS) estimator with a Bernoulli randomized design (RD). Since we focus on solutions proposing \textit{unbiased} estimators, we also list bounds on the variance of the estimators in the table, as a point of comparison. As Table \ref{tab:literature} indicates, the literature has either focused on analyzing the Horvitz-Thompson estimator under new randomized designs that exploit network structure by increasing the correlation between neighbors' treatment assignments or alternatively using regression-style estimators with Bernoulli randomized design, exploiting strong functional assumptions on the potential outcomes model. Without assuming any structure beyond neighborhood interference, the baseline solution of using the Horvitz-Thompson estimator under Bernoulli randomized design is an unbiased estimator whose variance scales exponentially in the degree of the network.

\newcommand{\nl}{\hspace{-24pt} \newline}

\begin{table}
\centering
\def\arraystretch{1.3}
\begin{tabular}{P{70pt}|P{85pt}|P{125pt}|P{100pt}} 
    & \multicolumn{3}{c}{Network Structure} \\ \hhline{~|---} 
    Model Structure & $C$ Disconnected \nl Subcommunities & $\kappa$-restricted Growth & \cellcolor{yellow} Fully General \\ \hhline{----}
\textbf{Linear} & \multicolumn{2}{c|}{\multirow{3}{*}[-20pt]{Directions for Future Work}} & OLS, \textit{Bernoulli RD}; \cite{ToulisKao13,GuiXuBhasinHan15,BasseAiroldi15,cai2015social,parker2016optimal,chin2019regression} \\ \hhline{-|~~|-}
\textbf{Generalized Linear} & \multicolumn{2}{c|}{} & Regression / Machine \nl Learning, \textit{Bernoulli RD}; \cite{chin2019regression} \\
\hhline{-|~~|-}
 \cellcolor{yellow} \textbf{$\beta$-order} \nl \textbf{Interactions} &  \multicolumn{2}{c|}{}  & \cellcolor{yellow} SNIPE, \textit{Bernoulli RD}; \nl $\displaystyle{O\left(\frac{Y_{\max}^2 d^{2\beta + 2}}{np^{\beta}}\right)}$ \nl \\ \hhline{-|--|-}
{\textbf{Arbitrary} \nl \textbf{Neighborhood} \nl \textbf{Interference}}  
& \textit{Horvitz-Thompson}, \nl Cluster RD; \nl $\displaystyle{O\left(\frac{Y_{\max}^2}{Cp}\right)}$;\space \space\cite{Sobel06,Rosenbaum07,HudgensHalloran08,TchetgenVanderWeele12}  
& \textit{Horvitz-Thompson}, \nl Randomized Cluster RD; \nl $\displaystyle{O\left(\frac{Y_{\max}^2 \kappa^4 d^2}{np}\right)}$;\space \space\cite{GuiXuBhasinHan15,EcklesKarrerUgander17,UganderKarrerBackstromKleinberg13,UganderYin2020} 
& \textit{Horvitz-Thompson}, \textit{Bernoulli RD}; $\displaystyle{O\left(\frac{Y_{\max}^2 d^2}{np^d}\right)}$;\space \space\cite{AronowSamii17}\\
\end{tabular}

\caption{Literature Summary. Each row corresponds to an assumption on the structure of the potential outcomes model and each column corresponds to an assumption on the structure of network.
We list a proposed solution under the corresponding model/network assumptions in the following order: an \underline{unbiased} estimator for the total treatment effect, the randomized design (RD), a bound on the variance (if available), and citations to related work. In the variance bounds, $Y_{\max}$ is a bound on the effects on any individual, $\dmax$ is the maximum neighborhood size, $C$ is the number of subcommunities or clusters, $p$ is the treatment {probability} which is assumed to be small, $\kappa$ is the restricted growth parameter, and $\beta$ is the polynomial degree of the potential outcomes model. Our result proposes an estimator for the TTE under a $\beta-$order interactions (equivalently, $\beta-$degree polynomial) structure, a fully general graph, and Bernoulli design. All the solutions in the table rely on full knowledge of the network under neighborhood interference and we focus on unbiased estimators.}
\label{tab:literature}
\end{table}

In our work, we propose a hierarchy of model classes that extrapolates between simple linear models and complex general models, such that a practitioner can choose the strength of the model assumptions they are willing to impose. Naturally, assuming a more limited model simplifies the causal inference task. We characterize the complexity of a potential outcomes model with the order of interactions $\beta$, which also corresponds to the polynomial degree of the potential outcomes function when viewed as a polynomial of the treatment vector. A $\beta$-order interactions model is one in which each neighborhood set of size at most $\beta$ can have a unique additive network effect on the potential outcome. Our model allows for heterogeneity in the influence of different sets of treated neighbors, strictly generalizing beyond the typical parametric model classes used in the literature, which oftentimes assumes anonymous interference. We make no assumptions on the graph beyond bounded degree, so the graph may be well connected and not easily clustered. We summarize our contributions and results below:

\begin{enumerate}[label=(\arabic*)]
\item Assuming a $\beta$-order interactions model, under a non-uniform Bernoulli randomized design with treatment probabilities $\{p_i\}_{i=1}^n$ for $p_i \in [p,1-p]$, we present the Structured Neighborhood Interference Polynomial Estimator (SNIPE), a simple unbiased estimator for the TTE.

\item We derive a bound on the variance of our estimator which scales polynomially in the degree $d$ of the network and exponentially in degree $\beta$ of the potential outcomes model. We also show that our estimator is asymptotically normal.

\item For a $d$-regular graph and uniform treatment probabilities $p_i = p$ with $p^\beta < 0.16$ and $\beta \ll d$, we prove that the minimax optimal mean squared error for estimating the TTE is lower bounded by $\Omega\big(\frac{1}{np^\beta}\big)$, implying that the exponential dependence on $\beta$ is necessary. 

\item We provide experimental results to illustrate that using regression models that do not allow for heterogeneity amongst the network effects can lead to considerable bias when the anonymous interference assumption is violated. The experiments validate that our estimator is unbiased for $\beta$-order interaction models, and obtains a lower mean squared error than existing alternatives.
\end{enumerate}

To interpret the upper bound on the variance of our proposed estimator for the TTE, we note that our variance scales as $O(\text{poly}(d)/np^{\beta})$, compared to the variance of the Horvitz-Thompson estimator under Bernoulli randomized design scales as $O(1/np^d)$, where $\beta$ is always bounded above by $d$. For smaller values of $\beta$, the $\beta$-order interactions model imposes stronger structural assumptions on the potential outcomes model than required by the Horvitz-Thompson estimator. In turn, our estimator has significantly lower variance, scaling only polynomially in the network degree $d$ yet exponential in $\beta$, as opposed to exponential in $d$. In addition, the minimax lower bound shows that the exponential dependence on $\beta$ is also minimax optimal amongst $\beta$-order interaction models, implying that the order of interactions, or polynomial degree, of the potential outcomes model is a meaningful property that expresses the complexity of estimating the total treatment effect.
\section{Related Literature}
While there has been a flurry of recent activity in addressing the challenges that arise from network interference, every proposed solution concept fundamentally hinges upon making key assumptions on the form of network interference. Without any assumptions, the vector of observed outcomes under a particular treatment vector $\bz \in \{0,1\}^n$ may have no relationship to the potential outcomes under any other treatment vector. We should naturally expect that if we are willing to assume stronger assumptions, then we may be able to obtain stronger results conditioned on those assumptions being satisfied. As such, the literature ranges between results that impose fewer assumptions on the model and graph, resulting in unbiased estimators with high variance, or results that impose strong assumptions on the model and graph, resulting in simple, unbiased estimators with low variance.

Early works propose framing assumptions via exposure functions, or constant treatment response \cite{Manski13, AronowSamii17}. This assumes that there is some known exposure mapping, $f(\bz, \theta_i) \in \Delta$, which maps the treatment vector $\bz$, along with unit-specific traits $\theta_i$, to a discrete space $\Delta$ of \textit{exposures}, or effective treatments. The potential outcomes function for unit $i$ is then assumed to only be a function of its \textit{exposure}, or effective treatment, such that if $f(\bz, \theta_i) = f(\bz',\theta_i)$, then $Y_i(\bz) = Y_i(\bz')$. If $|\Delta|$ is as large as $2^n$, then this assumption places no effective restriction on the potential outcomes function; thus this assumption is only meaningful when $|\Delta|$ is relatively small. One commonly used exposure mapping expresses the neighborhood interference assumption \cite{UganderKarrerBackstromKleinberg13,  SussmanAiroldi17, bargagli2020heterogeneous, pmlr-v115-bhattacharya20a}, in which each unit $i$ is associated to a neighborhood $\cN_i \subseteq [n]$, and unit $i$'s potential outcome is only a function of the treatments of units within $|\cN_i|$. We could use an exposure mapping to formulate this assumption, where $|\Delta| = 2^{d}$ for $d$ denoting the maximum neighborhood size, and $f(\bz,\cN_i) = f(\bz',\cN_i)$ if the treatments assigned to individuals in $\cN_i$ are the same between $\bz$ and $\bz'$. Another commonly used exposure mapping expresses the anonymous interference assumption \cite{HudgensHalloran08,LiuHudgens14,WagerLi2020,Davide2020}, in which the potential outcomes are only a function of the treatments through the number of treated neighbors, i.e. $f(\bz,\cN_i) = f(\bz',\cN_i)$ if the number of treated individuals in $\cN_i$ via $\bz$ and $\bz'$ are the same. While the exposure mapping framework provides a highly generalizable and abstract framework for assumptions, it is fundamentally discrete in nature and the complexity of estimation is characterized by the number of possible exposures $|\Delta|$, which could still be large. As a result, \cite{Manski13} suggests a collection of additional assumptions that can be imposed on top of anonymous neighborhood interference, including distributional or functional form assumptions, or additivity assumptions as suggested in \cite{SussmanAiroldi17}.



The majority of works in the literature (along with our work) assume neighborhood interference with respect to a known graph. Notable exceptions include \cite{Leung2021}, which considers a highly structured \textit{spatial interference} setting with network effects decaying with distance, and \cite{AtheyEcklesImbens17}, which provides methods for testing hypotheses about interference effects including higher order spillovers. Without imposing any additional assumptions on the potential outcomes besides neighborhood interference, a natural solution is to use the Horvitz-Thompson estimator to estimate the average outcomes under full neighborhood treatment and full neighborhood control \cite{AronowSamii17}. While the estimator is unbiased, the variance of the estimator scales inversely with the probability that a unit's full neighborhood is in either treatment or control. Under a Bernoulli($p$) randomized design, where each individual is treated independently with probability $p$, the variance scales as $O(1/np^d)$, as indicated in the bottom right cell of Table \ref{tab:literature}. The exponential dependence on $d$ renders the estimator impractical for realistic networks.

One approach in response to the above challenge is to consider cleverly constructing randomized designs that increase overlap, i.e. the probability that a unit's entire neighborhood is assigned to treatment or control.
The earliest literature in this line of work additionally assumes \textit{partial interference}, also referred to as the multiple networks setting, in which the population can be partitioned into disjoint groups, and network interference only occurs within groups and not across groups \cite{Sobel06,HudgensHalloran08,TchetgenVanderWeele12,LiuHudgens14,VanderweeleTchetgenHalloran14,pmlr-v115-bhattacharya20a, auerbach2021local}.
This assumption makes sense in contexts where interference is only expected between naturally clustered groups of individuals, such as households, cities or countries. 
Given knowledge of the groups, we can then randomly assign groups to different treatment saturation levels, e.g. jointly assigning groups to either treatment or control, increasing the correlation of treatments within neighborhoods. Then, a difference in means estimator or a Horvitz-Thompson estimator can be used to estimate the TTE. The asymptotic consistency of these estimators relies on the number of groups going to infinity, with a variance scaling inversely with the fraction of treated clusters, i.e. $O(1/Cp)$ as indicated in the bottom left cell of Table \ref{tab:literature}. 
In practice, even networks that are clearly clustered into separate groups may not have a sufficiently large number of groups to result in accurate estimates.

For general connected graphs one can still implement a cluster-based randomized design on constructed clusters, where the clusters are constructed to minimize the number of edges between clusters \cite{ UganderKarrerBackstromKleinberg13,GuiXuBhasinHan15, EcklesKarrerUgander17,UganderYin2020}. 
\cite{UganderKarrerBackstromKleinberg13,UganderYin2020} provide guarantees for graphs exhibiting a restricted growth property, which limits the rate at which local neighborhoods of the graph can grow in size, and \cite{UganderYin2020} prove that using randomized cluster randomized design along with the Horvitz-Thompson estimator achieves a variance of $O(1/np)$ for restricted growth networks, which is a significant gain from the exponential dependence on $d$ under Bernoulli design. 
A limitation of these solutions is that cluster randomized design can be difficult to implement due to incompatibility with existing experimentation platforms or to ethical concerns. When the existing experimentation platform is already set up for a unit randomized design, the experimenter may have the desire to avoid overhauling the platform to work with cluster randomized design due to time or resource constraints \cite{chin2019regression}. Additionally, \cite{taljaard2009ethical, edwards1999ethical, HuttonEthicalCRCTs, DonnerAllanPitfalsCRT} detail some ethical issues that arise in cluster randomized designs including, but not limited to, problems with informed consent (e.g. it may be infeasible to gain informed consent from every unit in a cluster) and concerns about distributional justice (e.g. with regards to how clusters are selected and assigned to treatment). Furthermore, both \cite{taljaard2009ethical} and \cite{HuttonEthicalCRCTs} comment that many of the existing, formal research ethics guidelines were designed with unit randomized design in mind and thus, offer little guidance for considerations that arise in cluster randomized designs. The work we present here focuses on scenarios with unit randomized designs, when cluster randomized designs may not be feasible or may be undesirable due to any of the aforementioned concerns.

The alternate approach that has gained traction empirically is to impose additional functional assumptions on the potential outcomes in addition to anonymous neighborhood interference. The most common assumption is that the potential outcomes are linear with respect to a particular statistic of the treatment vector, where the linear function is shared across units \cite{ToulisKao13, GuiXuBhasinHan15, BasseAiroldi15,cai2015social,parker2016optimal,chin2019regression}. 
Under this assumption, estimating the entire potential outcomes function reduces to linear regression, which one could solve using ordinary least squares, as indicated in the upper right most cell of Table \ref{tab:literature}. After recovering the linear model, one could estimate any desired causal estimand. The results rely on correctly choosing the statistic that governs the linearity, or more generally reduces the task to heuristic \textit{feature engineering} for generalized linear models \cite{chin2019regression}. One could plausibly extend the function class beyond linear and instead apply machine learning techniques to estimate the function that generalizes beyond linear regression. While we do not state a variance bound in Table \ref{tab:literature}, one would expect that when $p$ is sufficiently large, the variance would scale with $O(1/n)$, with some dependence on the complexity of the model class; when $p$ is very small, the variance may scale with $O(1/pn)$, as the regression still requires sufficient variance of covariates represented in the data, i.e. a sufficient spread of number of neighbors treated. \cite{WagerLi2020} considers nonparametric models yet focuses on estimating a locally linearized estimand. A drawback of these approaches is that they assume \textit{anonymous interference}, which imposes a symmetry in the potential outcomes such that the causal network effect of treating any of an individual's neighbors is equal regardless of which neighbor is chosen. Additionally they assume that the function that we are learning is shared across units, or at least units of the same known covariate type, which can be limiting.

While we have primarily focused on summarizing the literature for unbiased estimators, there has also been some work considering how structure in the potential outcomes model can be exploited to reduce bias of standard estimators. \cite{johari2022experimental,li2022interference, spang2021unbiased} use application domain specific structure in two-sided marketplaces and network routing to compare the bias of the difference in means estimators under different experimental designs. \cite{bright2022reducing} uses structure in ridesharing platforms to propose a new estimator with reduced bias. Additionally there has been some limited empirical work studying bias under model misspecification \cite{EcklesKarrerUgander17}.

Complementary to the literature on randomized experiments, there has been a growing literature considering causal inference over observational studies in the presence of network interference. However, the limitations similarly mirror the above stated concerns. A majority of the literature assumes partial interference \cite{TchetgenVanderWeele12, perez2014assessing, liu2016inverse, DiTraglia2020, vazquez-bare2022}, i.e. the multiple networks setting, which then enables causal inference of a variety of different estimands. In particular, it is commonly assumed that the different groups in the network are sampled from some underlying distribution, and the statistical guarantees are given with respect to the number of groups going to infinity. Alternately, other works assume that the potential outcomes only depend on a simple and known statistic of the neighborhood treatment, most commonly the number or fraction of treated \cite{verbitsky2012causal, chin2019regression, ogburn2017causal}. Either the neighborhood statistic only takes finite values, or assumptions are imposed on the functional form of the potential outcomes which imply anonymous interference and reduce inference to a regression task or maximum likelihood calculation. \cite{forastiere2021identification} considers a general exposure mapping model alongside an inverse propensity weighted estimator, but the estimator has high variance when the exposure mapping is complex.

In contrast to the majority of the mentioned literature, we neither rely on cluster randomized designs nor anonymous interference assumptions. We instead propose a potential outcomes model with \textit{low order interactions} structure, where the degree of interactions $\beta$ characterizes the difficulty of inference, also studied in \cite{YuCortezEichhorn22}. For $\beta=1$, this model is equivalent to heterogeneous additive network effects in \cite{YuAiroldiBorgsChayes22}, which can be derived from the joint assumptions of additivity of main effects and interference in \cite{SussmanAiroldi17}. When $\beta$ is larger than the network degree, then this assumption is equivalent to an arbitrary neighborhood interference, providing a nested hierarchy of models that encompass both the simple linear model class, the fully general model class, as well as model classes of varying complexity in between. \cite{YuAiroldiBorgsChayes22} and \cite{YuCortezEichhorn22}, which consider similar models as we do, focus on the setting when the underlying network is fully unknown, and yet there is richer available information either in the form of historical data \cite{YuAiroldiBorgsChayes22} or multiple measurements over the course of a multi-stage experiment \cite{YuCortezEichhorn22}, neither of which we assume in this work. Our estimator also has close connections to the pseudoinverse estimator in \cite{swaminathan2017off}\mcedit{, the Riesz estimator in \cite{harshaw2022design},} and the estimator in \cite{WagerLi2020}, which is discussed in \Cref{sec:estimator}.

\section{Model}
\label{sec:model}

\subsection{Causal Network}

Let $[n] := \{ 1, \hdots, n \}$ denote the underlying population of $n$ individuals. We model the network effects in the population as a directed graph over the individuals with edge set $E \subseteq [n] \times [n]$. An edge $(j,i) \in E$ signifies that the treatment assignment of individual $j$ affects the outcome of individual $i$. As an individual's own treatment is likely to affect their outcome, we expect self-loops in this graph. 
In much of the paper, we are concerned with neighborhood interference effects and we use $\cN_i := \{ j \in [n]: (j,i) \in E \}$ to denote the in-neighborhood of an individual $i$. Note that this definition allows $i \in \cN_i$. Many of our variance bounds reference the network degree. We let $\din$ denote the maximum in-degree of any individual, $\dout$ denote the maximum out-degree, and $\dmax := \max \{ \din, \dout\}$.

\subsection{Potential Outcomes Model}

To each individual $i$, we associate a treatment assignment $z_i \in \{0,1\}$, where we interpret $z_i=1$ as an assignment to the treatment group and $z_i=0$ as an assignment to the control group. We collect all treatment assignments into the vector $\bz$. We use $Y_i$ to denote the outcome of individual $i$. As our setting assumes network interference, the classical SUTVA assumption is violated. That is, $Y_i$ is not a function only of $z_i$. Rather, $Y_i : \{0,1\}^n \to \mathbb{R}$ may be a function of $\bz$, the treatment assignments of the entire population. Since each treatment variable $z_i$ is binary, we can indicate an exact treatment assignment as a product of $z_i$ (for treated individuals) and $(1-z_i)$ (for untreated individuals) factors. As such, we can express a general potential outcome function $Y_i$ as a polynomial in $\bz$, 
\[
    Y_i(\bz) = \sum_{\cT \subseteq [n]} a_{i,\cT} \prod_{j \in \cT} z_j \prod_{k \in [n] \setminus \cT} (1-z_k),
\] 
where $a_{i,\cT}$ is individual $i$'s outcome when their set of treated neighbors is exactly $\cT$.
Via a change of basis, we can equivalently express $Y_i(\bz)$ as a polynomial in the ``treated subsets''
\begin{equation} \label{eq:pom_general}
    Y_i(\bz) = \sum_{\cS' \subseteq [n]} c_{i,\cS'} \prod_{j' \in \cS'} z_{j'},
\end{equation}
where $c_{i,\cS'}$ represents the additive effect on individual $i$'s outcome that they receive when the entirety of subset $\cS'$ (perhaps among other individuals) is treated. Note that $c_{i,\varnothing}$ represents the \textit{baseline effect}, the component of $i$'s outcome that is independent of the treatment assignments.

So far, the potential outcomes model described in (\ref{eq:pom_general}) is completely general. However, it is parameterized by $2^n$ coefficients $\{c_{i,\cS'} \}$, which makes it untenable in most settings. To combat this, we impose some structural assumptions on these coefficients. First, we observe that the populations of interest can be quite large (e.g. the population of an entire country), and their influence networks may have high diameter. Throughout most of the paper, we assume that individuals' outcomes are influenced only by their immediate in-neighborhood. 

\begin{assumption}[Neighborhood Interference]
\label{assp:neighborhood_interference}
    $Y_i(\bz)$ only depends on the treatment of individuals in $\cN_i$. Equivalently, $Y_i(\bz) = Y_i(\bz')$ for any $\bz$ and $\bz'$ such that $z_j = z_j'$ for all $j \in \cN_i$. In our notation $c_{i,\cS'} = 0$ for any $\cS' \not\subseteq \cN_i$. 
\end{assumption}

Next, we note that the degree of each $Y_i(\bz)$ can (under the neighborhood interference assumption) be as large as $\din$. In such a model, one's outcome may be differently influenced by a treated coalition of any size in their neighborhood. Contrast this with a simpler linear potential outcomes model, wherein an individual's outcome receives only an independent additive effect from each of their treated neighbors. This illustrates that the degree of the polynomial may serve as a proxy for its complexity. In this work we consider the scenario where the polynomial degree may be significantly smaller than $\din$.

\begin{assumption}[Low Polynomial Degree]
\label{assp:low_degree}
    Each potential outcome function $Y_i(\bz)$ has degree at most $\beta$. In our notation, $c_{i,\cS'} = 0$ whenever $|\cS'| > \beta$. Along with \cref{assp:neighborhood_interference}, it follows that the potential outcomes function $Y_i(\bz)$ from (\ref{eq:pom_general}) can be expressed in the form,
    \begin{equation} \label{eq:pom_ld_neighbor}
        Y_i(\bz) = \sum_{\meedit{\cS' \in \cS_i^\beta}} c_{i,\cS'} \prod_{j \in \cS'} z_j,
    \end{equation}
    \meedit{where we define $\cS_i^\beta := \{ \cS \subseteq \cN_i \colon |\cS| \leq \beta \}$.}
\end{assumption}

We remark that while we use the formal mathematical term of ``low polynomial degree'', since this describes a function over a vector of binary variables, a low polynomial degree constraint is equivalent to a constraint on the \textit{order of interactions} amongst the treatments of neighbors. In the simplest setting when $\beta=1$, this is equivalent to a model in which the networks effects are additive across treated neighbors, strictly generalizing beyond all linear models that have been widely used in applied settings.

We use the notation in \eqref{eq:pom_ld_neighbor} to express the potential outcomes model for the remainder of the paper. 
If $\beta \geq \din$, note that $\beta$ could be completely removed from the definition of $Y_i$ in Equation \ref{eq:pom_ld_neighbor}, reducing to the arbitrary neighborhood interference setting. However, we turn our attention to settings where $\beta$ might be much smaller than the degree of the graph ($\beta \ll \din$) and we can assume that only low-order interactions within neighborhoods have an effect on an individual. As noted above, taking $\beta = 1$ corresponds to the heterogeneous linear outcomes model in \citep{YuAiroldiBorgsChayes22}. We include further examples to help in understanding this low polynomial degree assumption in Section~\ref{subsec:examples}.

Many of our variance bounds utilize an upper bound on the treatment effects for each individual. We define $Y_{\max}$ such that
\begin{align}\label{eq:Y_max}
    Y_{\max} := \max_{i \in [n]} \sum_{\cS' \subseteq \cS_i^\beta} |c_{i,\cS'}|.
\end{align}
It follows that $|Y_i(\bz)| \leq Y_{\max}$ for any treatment vector $\bz$. 

\begin{remark}
    We emphasize that the model in \eqref{eq:pom_ld_neighbor} captures fully general neighborhood interference when $\beta = d$. 
    \mcreplace{When $\beta < d$, our model may be viewed as \textit{semiparametric}. We note that the number of unknown parameters in this model is $\sum_{i \in [n]} \sum_{k = 0}^{\beta} \binom{|\cN_i|}{k}$, which scales as $n \din^{\beta}$. Although this quantity is finite for a fixed $n$, since we allow these coefficients to vary for each $i$, we can view the coefficients $c_{i,\cS}$ as being generated by an unknown family of functions $f_i:2^{\cN_i}\to\Reals$ in some function space $F$ such that for any subset $\cS\subseteq\cN_i$ we have $f_i(\cS)=c_{i,\cS}$. Without imposing assumptions on $F$, we can see that to specify the coefficients requires specifying $\{f_i\}_{i \in [n]}$, which belongs to an infinite-dimensional space of functions. On the other hand, if each $f_i$ was known, then for fixed $n$ the parameter space of the model is indeed finite. As our model consists of both finite and infinite components, we may view it as semiparametric.}{Even when $\beta < d$, the number of parameters in the model grows with $n$. This results in the total number of observations always being smaller than the number of parameters in the model, so using simple regression-style estimators to identify the model is impossible.}
\end{remark}

\begin{remark}
    We can consider the order of interactions or polynomial degree $\beta$ as a way to measure the complexity of the model class. Our subsequent results suggest that $\beta$ is meaningful as it also captures a notion of statistical complexity or difficulty of estimation. This is evidenced by the variance upper bound and minimax lower bound results in \Cref{sec:properties}, where we see that the smaller $\beta$ is relative to the degree of the graph, the smaller the variance incurred and the larger $\beta$ is with respect to the graph degree, the higher the variance incurred. In some sense, the ``lower'' the degree of the model, i.e. the more structure imposed, the ``easier'' it is to estimate the total treatment effect. On the flip side, the ``higher'' the degree of the model, i.e. the closer it is to being fully general, the ``harder'' it is to estimate.
\end{remark}


\begin{remark}
\cyreplace{Via a change of basis, we write potential outcomes models as linear combinations of treated subsets, i.e. an individual's outcome under neighborhood interference is a linear combination of fully treated subsets of their in-neighborhood. Function sparsity in this context means that there are few subsets of an individual's neighborhood that have an effect on the individual's outcome when treated. This is a natural way to interpret ``low order interactions'' structure: if $\beta$ is small relative to the degree $d$ of the graph, then the potential outcomes function is sparse in this basis.}{
Assumption \ref{assp:low_degree} implies that the model exhibits a particular type of sparsity in the coefficients with respect to the monomial basis, in which the coefficients corresponding to sets larger than size $\beta$ are zero. In our setting when the treatments are budgeted, i.e. $p$ is small, these coefficients precisely correspond to effects that are observable in a Bernoulli randomized design, i.e. the probability for observing or measuring the coefficient corresponding to set $\cS$ is $p^{|\cS|}$. As such, there is an direct connection between this hierarchy of models as parameterized by $\beta$ and the ability to measure and estimate the total treatment effect, which is formalized by our subsequent analysis.
}
\end{remark}


\subsubsection{Examples of Low Degree Interaction Models}\label{subsec:examples}

We provide a few examples to illustrate when the polynomial degree  of the potential outcomes model may naturally be smaller than the total neighborhood size (i.e. $\beta < d$). 

\begin{example}
Consider a potential outcomes model that exhibits the joint assumptions of additivity of main effects and interference effects as defined in \cite{SussmanAiroldi17}. This imposes that the potential outcomes satisfy
\[
    Y_i(\bz) = Y_i(\mathbf{0}) + \big(Y_i(z_i \, \mathbf{e}_i) - Y_i(\mathbf{0})\big) + \sum_{k\in [n]} \big(Y_i(z_k \,\mathbf{e}_k ) - Y_i(\mathbf{0}) \big),
\]
where $\mathbf{0} \in \mathbb{R}^n$ is a vector of all zeros and $\mathbf{e}_j\in\{0,1\}^n$ is a standard basis vector.
As discussed in \cite{YuAiroldiBorgsChayes22}, this assumption implies a low degree interaction model with $\beta = 1$, which they refer to as heterogeneous additive network effects.
\end{example}

\begin{example}
Consider a hypothesized setting where each individual's neighborhood can be divided into smaller \textit{sub-communities}. For example, one's close contacts may include their immediate family, their close friends, their work colleagues, etc. Within each of these sub-communities, there may be non-trivial (higher order) interactions between the treatments of its members. However, it is reasonable to assume that the cumulative effects of each sub-community have an additive effect on the individual's outcome. That is, there are no non-trivial interactions resulting from the treatment of individuals across different sub-communities. In this case, a natural choice for $\beta$ is the size of the largest sub-community, which could be significantly smaller than the size of the largest neighborhood.
\end{example}

\begin{example}
Suppose a social networking platform is testing a new ``hangout room'' feature where groups of up to five people can interact in a novel environment on the platform. One can posit that a natural choice for $\beta$ in this setting is $5$, as the change in any individual's usage on the platform can be attributed to how they utilize this new feature, which in turn is a function of it being introduced to various subsets of up to $5$ users in that individual's neighborhood.
\end{example}

\begin{example}
Consider a setting where an individual's outcome is a low degree polynomial in some auxiliary quantity which is itself linear in the treatment assignments of their neighborhood. For example, we might have
\[
    Y_i(\bz)=c_{i,\emptyset} + \sum_{j\in\cN_i}c_{ij}z_j + \Big(\sum_{j\in\cN_i}c_{ij}z_j\Big)^2 + \cdots + \Big(\sum_{j\in\cN_i}c_{ij}z_j\Big)^\beta.
\]
A similar setting is explored in our simulated experiments in Section~\ref{sec:experiments}.
\end{example}

\begin{example}
Consider an example in which network effects only arise from pairwise edge interactions and triangular interactions, i.e. individual $i$'s outcome consists of a sum of its baseline outcome, it's own direct effect $c_{i,i}$, pairwise edge effects $c_{i,j}$ for $j \in \cN_i$, and triangle effects $c_{i,\{j,j'\}}$ for $j,j' \in \cN_i$ such that there is also an edge between $j$ and $j'$, indicating that the three individuals are mutual connections. For such a model, $\beta$ would be 2 due to the triangular interactions.
\end{example}


\begin{remark}
We emphasize that any potential outcomes model that takes a binary treatment vector $\bz$ as its input can be written as a polynomial in $\bz$, taking the form in Equation \eqref{eq:pom_ld_neighbor} for general $\beta$. However, the low degree assumption, i.e. that $\beta \ll d$, will not generally admit high-degree models. For example, threshold models and saturation models both generally require the degree of $Y_i(\bz)$ to be $|\cN_i|$. In a threshold potential outcomes model, an individual experiences network effects once a particular threshold of their neighbors are treated \citep{GuiXuBhasinHan15}. An example of this type of model is given by
\[
    Y_i(\bz) = \Ind\left(\textstyle{\sum_{j\in \cN_i}z_j \geq \theta}\right)\cdot \sum_{\cS\subseteq \cN_i}c_{i,\cS}\prod_{j \in \cS}z_j,
\]
where $0 \leq \theta \leq|\cN_i|.$ Saturation models allow for network or peer effects to increase up to a particular saturation level, such as
\[
    Y_i(\bz) = \min \Big(\theta \;,\sum_{\cS\subseteq \cN_i}c_{i,\cS}\prod_{j \in \cS}z_j\Big),
\]
where $\theta$ is some maximum saturation threshold. In this model, an individual receives additive effects from each subset of treated neighbors until a certain threshold effect is reached, after which the networks effects have ``saturated'' and the treatment of additional neighbors contributes no additional effect.
\end{remark}

\subsection{Causal Estimand and Randomized Design}

Throughout most of the paper, we concern ourselves with estimating the total treatment effect ($\TTE$). This quantifies the difference between the average of individual's outcomes when the entire population is treated versus the average of individual's outcomes when the entire population is untreated:
\begin{equation} \label{eq:TTE_Ydefn}
    \TTE := \frac{1}{n} \sum_{i=1}^n \big( Y_i({\bf 1}) - Y_i({\bf 0}) \big),
\end{equation}
where $\mathbf{1}$ represents the all ones vector and $\mathbf{0}$ represents the zero vector. Plugging in our parameterization from equation (\ref{eq:pom_ld_neighbor}), we may re-express the total treatment effect as
\begin{equation} \label{eq:TTE_cdefn}
    \TTE = \frac{1}{n}\sum_{i=1}^{n} \sum_{\substack{\cS' \in \cS_i^\beta \\ \cS' \ne \varnothing}} c_{i,\cS'}.
\end{equation}

Since exposing individuals to treatment can have a deleterious and irreversible effect on their outcomes, we wish to estimate the total treatment effect after treating a small random subset of the population. We focus on a \textit{non-uniform} Bernoulli design, wherein each individual $i$ is independently assigned treatment with probability $p_i \in [p, 1-p]$ for $p > 0$. That is, each $z_i \sim \textrm{Bernoulli}(p_i)$. Such a randomized design is both straightforward to implement and to understand. Furthermore, many existing experimentation platforms are already designed for Bernoulli randomization, making it easy to collect new data or to re-analyze existing data and adjust for network interference, rather than requiring a complete overhaul of the existing experimentation platform to allow for more complex randomization schemes.

\section{Estimator} \label{sec:estimator}


    \mereplace{In this section, we introduce the class of estimators that we study, which are a special case of the \textit{pseudoinverse estimator} first introduced by Swaminathan et. al.~\citep{swaminathan2017off}. This class is fairly robust; it can be used for a wide variety of causal estimands and experimental designs, which we discuss more in Section \#. We focus on estimating the total treatment effect. We also limit our discussion in this section to non-uniform Bernoulli design, both because it is among the most widely-used and accessible designs and because it admits clean calculations.}
    {In this section, we introduce the estimator that will serve as our central focus through the rest of the paper: the Structured Neighborhood Interference Polynomial Estimator (SNIPE). While we restrict our attention to non-uniform Bernoulli design, the techniques to derive this estimator can be generalized to a wide variety of causal estimands and experimental designs. In fact, our estimator is connected to the Horvitz-Thompson estimator and turns out to be a special case of both the \textit{pseudoinverse estimator} first introduced by Swaminathan et. al.~\citep{swaminathan2017off} and more recently the \textit{Riesz estimator} of Harshaw et. al.~\citep{harshaw2022design}. We discuss these connections in~\Cref{ssec:est_connections}.} 
    
    To provide intuition, we first derive the estimator in the $\beta = 1$ case with the linear heterogenous potential outcomes model \citep{YuAiroldiBorgsChayes22} via a connection to ordinary least squares. Then, we show how it can be extended to the more general polynomial setting. Our main result (Section~\ref{sec:properties}) establishes both the unbiasedness of this family of estimators and a bound on its variance under Bernoulli design.
    
    \meedit{Recalling that $\cS_i^\beta := \{ \cS \subseteq \cN_i \colon |\cS| \leq \beta \}$, the vector $\bc_i$ collects the parameters $\{ c_{i,\cS}\}_{\cS \in \cS_i^\beta}$ in some canonical ordering. We'll assume throughout that $\varnothing \in \cS_i^\beta$ is always first in this ordering and otherwise index the entries of these vectors by their corresponding set $\cS$. As an example, when $\beta = 1$ and $\cN_i = \{ j_1, \hdots, j_{d_i} \}$, we may have $\bc_i = \big[ c_{i,\emptyset} \quad c_{i,\{j_1\}} \quad \hdots \quad c_{i,\{j_{d_i}\}} \big]^\intercal$, where $d_i$ is the in-degree of unit $i$. In a similar manner, the \textit{treated subsets vector} $\tbz_i$ collects the indicators $\big\{ \prod_{j \in \cS} z_j \big\}_{\cS \in \cS_i^\beta}$ in the same ordering. In our $\beta = 1$ example, $\tbz_i = \big[ 1 \quad z_{j_1} \quad \hdots \quad z_{j_{d_i}} \big]^\intercal$. Using this notation, we may express
    \begin{equation} \label{eq:tte_as_ip}
        \bY_i(\bz) = \Big\langle \bc_i, \tbz_i \Big\rangle, 
        \hspace{60pt}
        \TTE = \tfrac{1}{n} \sum_{i=1}^{n} \Big\langle (\mathbf{1}_{|S_i|} - \mathbf{e}_1), \bc_i \Big\rangle,
    \end{equation}
    where the inner product argument in the second equation is the $|S_i|$-length vector with first entry 0 and remaining entries 1.
    }
    
\subsection{Building Intuition in the Linear Setting \texorpdfstring{($\boldsymbol\beta$ = 1)}{(beta = 1)}} \label{subsec:pi_linear}
    
    To motivate our estimator, we consider the linear heterogeneous potential outcomes model ($\beta = 1$) under non-uniform Bernoulli randomized design. \meedit{Using the TTE expression from~\eqref{eq:tte_as_ip}, we can recast the problem of estimating the total treatment effect into the problem of estimating the parameter vector $\bc_i$: by linearity of expectation an unbiased estimator of $\bc_i$ will give rise to an unbiased estimator of $\TTE$. 
    
    As a thought experiment toward estimating $\bc_i$, imagine that we can perform $M$ independent replications of our randomized experiment. In each replication $m \in [M]$, we observe the treated subsets vector $(\tbz_i)^m$, obtained from our Bernoulli randomized design, and the outcome $(Y_i)^m$. We visualize,}
    \[
        \underbrace{\begin{bmatrix}
            (Y_i)^1 \\ (Y_i)^2 \\ \vdots \\ (Y_i)^M
        \end{bmatrix}}_{\bY_i \in \mathbb{R}^{M}}
        = 
        \underbrace{\begin{bmatrix}
            1 & z_{j_1}^{1} & \hdots & z_{j_{d_i}}^{1} \\
            1 & z_{j_1}^{2} & \hdots & z_{j_{d_i}}^{2} \\
            \vdots & \vdots & \ddots & \vdots \\
            1 & z_{j_1}^{M} & \hdots & z_{j_{d_i}}^{M} \\
        \end{bmatrix}}_{\bX_i \in \mathbb{R}^{M \times (d_i+1)}}
        \underbrace{\begin{bmatrix}
            c_{i,\varnothing} \\ c_{ij_1} \\ \vdots \\ c_{ij_{d_i}}
        \end{bmatrix}}_{\bc_i \in \mathbb{R}^{(d_i+1)}}.
    \]
    To minimize the sum of squared deviations from the true coefficients, we use ordinary least squares, computing
    \[
        \hat{\bc}_i = \big( \bX_i^\intercal \bX_i \big)^{-1} \bX_i^\intercal \bY_i. 
    \]
    \meedit{We consider the limit of this estimate as $M \to \infty$. The consistency of the least squares estimator ensures that $\hat{\bc}_i \overset{P}{\longrightarrow} \bc_i$. By the law of large numbers, we have
    \[
        \bX_i^\intercal \bX_i = \sum_{m=1}^{M}
        \begin{bmatrix}
            1 & z_{j_1}^m & z_{j_2}^m & \hdots & z_{j_{d_i}}^m \\ 
            z_{j_1}^m & \big(z_{j_1}^m\big)^2 & z_{j_1}^m z_{j_2}^m & \hdots & z_{j_1}^m z_{j_{d_i}}^m \\
            z_{j_2}^m & z_{j_1}^m z_{j_2}^m & \big(z_{j_2}^m\big)^2 & \ddots & z_{j_2}^m z_{j_{d_i}}^m \\
            \vdots & \ddots & \ddots & \ddots & \vdots \\
            z_{j_{d_i}}^m & z_{j_1}^m z_{j_{d_i}}^m & z_{j_2}^m z_{j_{d_i}}^m & \hdots & \big(z_{j_{d_i}}^m\big)^2
        \end{bmatrix}
        \overset{a.s.}{\longrightarrow} M \cdot \mathbb{E} \big[ \tbz_i\tbz_i^\intercal  \big].
    \]
    Similarly, the vector $\bX_i^\intercal \bY_i \overset{a.s.}{\longrightarrow} M \cdot \mathbb{E} \big[ \tbz_i \cdot Y_i(\bz) \big]$. Assuming the invertibility of the matrix $\mathbb{E} \big[ \tbz_i\tbz_i^\intercal  \big]$ (we show this explicitly for Bernoulli design below), we obtain in the limit,}
    \[
        \bc_i = \big( M \cdot \mathbb{E} \big[ \tbz_i\tbz_i^\intercal \big] \big)^{-1} M \cdot \mathbb{E} \big[ \tbz_i \cdot Y_i(\bz) \big] 
        = \mathbb{E} \big[ \tbz_i\tbz_i^\intercal \big]^{-1} \mathbb{E} \big[ \tbz_i \cdot Y_i(\bz) \big].
    \]    
    Now, we return from our thought experiment to the actual experimental setting wherein a single instantiation $(\tbz_i, Y_i)$ is realized. We consider the estimator
    \begin{equation} \label{eq:chat_est}
        \widehat{\bc_i} := \mathbb{E} \big[ \tbz_i\tbz_i^\intercal \big]^{-1} \; \tbz_i \; Y_i.
    \end{equation}
    Note that the matrix $\mathbb{E} \big[ \tbz_i\tbz_i^\intercal \big]^{-1}$ depends only on our experimental design and thus can be utilized by our estimator. The unbiasedness of this estimator follows from our above computation; applying linearity, 
    \[
        \mathbb{E} \big[ \widehat{\bc_i} \big] = \mathbb{E} \big[ \tbz_i\tbz_i^\intercal \big]^{-1} \; \mathbb{E} \big[ \tbz_i \; Y_i \big] = \bc_i.
    \]
    Applying linearity once more, we may obtain the unbiased estimator for the total treatment effect,
    \begin{equation}
    \label{eqn:TTE_hat implicit}
        \widehat{\TTE} 
        = \tfrac{1}{n} \sum_{i=1}^{n} \big\langle (\mathbf{1}_{|\cS_i^\beta|} - \mathbf{e}_1) , \widehat{\bc_i} \big\rangle
        = \tfrac{1}{n} \sum_{i=1}^{n} Y_i \Big\langle \mathbb{E} \big[ \tbz_i \tbz_i^\intercal \big]^{-1} (\mathbf{1}_{|\cS_i^\beta|} - \mathbf{e}_1),  \tbz_i \Big\rangle.
    \end{equation}
    
    For the specific setting of non-uniform Bernoulli design with $\beta = 1$, one can use the fact that $\mathbb{E}[z_{j_k}] = \mathbb{E}[z_{j_k}^2] = p_{j_k}$ and $\mathbb{E}[z_{j_k} z_{j_{k'}}] = p_{j_k} p_{j_{k'}}$ for $k \ne k'$ \meedit{to compute,
    \[
        \mathbb{E} \big[ \tbz_i\tbz_i^\intercal  \big] = 
        \begin{bmatrix}
            1 & p_{j_1} & p_{j_2} & \hdots & p_{j_{d_i}} \\ 
            p_{j_1} & p_{j_1} & p_{j_1}p_{j_2} & \hdots & p_{j_1}p_{j_{d_i}} \\
            p_{j_2} & p_{j_1}p_{j_2} & p_{j_2} & \ddots & p_{j_2}p_{j_{d_i}} \\
            \vdots & \ddots & \ddots & \ddots & \vdots \\
            p_{j_{d_i}} & p_{j_1}p_{j_{d_i}} & p_{j_2}p_{j_{d_i}} & \hdots & p_{j_{d_i}}
        \end{bmatrix},
    \]
    which we can invert to obtain
    \[
        \E\Big[\tbz_i\tbz_i^\intercal\Big]^{-1} = 
        \begin{bmatrix}
            1 + \sum_k \tfrac{p_{j_k}}{1-p_{j_k}} & -\tfrac{1}{1-p_{j_1}} & \cdots & \cdots & -\tfrac{1}{1-p_{j_{d_i}}} \\
            -\tfrac{1}{1-p_{j_1}} & \tfrac{1}{p_{j_1}(1-p_{j_1})} & 0 & \cdots & 0 \\
            \vdots & 0 & \tfrac{1}{p_{j_2}(1-p_{j_2})} & 0 & \vdots \\
            \vdots & \vdots & 0 & \ddots & 0 \\
            -\tfrac{1}{1-p_{j_{d_i}}} & 0 & \cdots & 0 & \tfrac{1}{p_{j_{d_i}}(1-p_{j_{d_i}})}
        \end{bmatrix}.
    \]
    We calculate
    \[
        \mathbb{E} \big[ \tbz_i \tbz_i^\intercal \big]^{-1} (\mathbf{1}_{|\cS_i^\beta|} - \mathbf{e}_1) = \begin{bmatrix}
            - \sum\limits_{j \in \cN_i} \tfrac{1}{1-p_j} & \tfrac{1}{p_{j_1}(1-p_{j_1})} & \dots & \tfrac{1}{p_{j_{d_i}}(1-p_{j_{d_i}})} 
        \end{bmatrix}^\intercal.
    \]
    Plugging into equation~\eqref{eqn:TTE_hat implicit}, we obtain the explicit form for our estimator:
    \[
        \widehat{\TTE}_{\text{SNIPE}(1)} := \tfrac{1}{n} \sum_{i=1}^{n} Y_i \sum_{j \in \cN_i} \frac{z_j-p_j}{p_j(1-p_j)}\mcreplace{.}{,}
    \]  
    where SNIPE$(1)$ refers to the fact that this is the SNIPE estimator with $\beta=1$.} 

\subsection{SNIPE in the General Polynomial Setting}
\label{subsec:pi_polynomial}

\mereplace{
The key intuition for the estimator introduced in the previous subsection is that the potential outcomes $Y_i(\bz)$ are linear functionals of the covariate vectors $\tilde{\bz}_i$. In the general polynomial setting, we can still express the potential outcomes as a linear function of an augmented covariate vector, $\tilde{\bz}_i$, which is a binary vector of length $\sum_{k=0}^{\min(\beta,|\cN_i|)} \binom{|\cN_i|}{k}$, i.e. the number of subsets of $|\cN_i|$ of size at most $\beta$. For simplicity of notation, we index elements in the vector $\tilde{\bz}_i$ with the corresponding subset of $\cN_i$. The element in $\tilde{\bz}_i$ corresponding to subset $\cS \subseteq \cN_i$ takes the value $\prod_{j \in \cS} z_j$, i.e. indicating if $\cS$ is fully treated. By construction, we can verify that $Y_i(\bz)$ is a linear functional of $\tilde{\bz}_i$ with coefficients given by $\{c_{i,\cS}\}$ for $\cS \subseteq \cN_i$ such that $|\cS| \leq \beta$.

As long the matrix of expectations $\mathbb{E} \big[ \tilde{\bz}_i \tilde{\bz}_i^\intercal \big]$ is invertible, the pseudoinverse estimator introduced in the previous subsection is well-defined. Then, we could invert $\mathbb{E} \big[ \tilde{\bz}_i\tilde{\bz}_i^\intercal \big]^{-1}$ and plug it into a similar formula as \eqref{eqn:TTE_hat implicit} to derive the estimator. The results in \citep{swaminathan2017off} suggest that if the randomized design is such that every possible treatment vector has nonzero mass, the matrix $\E[\tbz_i\tbz_i^\intercal]$ will likely be invertible. On the other hand, as the naming of this estimator class suggests, if the matrix $\E[\tbz_i\tbz_i^\intercal]$ is not invertible, we can replace the true inverse in the estimator with the Moore-Penrose pseudoinverse, $\E[\tbz_i\tbz_i^\intercal]^\dagger$.

\begin{remark} \label{rem:pseudoinv}
Consider an estimand that is a linear function of the model parameters as expressed by $\frac{1}{n} \sum_{i=1}^n \langle \bt_i, \bc_i \rangle$ for some vectors $\bt_i$, e.g. for the total treatment effect $\bt_i$ is the all-ones vector except for the first element taking value zero. In this case, the pseudoinverse estimator will be unbiased as long as the vector $\bt_i$ belong to the column space of $\mathbb{E} \big[ \tilde{\bz}_i \tilde{\bz}_i^\intercal \big]$ for each $i$, such that $\mathbb{E} \big[ \tilde{\bz}_i \tilde{\bz}_i^\intercal \big] \mathbb{E} \big[ \tilde{\bz}_i \tilde{\bz}_i^\intercal \big]^{\dagger} \bt_i = \bt_i$.
\end{remark}

As it is algebraically tedious to show that the matrix under Bernoulli design in the general $\beta$ setting is invertible, we instead derive the estimator by solving the unbiasedness conditions directly using an application of M\"obius Inversion, as detailed in \cref{sec:main_proof}. The unbiasedness argument also proves that the vector $[0,1,1,1 \dots 1]$ lies in the column space of $\E[\tilde{\bz}_i \tilde{\bz}_i^\intercal]$, which is also sufficient for the unbiasedness of the pseudoinverse estimator. An explicit proof of this is in \cref{subsec:colspace}. The calculations result in the following estimator 

}
{

For larger $\beta$, we can use the same least squares approach to obtain an unbiased estimate of the coefficient vector $\bc_i$, and then plug this into equation~\eqref{eq:tte_as_ip} to obtain $\widehat{\TTE}$. The outcomes $Y_i(\bz)$ remain linear functions in $\tilde{\bz}_i$ --- albeit for a significantly longer $\tilde{\bz}_i$ containing $|\cS_i^\beta| = \sum_{k=0}^{\min(\beta,|\cN_i|)} \binom{|\cN_i|}{k}$ entries --- so the same convergence results apply. Suppose we index the entries of $\mathbb{E} \big[ \tbz_i \tbz_i^\intercal \big]$ by the sets $\cS$ and $\cT$ corresponding to the matrix row and column. By the independence and marginal treatment probabilities of non-uniform Bernoulli design, we see that
\[
    \Big(\mathbb{E} \big[ \tbz_i \tbz_i^\intercal \big]\Big)_{\cS,\cT} = \prod_{j \in \cS \cup \cT} p_j.
\]
Working with this matrix is significantly more tedious, so we relegate the details to Appendix~\ref{sec:explicit_est}. There, we show that $\mathbb{E} \big[ \tbz_i \tbz_i^\intercal \big]$ is invertible by giving an explicit formula for the entries of its inverse. Then, we plug these entries into equation~\eqref{eqn:TTE_hat implicit} to derive the following explicit formula for the estimator:
}

\begin{equation} \label{eqn:TTE_PI}
    \widehat{\TTE}\mcedit{_{\text{SNIPE}(\beta)}}
    := \frac{1}{n} \sum_{i=1}^n Y_i \sum_{\cS \in \cS_i^\beta} g(\cS) \prod_{j \in \cS} \tfrac{z_j - p_j}{p_j(1-p_j)},
\end{equation}
where we define the coefficient function $g: 2^{[n]} \to \Reals$ such that
\begin{equation*}
    g(\cS) = \prod_{s \in \cS} (1-p_s) - \prod_{s \in \cS} (-p_s)
\end{equation*}
for each $\cS \subseteq [n]$ and $g(\emptyset)=0$. \mcedit{When it is clear from context that $ \widehat{\TTE}$ refers to the SNIPE estimator, we suppress the subscript for cleaner notation.}

\meedit{
    \begin{remark}
        We pause here to again emphasize that our technique gives us an unbiased estimate of the \textit{coefficient vector} $\bc_i$ for each individual $i$. From here, we leverage the linearity of expectation to obtain an unbiased estimator for the $\TTE$, which is a linear function of these $\bc_i$ coefficients. This same strategy can be applied to develop estimators for any causal effect that is linear in the $\bc_{i,\cS}$ coefficients. We highlight some other potential estimands and the explicit form of their estimators in \Cref{sec:proof_other_est}. The techniques that we discuss in \Cref{sec:properties} can be used to establish further properties of these estimators. 
    \end{remark}
}

This estimator can be evaluated in $O(n \din^{\beta})$ time and only utilizes structural information about the graph (not any influence coefficients $c_{i,\cS}$). Structurally, the estimator takes the form of a weighted average of the outcomes $Y_i$ of each individual $i$, where the weights themselves are functions of the treatment assignments of all members $j$ of the in-neighborhood $\cN_i$. To make use of the low-order interference assumption, the estimator separately scales the effect of treatment of each sufficiently small subset of $\cN_i$ using the scaling function $g(\cS)$. The definition of this $g(\cS)$ ensures the unbiasedness of the estimator.

In the special case of a uniform treatment probability $p_i = p$ across all nodes, we can simplify this estimator to show that it is only a function of the number of treated individuals in $i$'s neighborhood and not the identities. Let $\cT = \{j: z_j = 1\}$ denote the set of treated units. We can rewrite the estimator as follows,
\begin{align} \label{est:uniform_prob}
\widehat{\TTE} &= \tfrac{1}{n} \sum_{i=1}^n Y_i \sum_{\cS \in \cS_i^\beta} \Big( \prod_{j \in \cS} \tfrac{z_j-p}{p} - \prod_{j \in \cS} \tfrac{z_j-p}{p-1} \Big) \nn \\
&= \tfrac{1}{n} \sum_{i=1}^n Y_i \sum_{k=0}^{\min(\beta,|\cN_i\cap\cT|)} \sum_{\substack{\cA \subseteq \cN_i\cap\cT \\ |\cA| = k}} \sum_{\substack{\cB \subseteq \cN_i\setminus\cT \\ |\cB| \leq \beta - k}} \left( \left(\tfrac{1-p}{p}\right)^{k} (-1)^{|\cB|} - (-1)^{k} \left(\tfrac{p}{1-p}\right)^{|\cB|}\right)  \nn \\
&= \tfrac{1}{n} \sum_{i=1}^n Y_i \sum_{k=0}^{\min(\beta,|\cN_i\cap\cT|)} \binom{|\cN_i\!\cap\!\cT|}{k} \sum_{\ell=0}^{\min(\beta-k,|\cN_i\setminus\cT|)} \binom{|\cN_i\!\!\setminus\!\! \cT|}{\ell} (-1)^{k+\ell} \left( \left(\tfrac{p-1}{p}\right)^{k} - \left(\tfrac{-p}{1-p}\right)^{\ell} \right).
\end{align}
Thus, while our estimation guarantees allow for heterogeneity in the potential outcomes model, when treatment probabilities are uniform, computing the estimator does not depend on the identity of the treated individuals, but only the number of treated neighbors. As a result, the estimator can be evaluated in $O(n\beta^2)$ time, which is a significant improvement compared to the $O(n\din^\beta)$ computational complexity when the treatment probabilities are nonuniform.

\meedit{
\subsection{Connection to Other Estimator Classes} \label{ssec:est_connections}

Our estimator takes the form of a linear weighted estimator
\[
    \widehat{\TTE} = \tfrac{1}{n} \sum_{i=1}^n Y_i \cdot w_i(\bz),
\]
under specifically constructed weight functions $w_i \colon \{0,1\}^n \to \mathbb{R}$. From this perspective, we can draw connections between our estimator and others appearing in the literature. 

\paragraph*{Horvitz-Thompson Estimator}

First, we show that in the special case where $|\cN_i| \leq \beta$, our estimator is identical to the classical Horvitz-Thompson estimator. In this case, the restriction $|\cS| \leq \beta$ is satisfied for every $\cS \subseteq \cN_i$, so that $\cS_i^\beta$ includes all subsets of $\cN_i$. Using this, we may simplify
\begin{align*}
    w_i(\bz) &= \sum_{\cS \subseteq \cN_i} g(\cS) \prod_{j \in \cS} \tfrac{z_j - p_j}{p_j(1-p_j)} \\
    &= \sum_{\cS \subseteq \cN_i} \prod_{j \in \cS} \tfrac{z_j - p_j}{p_j} - \sum_{\cS \subseteq \cN_i} \prod_{j \in \cS} \tfrac{-(z_j - p_j)}{(1-p_j)} \tag{by definition of $g(\cS)$}\\
    &= \prod_{j \in \cN_i} \Big(1 + \tfrac{z_j - p_j}{p_j} \Big) - \prod_{j \in \cN_i} \Big(1 - \tfrac{z_j - p_j}{1-p_j} \Big) \tag{distributivity} \\
    &= \prod_{j \in \cN_i} \tfrac{z_j}{p_j} - \prod_{j \in \cN_i} \tfrac{1-z_j}{1-p_j} \\
    &= \frac{\Ind(\bz \textrm{ treats all of } \cN_i)}{\Pr(\bz \textrm{ treats all of } \cN_i)} - \frac{\Ind(\bz \textrm{ treats none of } \cN_i)}{\Pr(\bz \textrm{ treats none of } \cN_i)}.
\end{align*}
Thus,
\[
    \widehat{\TTE}_\text{SNIPE} 
    = \tfrac{1}{n} \sum_{i=1}^n Y_i \left(\frac{\Ind(\bz \textrm{ treats all of } \cN_i)}{\Pr(\bz \textrm{ treats all of } \cN_i)} - \frac{\Ind(\bz \textrm{ treats none of } \cN_i)}{\Pr(\bz \textrm{ treats none of } \cN_i)}\right)
    = \widehat{\TTE}_{HT},
\]
so we exactly recover the Horvitz-Thompson estimator. As a result, when $\beta$ is sufficiently large relative to the degree of the nodes in the graph, our estimator is very similar to the Horvitz-Thompson estimator, only differing for the nodes which have graph degrees larger than $\beta$. In this sense, under Bernoulli randomization, our estimator can be thought of as a generalization of Horvitz-Thompson to additionally account for low polynomial degree structure, which is most relevant for simplifying the potential outcomes associated with high-degree vertices.

\paragraph*{Pseudoinverse Estimator}

\cyedit{The key technical steps of deriving our estimator can be described as constructing unbiased estimators for each unit $i$'s contribution to the total treatment effect using a connection to ordinary least squares for linear models, and subsequently averaging the unbiased estimates due to the fact that the causal estimand is a linear function of these individual contributions. This overall technique has also appeared in previous literature in semiparametric estimation, with one clear example described by}
Swaminathan et. al.~\citep{swaminathan2017off} in a seemingly different context of off-policy evaluation for online recommendation systems. In their model, a context $x \in X$ arrives, at which point the principal selects a tuple (or \textit{slate}) of actions $\mathbf{s} = (s_1, \hdots, s_\ell)$ and observes a random reward $r$ based on the interaction between the context and the slate. The authors make a linearity assumption that posits that $V(x,\mathbf{s}) = \E_r [ r | x,\mathbf{s}] =  \mathbf{1}_s^\intercal \phi_x$, where $\mathbf{1}_s^\intercal$ indicates the choice of a particular action in each entry of the tuple, and $\phi_x$ is a context-specific reward weight vector associated to context $x$. In our setting, the context $x$ is an individual $i$, the reward weights $\phi_x$ are their effect coefficients $\bc_i$, and the slate indicator vector $\mathbf{1}_s^\intercal$ is our treated subsets vector $\tbz_i$.

A primary goal of \citep{swaminathan2017off} is to estimate $\phi_x$, which can be used to inform a good slate selection policy. The mean squared error of an estimate $\mathbf{w}$ for $r$ is $\E_{\mu} [ (\mathbf{1}_s^\intercal \mathbf{w} - r)^2 ]$, where $\mu$ encodes the random selection of the context, slate, and reward. In this framing, the least squares estimator 
\[
    \overline{\phi}_x = \Big( \E_\mu [ \mathbf{1}_s\mathbf{1}_s^\intercal | x] \Big)^\dagger \E_\mu[ r\mathbf{1}_s^\intercal | x] 
\]
is the minimizer of the MSE with the minimum norm. As the reward distribution is unknown, it is replaced with an empirical estimate of $r$ given past data. We can perform the substitutions as described in the previous paragraph, noting that Bernoulli design ensures that $\E[ \tbz_i \tbz_i^\intercal]$ is invertible (see Appendix~\ref{sec:explicit_est}), to recover our estimator $\widehat{\bc_i}$ from \eqref{eq:chat_est}.

\paragraph*{Riesz Estimator}

In their recent work, Harshaw et. al.~\citep{harshaw2022design} consider a very general causal inference framework wherein a treatment $\bz$ is drawn from an underlying experimental design distribution over an intervention set $\mathcal{Z}$. We observe an outcome $Y_i(\bz)$ for each unit $i \in [n]$, where we assume that $Y_i$ belongs to a \textit{model space} $\mathcal{M}_i$, which is a subspace of $\mathcal{Y}$ containing all measurable and square-integrable (with respect to the distribution over $\mathcal{Z}$) functions $\mathcal{Z} \to \mathbb{R}$. Each individual is additionally endowed with a linear effect functional $\theta_i \colon \mathcal{Y} \to \mathbb{R}$, and the goal is to estimate the \textit{average individual treatment effect} $\tau = \tfrac{1}{n} \sum_{i=1}^{n} \theta_i(Y_i)$.

In our setting, $\mathcal{Z} = \{0,1\}^n$ with the non-uniform Bernoulli design distribution (i.e. $\textrm{Pr}(z_i = 1) = p_i$, with $\{z_i\}$ mutually independent). $\mathcal{Y}$ consists of all linear functions over the basis $\{ \prod_{j \in \cS} z_j \colon \cS \subseteq [n], |\cS| \leq \beta \}$, and each $\mathcal{M}_i$ restricts to the linear functions over $\{ \prod_{j \in \cS} z_j \colon \cS \in \cS_i^\beta \}$. Finally, each unit $i$ has effect functional $\theta_i(Y_i) = Y_i(\mathbf{1}) - Y_i(\mathbf{0})$, so that $\tau = \TTE$.

Under two assumptions --- correct model specification (needed in our work as well) and positivity (always satisfied in our setting since Bernoulli design ensures each treatment in $\mathcal{Z}$ occurs with positive probability) --- the Riesz Representation theorem guarantees the existence of an unbiased estimator for $\tau$, referred to as a Riesz Estimator, of the form
\[
    \widehat{\tau} = \tfrac{1}{n} \sum_{i=1}^{n} Y_i(\bz) \:\psi_i(\bz),
\]
where $\psi_i \in \mathcal{M}_i$ is a Riesz representer for $\theta_i$ with the property that 
\begin{equation}
    \theta_i(Y) = \mathbb{E}_{\bz} \big[ \psi_i(\bz) Y(\bz) \big] \tag{$\star$}
\end{equation}
for each $Y \in \mathcal{M}_i$. As each model space $\mathcal{M}_i$ has dimension $|\cS_i^\beta|$, we can identify $\psi_i$ by solving the linear system that verifies that ($\star$) holds for each function in some basis for $\mathcal{M}_i$. A canonical choice is the standard basis, giving rise to equations 
\[
    \theta_i\Big( \prod_{j' \in \cS'} z_{j'} \Big) = \mathbb{I} \big( \cS' \ne \varnothing \big) = \mathbb{E}_{\bz} \Big[ \psi_i(\bz) \prod_{j' \in \cS'} z_{j'} \Big]. 
\]
for each $\cS' \in \cS_i^\beta$. The choice $\displaystyle \psi_i(\bz) = \sum_{\cS \in \cS_i^\beta} g(\cS) \prod_{j \in \cS} \tfrac{z_j - p_j}{p_j (1-p_j)}$ is a solution to this system (see the unbiasedness calculation in  Appendix~\ref{sec:main_proof}) and gives rise to our estimator. 

}
\section{Properties of Estimator under Bernoulli Design}
\label{sec:properties}

The following theorem summarizes the key properties of our estimator.

\begin{theorem} \label{thm:est_bias_var}
    Under a potential outcomes model satisfying the neighborhood interference assumption with polynomial degree at most $\beta$, the estimator defined in \eqref{eqn:TTE_PI} is unbiased with variance upper bounded by
    \[
        \frac{\din \; \dout \; {Y_{\max}}^2}{n} \cdot \Big( \tfrac{e\din}{\beta} \cdot \max\big(4\beta^2, \tfrac{1}{p(1-p)}\big) \Big)^{\beta}, 
    \]
    where each $p_i \in [p,1-p]$ and $p>0$.
\end{theorem}

Notably, a sequence of networks with $n \to \infty$ and $d = o(\log n)$ has variance asymptotically approaching 0. We defer the proof of this theorem to \cref{sec:main_proof}. 
Rather than appealing to the convergence properties of the pseudoinverse estimator to establish unbiasedness, we present an alternate combinatorial proof. To bound the variance, we carefully consider different possible overlapping subsets of individuals to separately bound many covariance terms that make up the overall variance expression.

In order to understand the variance bounds for our estimator, we can compare against the variance of Horvitz-Thompson under a Bernoulli design. In the simple setting of a $d$-regular graph and uniform Bernoulli$(p)$ randomization, \cite{UganderKarrerBackstromKleinberg13} showed that the Horvitz-Thompson estimator has a variance that is lower bounded by $\Omega(1/np^d)$. In contrast, the variance of our estimator only scales polynomially in the degree $d$, but exponentially in the polynomial degree $\beta$, which is achieved by simply changing the estimator, without requiring any additional clustering structure of the graph and without utilizing complex randomized designs. This is a significant gain when the polynomial degree $\beta$ is significantly lower than the graph degree $d$. The simplest setting of $\beta=1$ already expresses all potential outcomes models which satisfy additivity of main effects and additivity of interference, as defined in \cite{SussmanAiroldi17}; this subsumes all linear models which are commonly used in the practical literature, yet which require additional homogeneity assumptions.

We also remark that when $\beta=\dmax$, both the Horvitz-Thompson estimator and our estimator are unbiased for the same class of functions. When $\beta < \dmax$, the Horvitz-Thompson estimator is unbiased for a strictly larger class of functions than our estimator, precisely those characterized by $\beta=d$. In this way, if the practitioner can use domain knowledge to argue that the true potential outcomes model belongs to the class of functions parametrized by $\beta$ from Equation \ref{eq:pom_ld_neighbor} for $\beta < d$, then our estimator provides an advantage over the Hortvitz-Thompson estimator. This is because both estimators are unbiased but the variance of our estimator does not have exponential dependence on the graph degree. However, depending on the flexibility that the practitioner desires or needs to express in the potential outcomes, there is no clear winner. For example, if one desires a fully nonparametric potential outcomes to capture the most general neighborhood interference settings, they might set $\beta=d$. Then, both our estimator and the Horvitz-Thompson estimator are unbiased and both have variance scaling exponentially in $d$. On the other hand, suppose anonymous interference was satisfied for a particular application and the potential outcomes could be modeled: 
\[
    Y_i(\bz) = c_0 + c_1z_i + c_2\Big(\sum_{j \in \cN_i}z_j\Big).
\]
Then, using an ordinary least squares estimator would give an unbiased estimate with lower variance than our estimator. Thus, our model and estimator can be viewed as simply ``adding to the toolbox'' that practitioners may use depending on how expressive they need their potential outcomes models to be.

In the special setting of uniform Bernoulli design and $\beta=1$, our estimator as stated in \eqref{est:uniform_prob} is the same as an estimator presented in \cite{WagerLi2020}. They consider a fully nonparametric setting under anonymous interference, in which the goal is to estimate the derivative of the total outcomes under changes of the population-wide treatment probability. As the derivative can be estimated by locally linearizing the outcomes function, they derive the special case of the estimator in \eqref{est:uniform_prob} under $\beta=1$ by taking the derivative of the expected population outcomes under a Bernoulli randomization, constructed by inverse propensity weights. This suggests under a fully nonparametric setting, our estimator may be used for estimating an appropriately defined local estimand. There may also be opportunities to perform variance reduction on our estimator given knowledge of the graph structure, as proposed in \cite{WagerLi2020}; however, their solution requires anonymous interference, and it is not clear how to extend their solution concept beyond $\beta=1$, anonymous interference, and uniform treatment probabilities.

\subsection{Minimax Lower Bound on Mean Squared Error}

\mereplace{To understand the optimality of our estimator, we construct a lower bound on the minimax optimal mean squared error rate. In particular, we show that for a setting with a $d$-regular graph and uniform treatment probabilities $p$, the best achievable mean squared error is lower bounded by $\Omega\big(\frac{\sigma^2}{n} \big(\frac{d}{p}\big)^{\beta}\big)$ when $pd \ll 1/\beta$ and $\beta \ll d$. When $pd \gg \beta$ and $\beta \ll d$, the best achievable mean squared error is lower bounded by $\Omega\big(\frac{ \sigma^2}{n p^{2\beta}}\big)$. The first lower bound suggests that when $pd$ is small, our estimator is asymptotically minimax optimal. One can interpret $pd$ as the expected number of treated neighbors; if $pd$ is small, this means that the treatment budgets is limited enough that on average, many units have no treated neighbors.

Let $M = (E, \{c_{i,S}\}_{i, \cS \subset \cN_i, |\cS| \leq \beta})$ define the edge set $E$ and network effects for a $\beta$-degree model associated to the graph with edge set $E$. Let $\mathcal{M}_{d,\beta}$ denote the collection of all graphs with maximum degree bounded by $d$ and associated potential outcomes models with polynomial degree bounded by $\beta$, such that $M \in \mathcal{M}_{d,\beta}$.

\begin{theorem}[Minimax Lower Bound]
For any $n,d,\beta,p$, for any estimator $\widehat{\TTE}$, there exists an instance $M \in \mathcal{M}_{d,\beta}$ such that the minimax squared error under uniform treatment probabilities $p_i = p$ is bounded below by
\begin{align*}
 \E_M[(\widehat{\TTE} - TTE)^2] \geq \begin{cases}
\frac{\sigma^2}{32n} \left(\frac{d-\beta}{p}\right)^{\beta} (1 - p \beta (d-\beta)) &\text{ if } p(d-\beta) < 1/\beta, \\
\frac{\sigma^2}{16n} \left(\frac{d-\beta}{p}\right)^{\beta} \left(\frac{(p(d-\beta))^{\beta/2+1}\beta^{\beta/2}}{\beta - p(d-\beta)}
+ \frac{(p(d-\beta))^{\beta/2}}{p \beta (d-\beta) - 1}\right)^{-1}  &\text{ if } 1/\beta < p(d-\beta) < \beta, \\
\frac{\sigma^2}{32n p^{2\beta}} \left(\frac{p(d-\beta) - \beta}{p(d-\beta)}\right)  &\text{ if } p(d-\beta) > \beta.
\end{cases}
\end{align*}
\end{theorem}

When $p\beta(d-\beta)$ is small, the lower bound scales as $\Omega\big(\frac{\sigma^2}{n} \big(\frac{d-\beta}{p}\big)^{\beta}\big)$ as stated below. 
\begin{corollary}
For any $n,d,\beta,p$ satisfying $p\beta(d-\beta) < 1- C$ for some constant $C \in (0,1)$,
\begin{align*}
\inf_{\widehat{TTE}} \sup_{M \in \mathcal{M}_{d,\beta}} \E_M[(\widehat{\TTE} - TTE)^2] \geq \frac{C \sigma^2}{32n} \left(\frac{d-\beta}{p}\right)^{\beta}
\end{align*}
\end{corollary}
For a $\beta$-order interactions model, estimating the total treatment effect requires being able to measure the network effect of the size $\beta$ subsets of a unit's neighborhood. There are roughly $O(d^\beta)$ subsets of size $\beta$, and the probability that a set of size $\beta$ is jointly assigned to treatment is $p^{\beta}$. As a result, the scaling of $(d/p)^{\beta}$ is somewhat intuitive.

When $p(d-\beta)$ is larger than $\beta$, the lower bound scales as $\Omega\big(\frac{ \sigma^2}{n p^{2\beta}}\big)$ as stated below.
\begin{corollary}
For any $n,d,\beta,p$ satisfying $\beta < (1-C) p(d-\beta)$ for some constant $C \in (0,1)$,
\begin{align*}
\inf_{\widehat{TTE}} \sup_{M \in \mathcal{M}_{d,\beta}} \E_M[(\widehat{\TTE} - TTE)^2] \geq \frac{C \sigma^2}{32n p^{2\beta}}
\end{align*}
\end{corollary}

The proof of Theorem \ref{thm:lower_bound} follows from an application of LeCam's method \cite{lecam1973convergence}, which states that for some $\delta > 0$, for any two instances $M_A, M_B \in \mathcal{M}$ such that $|\TTE(M_A) - \TTE(M_B)| \geq 2\delta$, it holds that
\begin{align*}
\inf_{\widehat{\TTE}} \sup_{M \in \mathcal{M}_{d,\beta}} \E_M[(\widehat{\TTE} - \TTE(M))^2] \geq \frac{\delta^2}{2} (1 - \sqrt{D_{KL}(P_A \| P_B)/2}),
\end{align*}
where $P_A$ denotes the distribution over the data with respect to instance $M_A$, $P_B$ denotes the distribution over the data with respect to instance $M_B$, and $D_{KL}(\cdot \| \cdot)$ denotes the KL-divergence.

We assume the graph is $d$-regular such that the in-degrees and out-degrees are all equal to $d$, and we consider uniform treatment probability $p_i = p$. We construct instance $M_A$ by setting $c_{i, \cS} = \delta / \binom{|\cN_i|}{\beta}$ for all $i \in [n]$ when $|\cS| = \beta$, and $c_{i,\cS} = 0$ when $|\cS| < \beta$. Let $c'_{i,S}$ denote the corresponding network effect coefficients for instance $B$. We set $c'_{i, \cS} = -\delta / \binom{|\cN_i|}{\beta}$ for all $i \in [n]$ when $|\cS| = \beta$, and $c'_{i,\cS} = 0$ when $|\cS| < \beta$. As a result, $\TTE(M_A) = \delta$, and $\TTE(M_B) = -\delta$, satisfying the condition that $|\TTE(M_A) - \TTE(M_B)| \geq 2\delta$.

The rest of the proof stated in \cref{sec:lower_bd_pf} follows from upper bounding $D_{KL}(P_A \| P_B)$ and subsequently choosing $\delta$ to set the upper bound on $D_{KL}(P_A \| P_B)$ to $1/2$.

\begin{remark}
The minimax lower bounds presented above require the presence of observation noise, i.e. $\sigma > 0$. As a result, when $\sigma=0$, the minimax lower bound is meaningless. This limitation is a product of our particular analysis, and it would be interesting to see if one could relax this assumption and prove a minimax lower bound result directly on the noiseless setting.
\end{remark}
}
{

To understand the optimality of our estimator, we construct a lower bound on the minimax optimal mean squared error rate. In particular, we show that for a setting with a $d$-regular graph and sufficiently-small uniform treatment probabilities $p$, the best achievable mean squared error is lower bounded by $\Omega\big( \tfrac{1}{np^{\beta}} \big)$.

\begin{theorem}[Minimax Lower Bound]\label{thm:lower_bound}
For any $n,d,\beta,p$ with $p^\beta < 0.16$, and any estimator $\widehat{\emph{TTE}}$, there exists a causal network on $n$ nodes with maximum degree $d$ and effect coefficients $\{c_{i,\cS}\}_{i \in[n], \cS \in \cS_i}$ for which the minimax squared error under uniform treatment probabilities $p_i = p$ is bounded below by
\[
 \E\big[(\widehat{\emph{TTE}} - \emph{TTE})^2\big] = \Omega \big( \tfrac{1}{np^\beta} \big).
\]
\end{theorem}

For a $\beta$-order interactions model, estimating the total treatment effect requires being able to measure the network effect of the size $\beta$ subsets of a unit's neighborhood. The probability that a set of size $\beta$ is jointly assigned to treatment is $p^{\beta}$. As a result, the scaling of $\tfrac{1}{p^\beta}$ is somewhat intuitive.

The proof of Theorem \ref{thm:lower_bound} (given in Appendix~\ref{sec:lower_bd_pf}) uses a generalized variation of  LeCam's method for fuzzy hypothesis testing \cite{lecam1973convergence}, \cite[Sec. 2.7.4]{tsybakov2009springer}. We reduce the problem of $\TTE$ estimation to the creation of a hypothesis test to distinguish between two priors. 
We consider a setting with a $d$-regular graph, uniform treatment probabilities $p_i = p$, and two Gaussian priors $\Gamma_0, \Gamma_1$ over the effect coefficients. The priors have the same variances but with shifted means such that $\big| \E_{\Gamma_0} [ \TTE ] - \E_{\Gamma_1} [ \TTE ] \big| = 2\delta$ for some carefully tuned parameter $\delta$. The failure of hypothesis tests that rely on $\widehat{\TTE}$ is attributable to one of two factors: (1) a significant shift in $\TTE$ brought about by the variability of the coefficients, or (2) inaccuracy of the estimator. We bound the probability of the former with a Gaussian tail bound. We bound the latter in terms of the KL-Divergence between the distributions over $(\mathbf{Y},\mathbf{z})$ induced by $\Gamma_0$ and $\Gamma_1$. The rest of the argument involves a calculation of this KL-Divergence and a selection of $\delta$ to ensure a non-vanishing error probability. 

While our lower bound clearly indicates that the exponential dependence on $\beta$ as exhibited by $p^{-\beta}$ is necessary, a notable difference between our lower and upper bounds is the dependence on the graph degree $d$. Recall that the upper bound on the variance of our SNIPE estimator in \Cref{thm:est_bias_var} scales roughly as $d^{\beta+2}$; however, our lower bound result has no dependence on $d$. We attribute this gap to a weakness in the analyses and believe that it can be tightened with more careful calculation. In the upper bound, the $d^\beta$ arises as a bound on the sum of binomial coefficients of the form $\binom{d}{k}$ for $k \leq \beta$. While this bound is precise asymptotically for a regime where $\beta$ is a small constant and $d$ is large, it is loose in the most general case where $\beta = d$, where we know $2^d$ is sufficient. On the other hand, our lower bound argument considers a setting with a $d$-regular graph, but it only uses the degree of the graph to normalize the parameters on the prior distributions, such that we see no dependence on $d$ in the final lower bound.
}

\subsection{Central Limit Theorem}
\label{sec:clt}

\mcedit{In this section, we show $\widehat{\TTE}_\text{SNIPE}$ is asymptotically normal using Stein's method, a common approach to proving central limit theorems \cite{AronowSamii17, chin2018central, Leung2020, harshaw2023design, ogburn2022causal}.} In particular, we use Theorem 3.6 from Ross~\cite{ross2011fundamentals}, which we restate below for convenience.

\textbf{\cite{ross2011fundamentals} Theorem 3.6.} 
Let $X_1,\ldots,X_n$ be random variables such that $\E[X_i^4] < \infty$, $\E[X_i]=0$, $\nu^2 = \Var \big(\sum_i X_i\big)$, and define $W = \sum_i X_i/\nu.$ Let collection $(X_1, \ldots, X_n)$ have dependency neighborhoods $D_i$ for $i\in [n]$, and also define $D:=\max\limits_{1 \leq i \leq n}|D_i|.$ Then, for $Z$ a standard normal random variable,
\begin{align} \label{eq:d_W_bd}
d_\text{W}(W,Z) \leq \frac{D^2}{\nu^3} \sum_{i=1}^n \E|X_i|^3 + \frac{\sqrt{28}D^{3/2}}{\sqrt{\pi}\nu^2} \sqrt{\sum_{i=1}^n \E[X_i^4]},
\end{align}
where $d_\text{W}(W,Z)$ is the Wasserstein distance between $W$ and $Z$.

Our estimator can be written in the form of $\widehat{\TTE} = \frac{1}{n} \sum_{i \in [n]} Y_i \,  w_i(\bz)$, where $\Yobs{i}$ and $w_i(\bz)
$ depend only on the treatment assignments of individuals $j \in \cN_i$.
We apply Theorem 3.6 from \cite{ross2011fundamentals}, to the random variables $X_i := \tfrac{1}{n} (\Yobs{i}w_i(\bz) - \E\big[\Yobs{i}w_i(\bz)\big])$, such that by construction $W = \sum_i X_i/\nu = (\widehat{\TTE} - \TTE) / \nu$ and $\widehat{\TTE} = \nu W + \TTE$. An application of Theorem 3.6 from \cite{ross2011fundamentals} implies that $\widehat{\TTE}$ is asymptotically normal as long as the bound in \eqref{eq:d_W_bd} converges to 0 with $n$. 
As the size of the dependency neighborhoods $D$, the moments of $X_i$, and the scaling of the variance $\nu$ all depend on the network parameters, for simplicity we additionally assume in the conditions of Theorem \ref{thm:clt} that $d$, $Y_{\max}$, and $\beta$ are bounded above by a constant, and we also assume the treatment probability $p$ is also lower bounded as a function of $n$.

\mcedit{
\begin{assumption}
\label{assp:nondegeneracy}
For some constant $c > 0$, we have $n \cdot \emph{\Var}\big[\widehat{\TTE}\big] \to c$ as $n\to \infty$.
\end{assumption}
This is a typical assumption made in the literature (see \cite{AronowSamii17}, \cite{leung2022rate}, \cite{harshaw2023design}) that rules out degenerate cases, such as all the potential outcomes being $0$, that would result in the estimator having an unnaturally low variance scaling smaller than $1/n$. 
}

\mcedit{\begin{theorem}[CLT] \label{thm:clt}
Under Assumptions \ref{assp:neighborhood_interference} - \ref{assp:nondegeneracy}, and assuming that
$d$, $Y_{\max}$, and $\beta$ are all $O(1)$ with respect to $n$, $p = \omega(n^{-1/4\beta})$, and $p_i \leq \frac12$ for all $i$, the normalized error $\big(\widehat{\TTE} - \TTE\big)/\nu$ converges in distribution to a standard normal random variable as $n \to \infty$, for $\nu^2 = \Var[\widehat{\TTE}]$.
\end{theorem}}

Theorem \ref{thm:clt} states that our estimator is asymptotically normal, 
assuming the boundedness of the magnitude of the potential outcomes, the polynomial degree, and the network degree. \mcedit{For the complete proof, refer to Appendix \ref{sec:CLT_proof}.}
The proof follows from a straightforward application of Theorem 3.6 from \cite{ross2011fundamentals} with appropriate bounds for the moments $\E|X_i|^3$ and $\E[X_i^4]$, the dependency neighborhood size $D$, and the variance $\nu$.
The dependency neighborhood for variable $X_i$ is the set $\{j \in [n] ~\text{s.t.}~ \cN_i \cap \cN_j \neq \emptyset\}$, i.e. the set of individuals $j$ that share in-neighbors with individual $i$. This follows from the observation that $X_i$ and $X_j$ are both are a function of the treatment variables $z_k$ for shared in-neighbors $k \in \cN_i \cap \cN_j$. 
The size of the largest dependency neighborhood is thus bounded by $D \leq \din \dout$. The moments $\E|X_i|^3$, $\E[X_i^4]$ will be bounded as a function of $Y_{\max}$ as defined in \eqref{eq:Y_max} along with the network degree $d$, treatment probability $p$, and the polynomial degree $\beta$. 
The boundedness assumptions are used to argue that these moments do not grow too quickly in $n$, so that the Wasserstein distance in \eqref{eq:d_W_bd} converges to 0. The details of these calculations are deferred to \Cref{sec:CLT_proof}.
We remark that the strict boundedness assumption can be relaxed so that $d$, $Y_{\max}$, and $\beta$ are polylogarithmic with respect to $n$, so long as we tighten the lower bound on the growth of $p$ by a corresponding logarithmic factor.

\mcedit{\subsection{Variance Estimator} \label{ssec:variance_est}}

The central limit theorem result in Theorem \ref{thm:clt} implies that we can construct asymptotically valid confidence intervals by $ [\widehat{\TTE}- \nu \Phi^{-1}(1-\alpha), \widehat{\TTE} + \nu \Phi^{-1}(1-\alpha)]$ if we knew $\nu$, where $\Phi$ indicates the cdf of a standard normal. 
\mcreplace{While we are not able to construct a consistent estimator for $\nu$, we can construct conservative confidence intervals by substituting $\nu$ with a conservative bound on the variance as stated in Theorem \ref{thm:est_bias_var},
\begin{align*}
\nu = \sqrt{\Var[\widehat{\TTE}]} \leq \sqrt{\frac{\din \; \dout \; {Y_{\max}}^2}{n} \cdot \Big( \frac{\din^2}{p(1-p)} \Big)^{\beta}
        \mcdelete{+\frac{\sigma^2}{n} \left(\frac{\din}{p(1-p)}\right)^{\beta}}}.
\end{align*}
An open question for future work is to obtain tighter estimators of the variance $\nu$ to strengthen the constructed confidence intervals.}{Since the practitioner typically does not know $\nu$, we construct a conservative estimator for $\nu$ by applying the approach from Aronow and Samii \cite{aronow2013conservative, AronowSamii17}.}
\mcedit{We begin by rewriting the SNIPE estimator \mereplace{. Denote $\tau(\bz) = \{j ~\text{s.t.} \ z_j = 1\}$ and note that we can write}{as a sum over possible exposures $\bx$:
\[
    \widehat{\TTE} 
    = \tfrac{1}{n} \sum_{i \in [n]} Y_i(\bz) w_i(\bz) 
    = \tfrac{1}{n} \sum_{i \in [n]} \sum_{\bx \in \{0,1\}^{|\cN_i|}} \Ind(\bz_{\cN_i} = \bx) Y_i(\bx) w_i(\bx). 
\]}
Note that both $Y_i(\bx)$ and $w_i(\bx)$ are deterministic, and the only randomness is expressed in the indicator function $\Ind(\bz_{\cN_i} = \bx)$.
Overloading notation, we let $Y_i$ and $w_i$ be expressed only as a function of $\bz_{\cN_i}$, where we assume the order $\beta$ is clearly specified. Then, the variance of our estimator is given by 
\begin{align*}
\Var[\widehat{\TTE}] &=
\frac{1}{n^2}\sum_{i,j\in[n]} \sum_{\bx \in \{0,1\}^{|\cN_i|}} \sum_{\bx' \in \{0,1\}^{|\cN_j|}} Y_i(\bx) w_i(\bx) \ Y_j(\bx') w_j(\bx') \ \Cov\big(\Ind(\bz_{\cN_i} = \bx), \Ind(\bz_{\cN_j} = \bx')\big).
\end{align*}
}
\mcedit{If $\Prob(\{\bz_{\cN_i} = \bx\} \cap \{\bz_{\cN_j} = \bx'\}) > 0$, \cyreplace{we can approximate that term in the variance expression by}{an unbiased estimate for the corresponding term in the variance expression is given by} 
\[\frac{\Ind(\bz_{\cN_i} = \bx)\Ind(\bz_{\cN_j} = \bx')}{\Prob(\{\bz_{\cN_i} = \bx\} \cap \{\bz_{\cN_j} = \bx'\})} Y_i(\bx) w_i(\bx) Y_j(\bx') w_j(\bx')  \Cov(\Ind(\bz_{\cN_i} = \bx), \Ind(\bz_{\cN_j} = \bx')).\]
Otherwise, if $\Prob(\{\bz_{\cN_i} = \bx\} \cap \{\bz_{\cN_j} = \bx'\}) = 0$,  by the same technique as \cite{aronow2013conservative}, \cyedit{using the observation that $\Cov(\Ind(\bz_{\cN_i} = \bx), \Ind(\bz_{\cN_j} = \bx')) = \Prob(\bz_{\cN_i} = \bx)\Prob(\bz_{\cN_j} = \bx')$,} the corresponding term in the variance expression can be approximated by the inflated estimate
\[\frac12 \Big(\Ind(\bz_{\cN_i} = \bx) \Prob(\bz_{\cN_i} = \bx) Y_i^2(\bx) w_i^2(\bx)
+ \Ind(\bz_{\cN_j} = \bx') \Prob(\bz_{\cN_j} = \bx') Y_j^2(\bx') w_j^2(\bx')
\Big).\]
\cyedit{The expectation of this estimate will be higher than the true term in the variance via Cauchy-Schwarz inequality.}
Therefore, a conservative variance estimator $\widehat{\Var}(\widehat{\TTE}_\text{SNIPE})$ is given by
\begin{align*}
&\frac{1}{n^2}\sum_i \sum_{j \in \cM_i} \sum_{\bx \in \{0,1\}^{|\cN_i \cup \cN_j|}} \frac{\Ind(\bz_{\cN_i \cup \cN_j} = \bx)}{\Prob(\bz_{\cN_i \cup \cN_j} = \bx)} Y_i(\bx) w_i(\bx) Y_j(\bx) w_j(\bx) \; \Cov\Big(\Ind\big(\bz_{\cN_i} = \bx_{\cN_i}\big), \Ind\big(\bz_{\cN_j} = \bx_{\cN_j}\big) \Big) \\
&\quad \hspace{30mm} + \ \frac{1}{n^2}\sum_i \sum_{\bx \in \{0,1\}^{|\cN_i|}} \Ind(\bz_{\cN_i} = \bx) \Prob(\bz_{\cN_i} = \bx) Y_i^2(\bx) w_i^2(\bx) \sum_{j \in \cM_i} \left(2^{|\cN_j|} - 2^{|\cN_j \setminus \cN_i|}\right),
\end{align*}
where 
\[
\Cov\Big(\Ind\big(\bz_{\cN_i} = \bx_{\cN_i}\big), \Ind\big(\bz_{\cN_j} = \bx_{\cN_j}\big) \Big) = 
\prod_{k\in\cN_i \cup \cN_j} p_k^{x_k} (1-p_k)^{1-x_k}
\left(1 - \prod_{k'\in\cN_i \cap \cN_j} p_{k'}^{x_{k'}} (1-p_{k'})^{1-x_{k'}}\right).
\]
In the Section \ref{sec:experiments}, we compare the empirical variance of the SNIPE estimator $\widehat{\TTE}$ with this conservative variance estimate $\widehat{\Var}(\widehat{\TTE})$ in simulations. We also compare against the worst-case variance bound from Theorem \ref{thm:est_bias_var}. Most notable is that both the conservative variance estimate and the worst-case variance bound are orders of magnitude larger than the empirical variance, suggesting that there is significant work needed to develop tighter variance estimates.}

\section{Experimental Results} \label{sec:experiments}

Using computational experiments on simulated data, we compare the performance of our estimator with existing estimators. Using an Erd\H os-R\'enyi model, we generate random directed graphs of $n$ nodes for a population of $n$ individuals. 
Figure \ref{fig:nonuniform_resultsER} shows results from networks made using the Erd\H os-R\'enyi model with $n$ nodes and probability $p_{\text{edge}} = 10/n$ of an edge existing between any two nodes. Hence, the expected in-degree and out-degree of each node is $10$. For degree $\beta,$ we construct the same potential outcomes model as in \cite{YuCortezEichhorn22}:
\begin{equation}
    Y_i(\bz) = c_{i,\emptyset} + \sum_{j\in \cN_i} \tilde{c}_{ij}z_j + \sum_{\ell = 2}^{\beta} \left( \frac{\sum_{j \in \cN_i}\tilde{c}_{ij}z_j}{\sum_{j \in \cN_i}\tilde{c}_{ij}} \right)^{\ell},
\end{equation}
where $c_{i,\emptyset} \sim U[0,1]$, $\tilde{c}_{ii} \sim U[0,1]$, and for $i \neq j$, $\tilde{c}_{ij} = v_j |\cN_i|/\sum_{k: (k,j) \in E} |\cN_k|$ for $v_j \sim U[0,r]$, where $r$ denotes a hyperparameter that governs the magnitude of the network effects relative to the direct effects. 
We represent the magnitude of individual $j$'s influence by the parameter $v_j$. This influence is shared among individual $j$'s out-neighbors proportional to their in-degrees.


\subsection{Other Estimators}
\label{subsec:otherestimators}
We compare the performance of SNIPE with the performance of least-squares regression and difference-in-means estimators, also as in \cite{YuCortezEichhorn22}. Although the Horvitz-Thompson estimator is unbiased in this setting, under unit Bernoulli randomized design its variance is very high in practice. Thus, we omit this estimator from our experiemental results and comparison. Another related estimator is the H\'ajek estimator, which is only approximately unbiased but with lower variance than the Horvitz-Thompson estimator. However, under our randomized design, its variance is still very high in practice. Additionally, as we implemented a Bernoulli randomized design on a graph with expected network degree of 10, both the H\'ajek and Horvitz-Thompson estimators consistently took a value of $0$, giving us no meaningful results to consider. For these reasons, we chose to omit these two estimators from our experimental results and comparison.
The simplest difference-in-means estimator is the difference between the average outcome of individuals assigned to treatment and the average outcome of individuals assigned to control, given by
\begin{equation}
    \widehat{\TTE}_{\text{DM}} = \frac{\sum_{i\in[n]}z_i Y_i}{\sum_{i\in[n]}z_i} - \frac{\sum_{i\in[n]}(1-z_i) Y_i}{\sum_{i\in[n]}(1-z_i)}. \label{eq:DiffMeans-Stnd}
\end{equation}
This estimator does not take into account any information about each individual's neighborhood and is biased under network interference. We also consider a modified version of this estimator that uses information about the number of treated neighbors of each individual. Let $U_i$ denote the number of individuals in $\cN_i \setminus \{i\}$ assigned to treatment, and let $\tilde{U}_i$ denote the number of neighbors individuals in $\cN_i \setminus \{i\}$ assigned to control. Then, the estimator is given by
\begin{equation}
    \widehat{\TTE}_{\text{DM}(\lambda)} = \frac{\sum_{i\in[n]}z_i \Ind(U_i \geq \lambda) Y_i}{\sum_{i\in[n]} z_i \Ind(U_i \geq \lambda) } - \frac{\sum_{i\in[n]}(1-z_i)\Ind(\tilde{U}_i \geq \lambda) Y_i}{\sum_{i\in[n]}(1-z_i)\Ind(\tilde{U}_i \geq \lambda)}, \label{eq:DiffMeans-Thresh}
\end{equation}
for some user-defined tolerance $\lambda \in [0,1]$. We set $\lambda=0.75$ for our experiments. Note that $\widehat{\TTE}_{\text{DM}(\lambda)}$ counts an individual $i$'s outcome only when at least $\lambda$ of their neighborhood is assigned to the same treatment as them.

We also compare with least-squares regression models of degree $\beta$ which assume the potential outcomes model is given by 
\begin{equation} \label{eq:LS-Prop-model}
    Y_i(\bz) = g(z_i, \bar{z}_i) = \Big(\rho + \textstyle\sum_{k=1}^{\beta} \gamma_k \, X_i^k\Big) + z_i \Big(\tilde{\rho} + \textstyle\sum_{k=1}^{\beta-1} \tilde{\gamma}_k \, X_i^k\Big),
\end{equation}
for some covariate $X_i$. We consider two variations. In the first, we set $X_i$ equal to the number of treated neighbors. In the second, we let $X_i$ equal the proportion of treated neighbors. In both cases, we do not include $i$ in its neighborhood. The two sets of coefficients $(\rho, \gamma_1, \dots \gamma_{\beta})$ and $(\tilde{\rho}, \tilde{\gamma}_1, \dots \tilde{\gamma}_{\beta})$ allow for the model to be different when $i$ is treated vs not treated, and since we only allow up to degree $\beta$ interactions, the second summation stops at $\beta-1$. Overall, there are $2\beta+1$ coefficients in the model. Using least-squares regression, we determine the set of coefficients minimizing the least-squares predictive error on the data set $\{z_i, X_i, Y_i(\bz)\}_{i \in [n]}$. These coefficients define an estimate for the function $\hat{g}$ in Equation \ref{eq:LS-Prop-model}. When $X_i = U_i$, the number of treated neighbors, the estimate is given by
\begin{equation}
    \widehat{\TTE}_{\text{LS-Num}} = \tfrac{1}{n} \textstyle\sum_{i=1}^n (\hat{g}(1,|\cN_i|-1) - \hat{g}(0,0)).
\end{equation}
When we set $X_i = U_i/(|\cN_i| - 1)$, the proportion of treated neighbors, we have
\begin{equation}
    \widehat{\TTE}_{\text{LS-Prop}} = \tfrac{1}{n} \textstyle\sum_{i=1}^n (\hat{g}(1,1) - \hat{g}(0,0)).
\end{equation}

\subsection{Results and Discussion}

For each population size $n$, we sample $G$ networks from the Erd\H os-R\'enyi model described previously. For every configuration of parameters in the experiment, we sample $N$ treatment assignment vectors $\bz_1,\ldots,\bz_N$ from a uniform Bernoulli distribution with treatment probability $p$ to compute the TTE using each estimator. Each plot we include also shows the relative bias of the TTE estimates, averaged over the results from these $GN$ samples and normalized by the magnitude, for each estimator. The width of the shading around each line in the plots shows the standard deviation across the $GN$ estimates. For our experiments\footnote{The Python scripts for the experiments and the data used in our results are available at: \url{https://github.com/mayscortez/low-order-unitRD}}, we chose $G=10$ and $N=500$. 

\begin{figure}[t]
    \centering
    \begin{subfigure}[b]{0.32\textwidth}
        \centering
        \includegraphics[width=\textwidth]{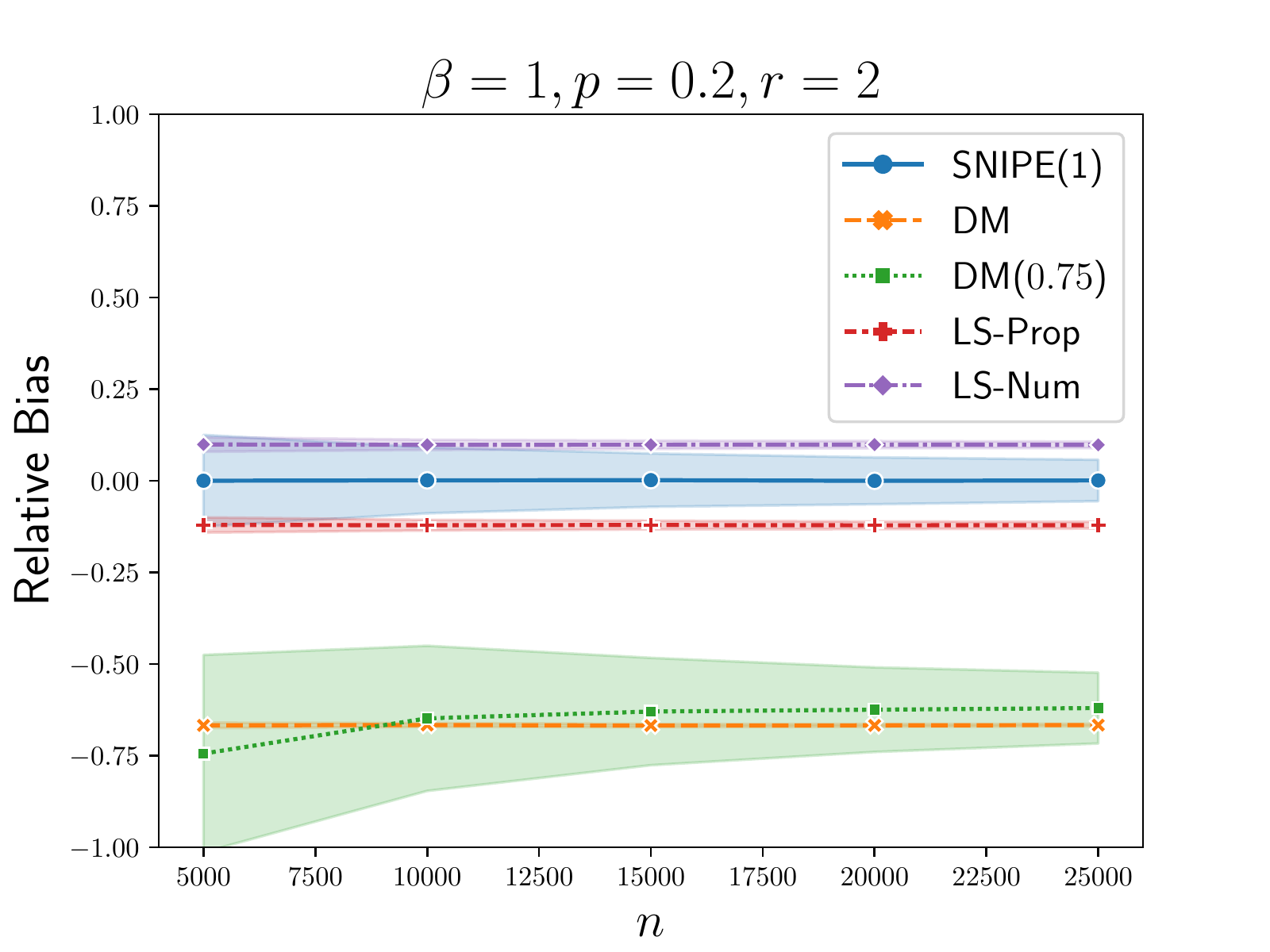}
        \includegraphics[width=\textwidth]{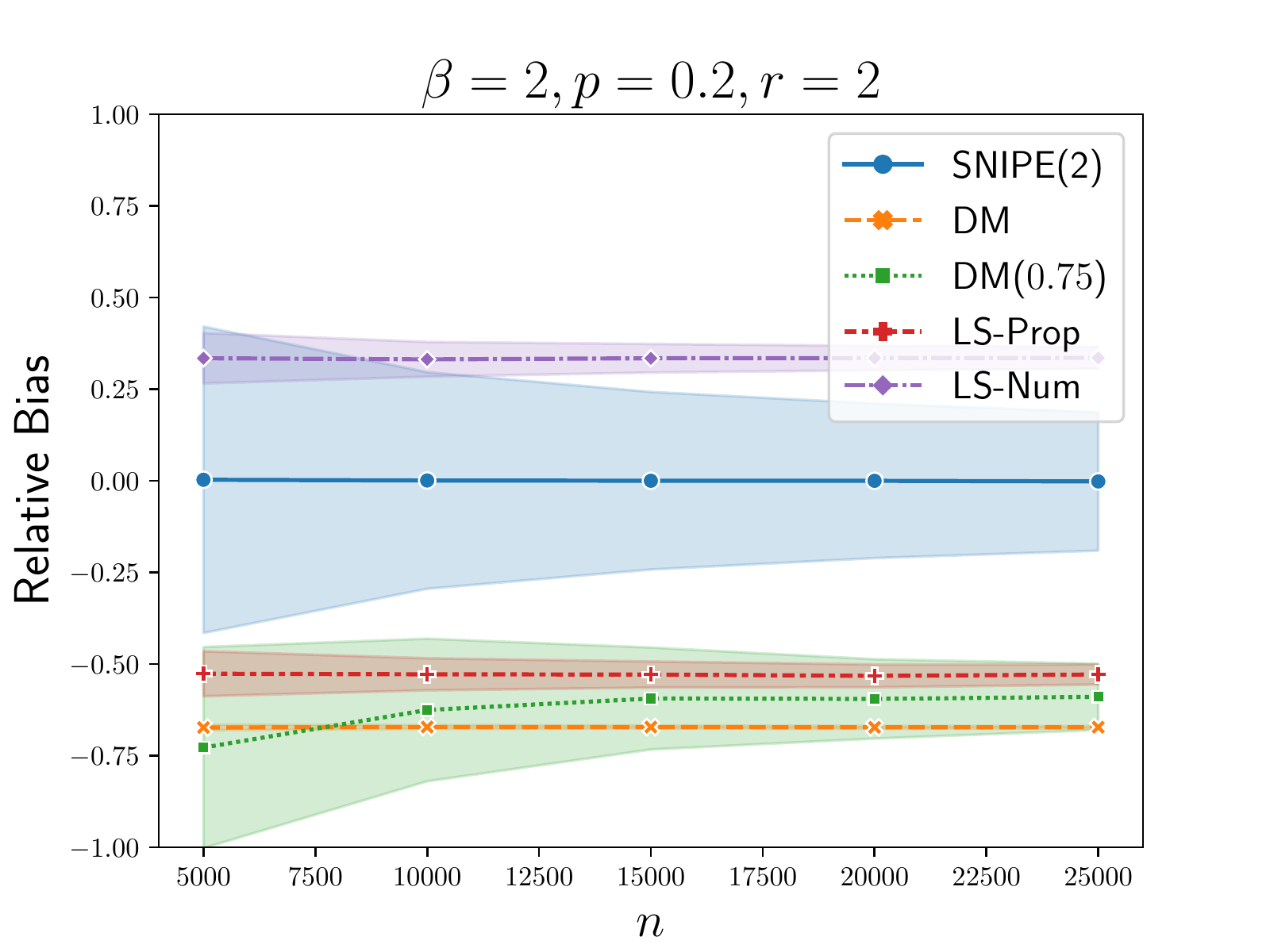}
        \caption{Varying population size}  \label{fig:sizeER}
    \end{subfigure}
    \begin{subfigure}[b]{0.32\textwidth}
        \centering
        \includegraphics[width=\textwidth]{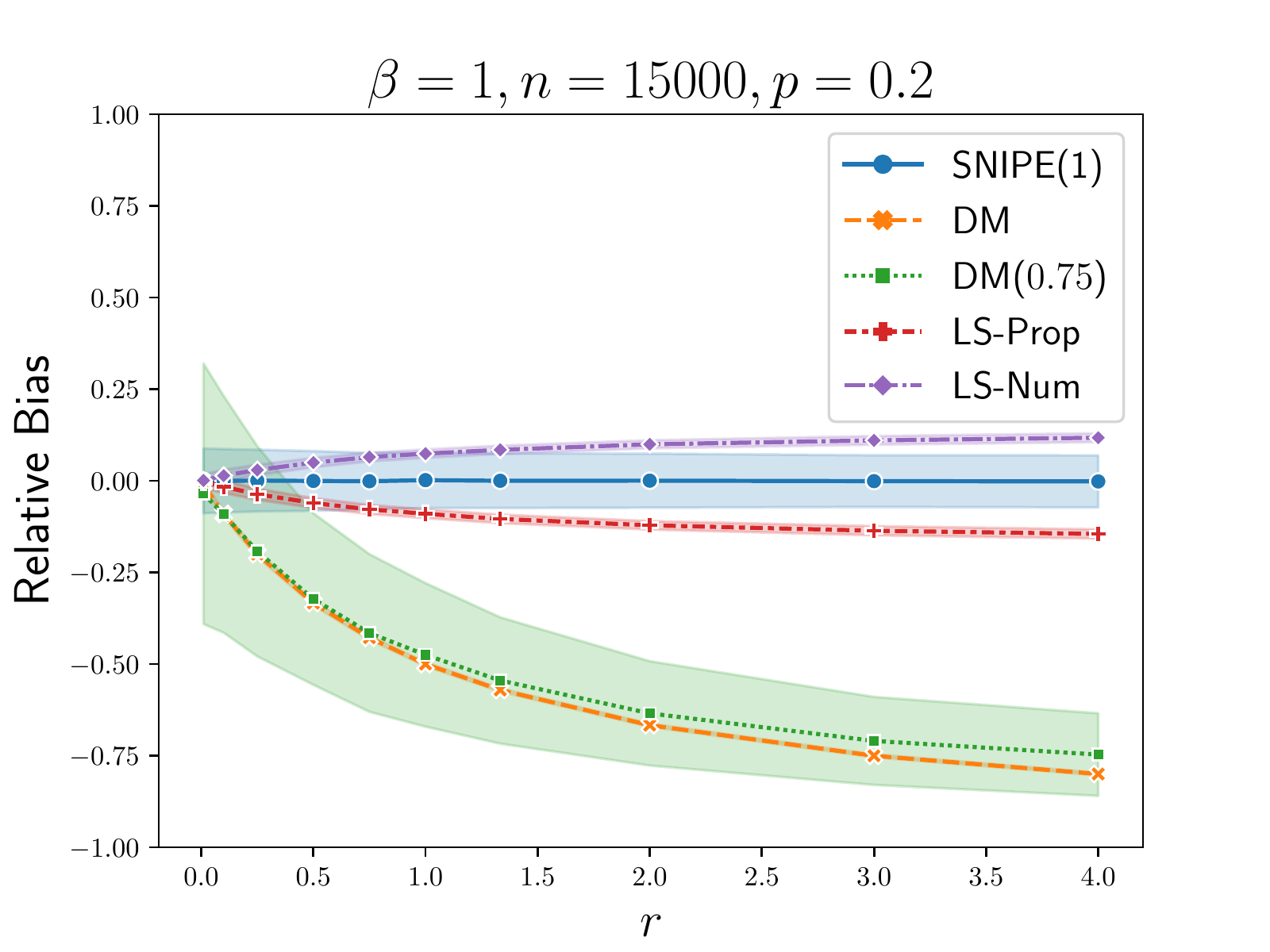}
        \includegraphics[width=\textwidth]{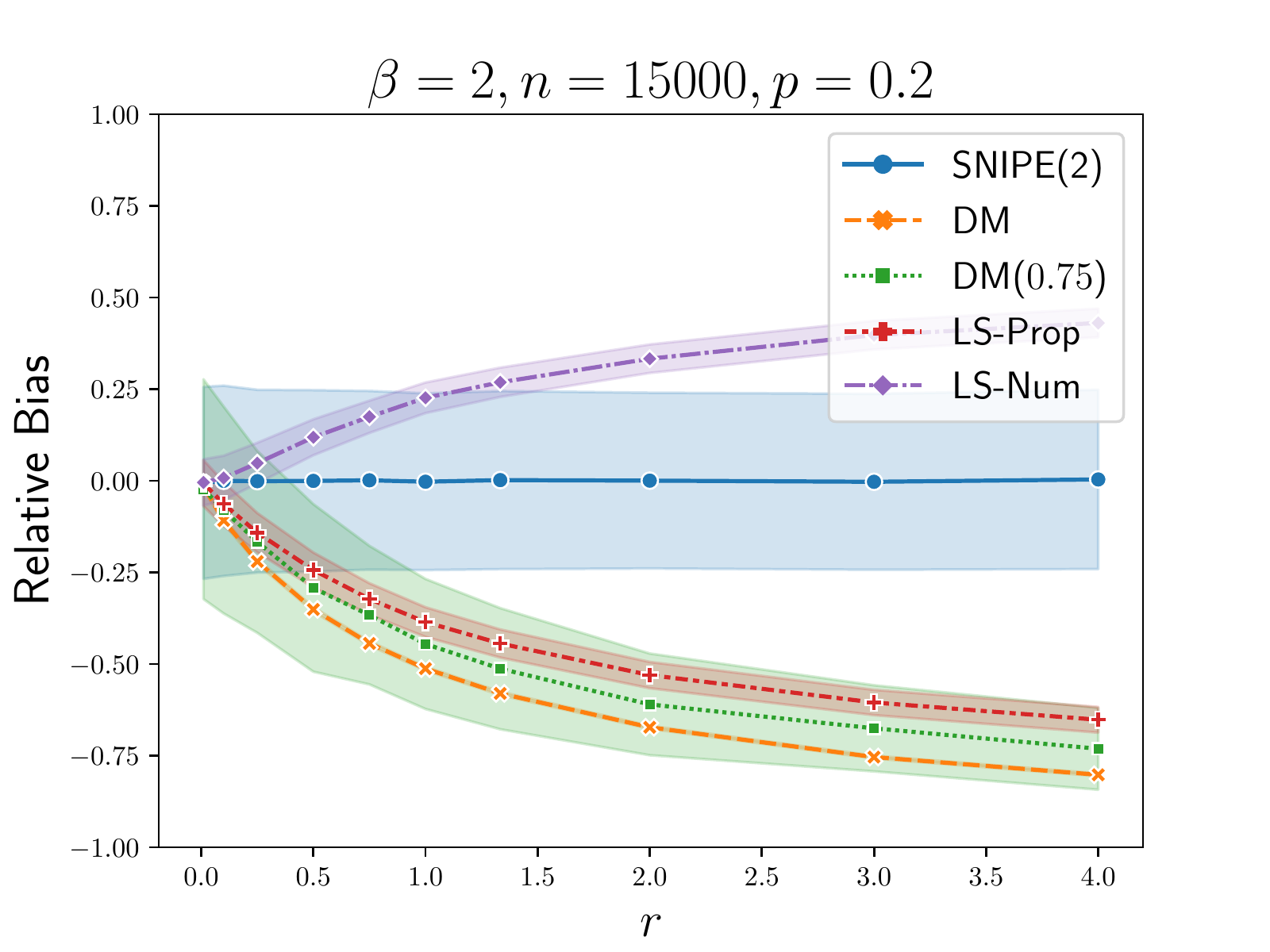}
        \caption{Varying direct:indirect effects}  \label{fig:ratioER}
    \end{subfigure}
    \begin{subfigure}[b]{0.32\textwidth}
        \centering
        \includegraphics[width=\textwidth]{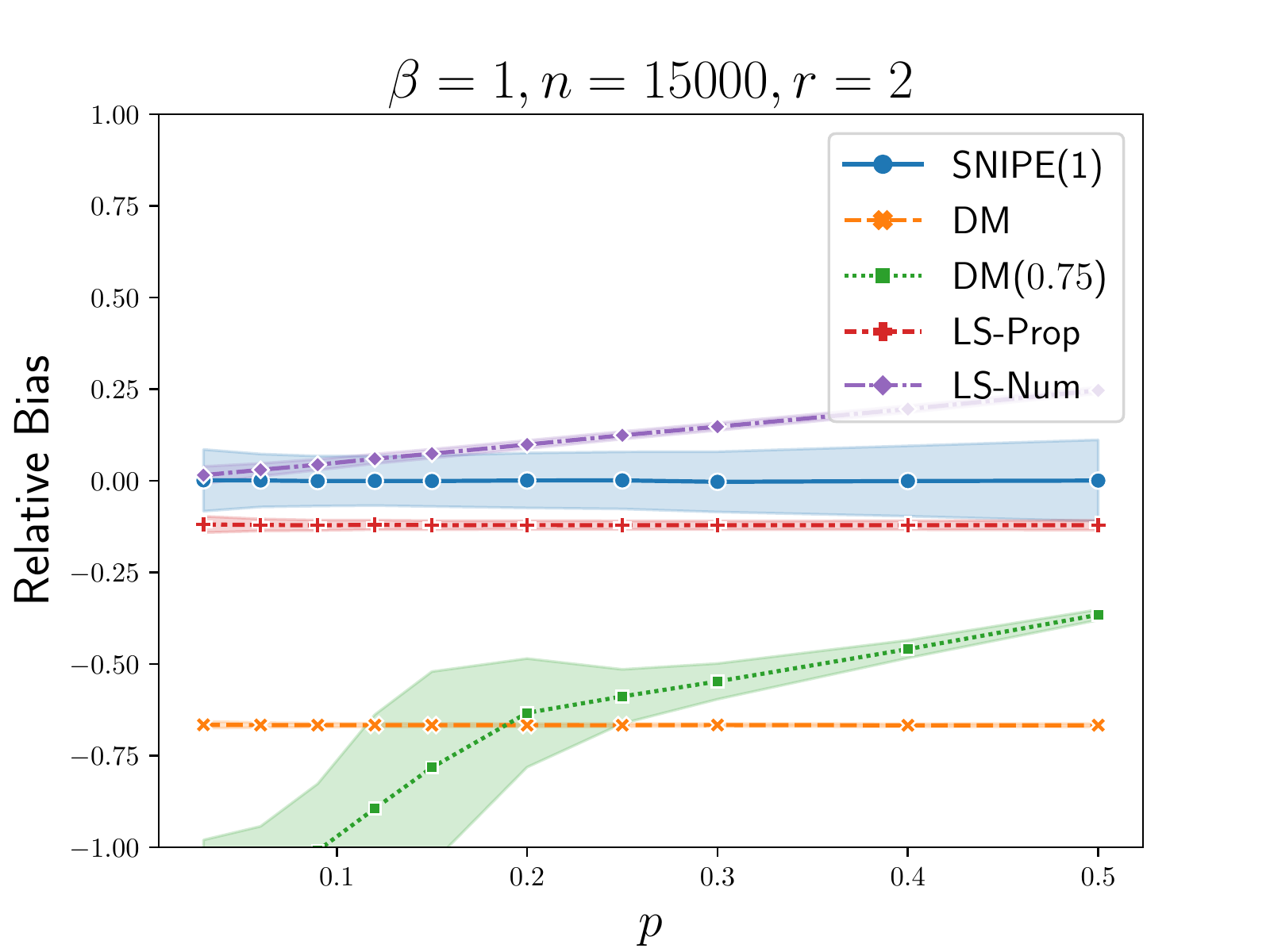}
        \includegraphics[width=\textwidth]{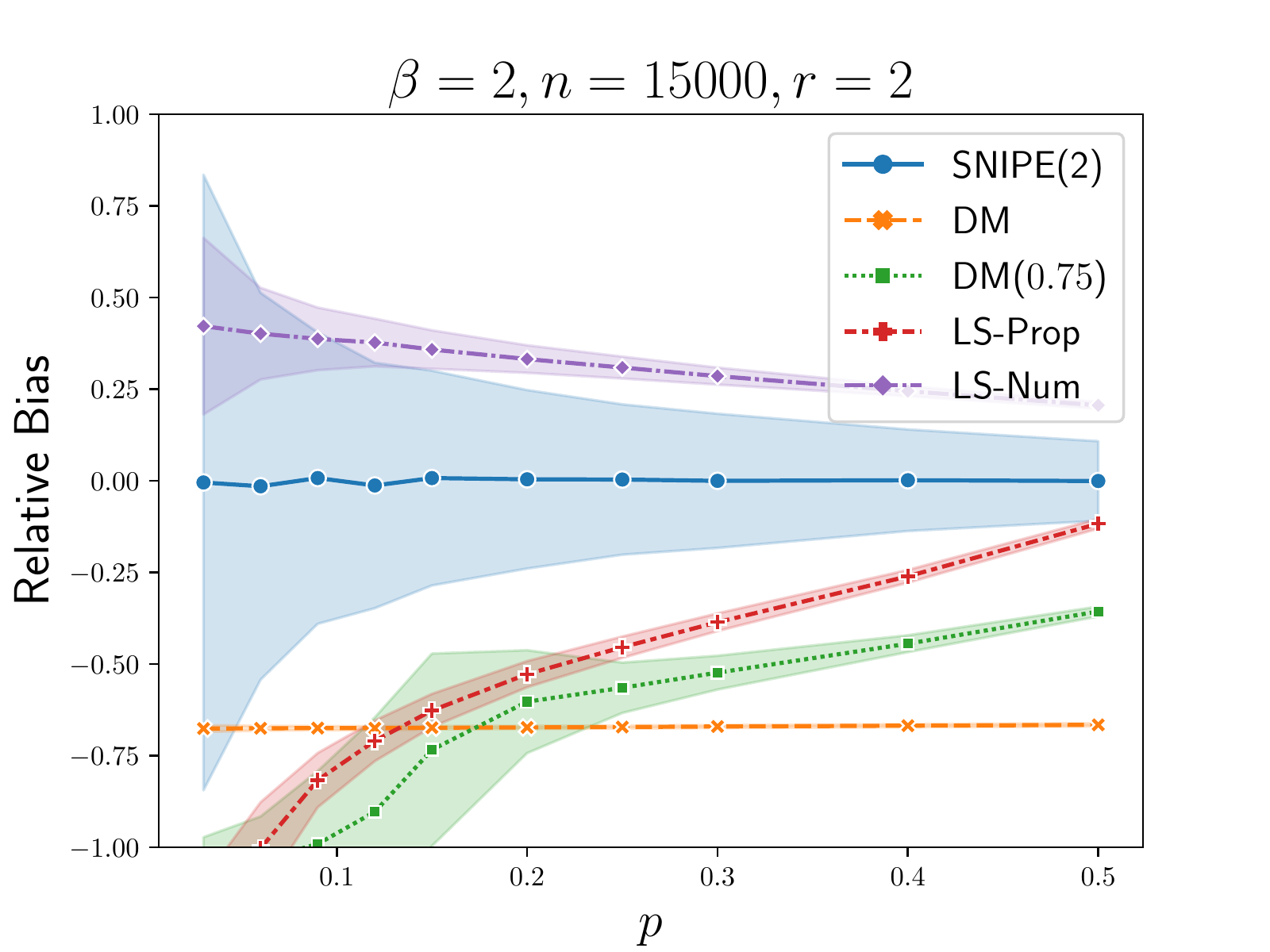}
        \caption{Varying treatment budget}  \label{fig:pER}
    \end{subfigure}
        \caption{Plots visualizing the performance of various $\TTE$ estimators under Bernoulli design on Erd\H os-R\'enyi networks for both linear and quadratic potential outcomes models. The height of each line on a plot depicts the experimental relative bias of the estimator and the shaded width depicts the experimental standard deviation. The SNIPE estimator is parametrized by $\beta$, the degree of the potential outcomes model.}\label{fig:nonuniform_resultsER}
\end{figure}

Figures \ref{fig:nonuniform_resultsER} and \ref{fig:nonuniform_resultsER_MSE} visualize the effects of various network or estimator parameters on the performance of each of the four TTE estimators described in Section \ref{subsec:otherestimators} and $\widehat{\TTE}_{\textrm{SNIPE}(\beta)}$, all under Bernoulli randomized design. In particular, we consider the effects of the population size ($n$), the treatment budget ($p$), the ratio between the network and direct effects ($r$), and the degree of the potential outcomes model ($\beta$). We list specific values for the parameters above each plot. Figure \ref{fig:nonuniform_resultsER} shows the bias and empirical standard deviation of each estimator, where the values are all normalized by the magnitude of the true TTE. Figure \ref{fig:nonuniform_resultsER_MSE} plots the empirical mean squared error (MSE) of each estimator, also normalized by the magnitude of the true TTE. The normalization can alternately be viewed as standardizing all models so that the ground truth TTE is 1.

The top row of plots in Figure \ref{fig:nonuniform_resultsER} features results for a linear ($\beta=1$) potential outcomes model while the bottom row shows results for a quadratic ($\beta=2$) potential outcomes model. As expected, the SNIPE estimator, shown in blue, has no relative bias and its variance decreases as $n$ increases.
With the exception of the modified difference-in-means estimator $\widehat{\TTE}_{\text{DM($0.75$)}}$ in green, the variances of the other estimators are lower than ours. However, the biases of the other estimators are larger than the standard deviation of our unbiased estimator overall. Moreover, as $r$ increases, the networks effects are more significant than the direct effects and we see the biases of the other estimators grow larger. Note that the variance of our estimator remains relatively constant as $r$ varies. When $r$ is close to $0$, there are essentially no network effects, SUTVA holds and as expected, all the estimators are unbiased.

Figure \ref{fig:nonuniform_resultsER_MSE} shows that for many parameter combinations the MSE of our estimator is lower than the other estimators; this is particularly the case for sufficiently large population sizes (large $n$) and sufficiently significant relative network effects (large $r$). In the top row of plots, corresponding to $\beta=1$, the difference in means estimators perform poorly relative to the other estimators so that they are beyond the upper limit of the displayed plot. While the performance of the least squares estimators and our estimator is comparable for $\beta=1$, the MSE of our estimator is solely due to variance, which will decrease with large $n$, yet the MSE of the least squares estimators is largely due to its bias, which will not decrease with large $n$, highlighting that they are not consistent estimators for our heterogeneous model.

\begin{figure}[t]
    \centering
    \begin{subfigure}[b]{0.32\textwidth}
        \centering
        \includegraphics[width=\textwidth]{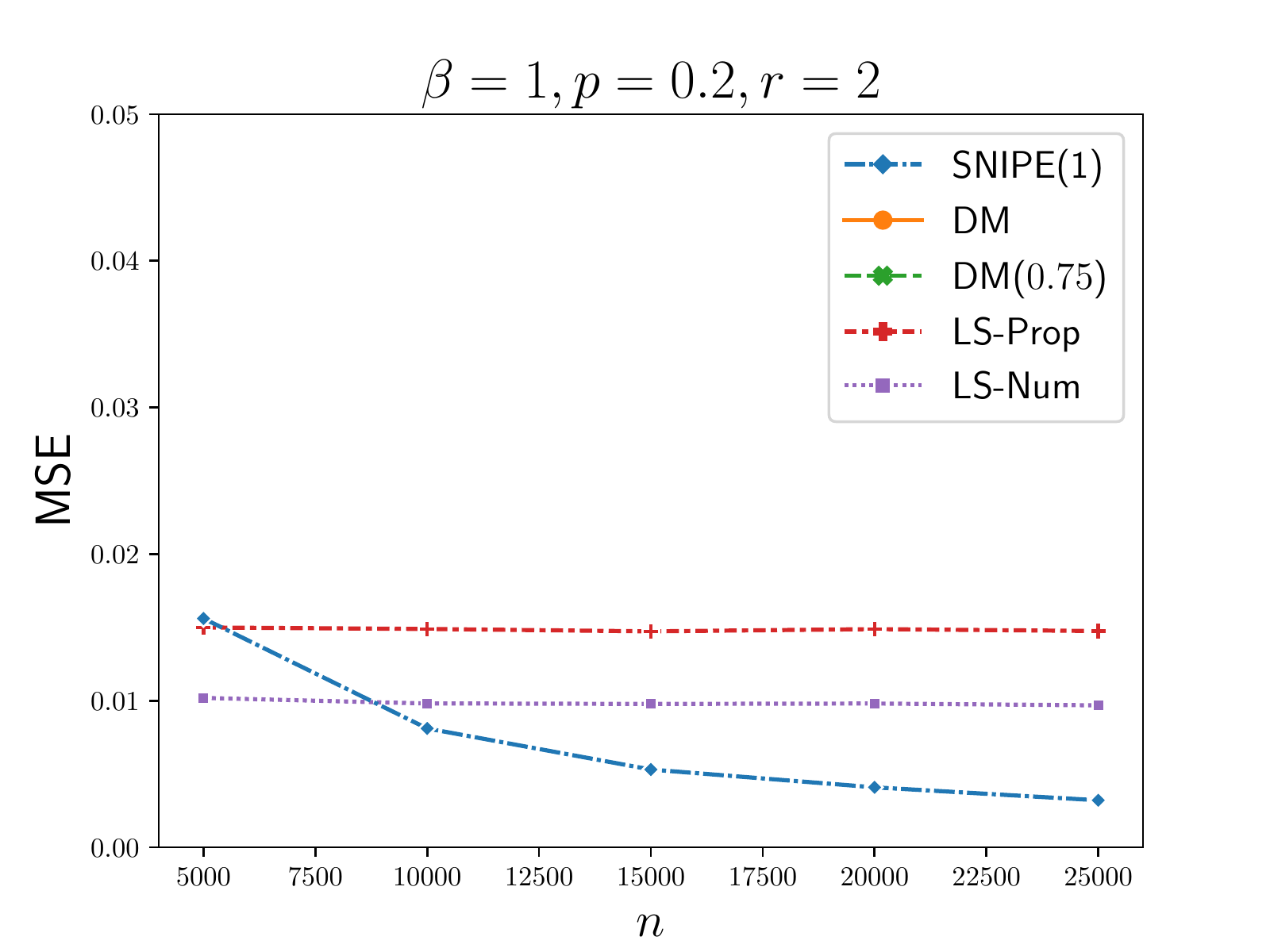}
        \includegraphics[width=\textwidth]{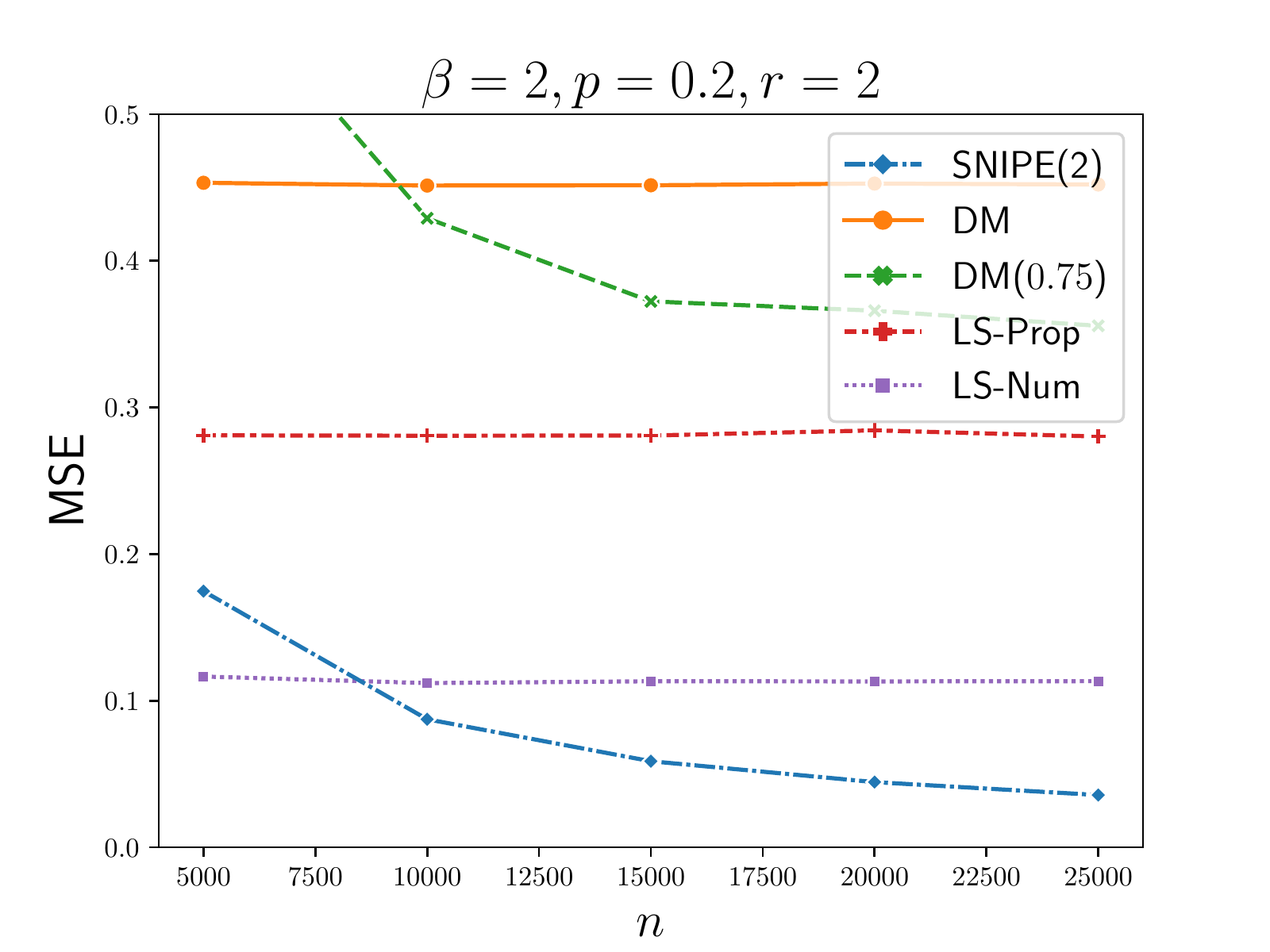}
        \caption{Varying population size}  \label{fig:sizeER_MSE}
    \end{subfigure}
    \begin{subfigure}[b]{0.32\textwidth}
        \centering
        \includegraphics[width=\textwidth]{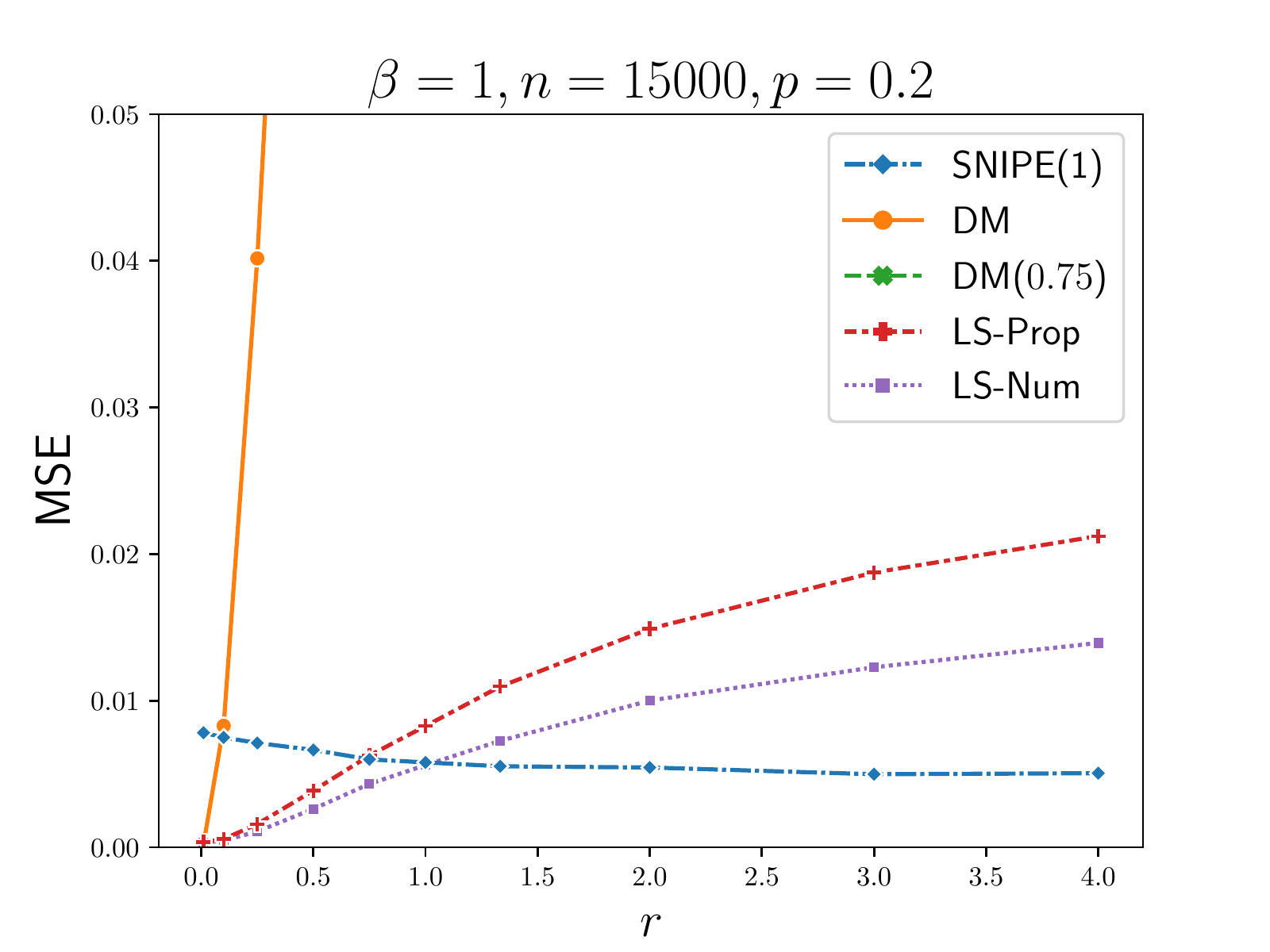}
        \includegraphics[width=\textwidth]{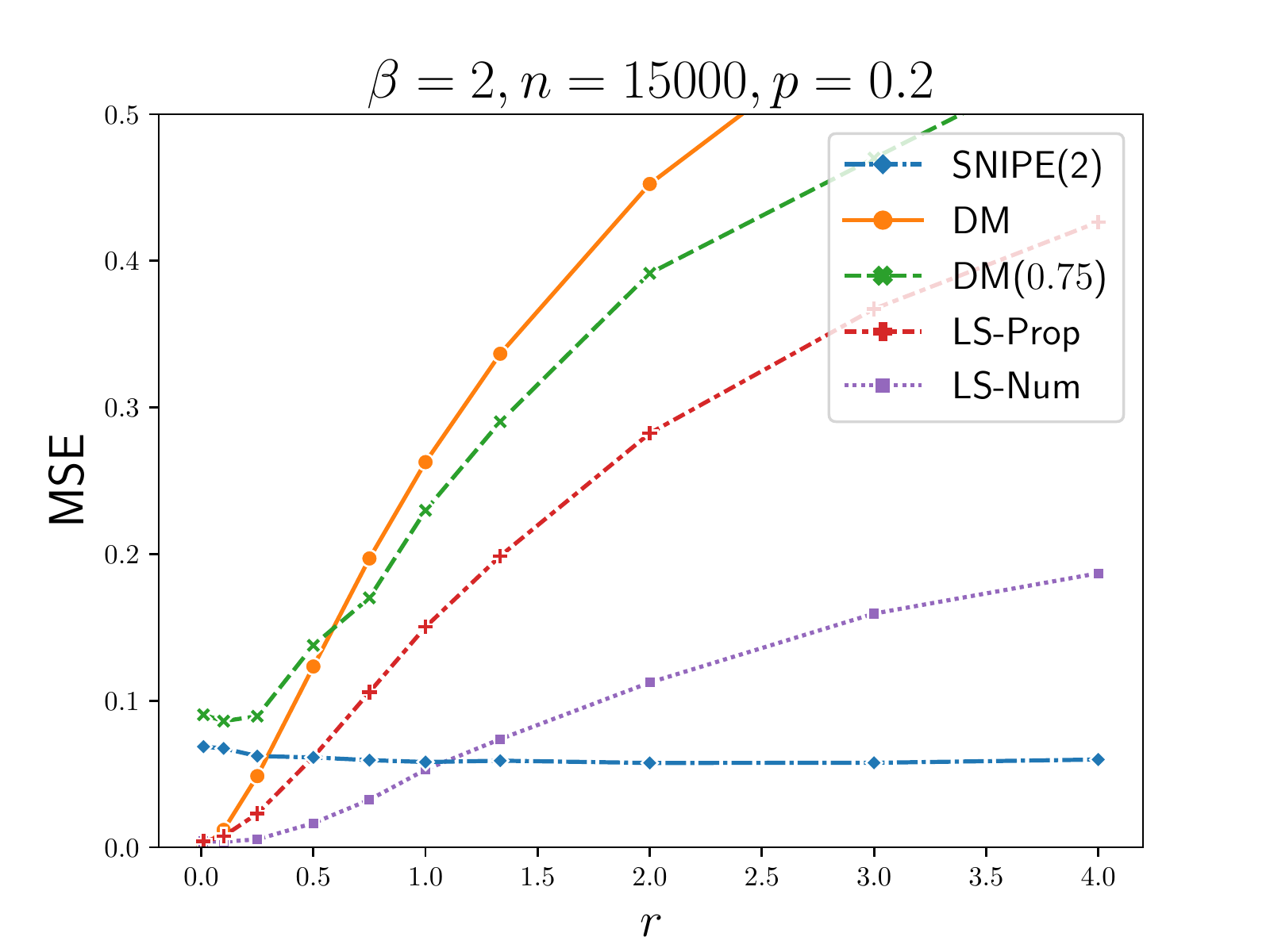}
        \caption{Varying direct:indirect effects}  \label{fig:ratioER_MSE}
    \end{subfigure}
    \begin{subfigure}[b]{0.32\textwidth}
        \centering
        \includegraphics[width=\textwidth]{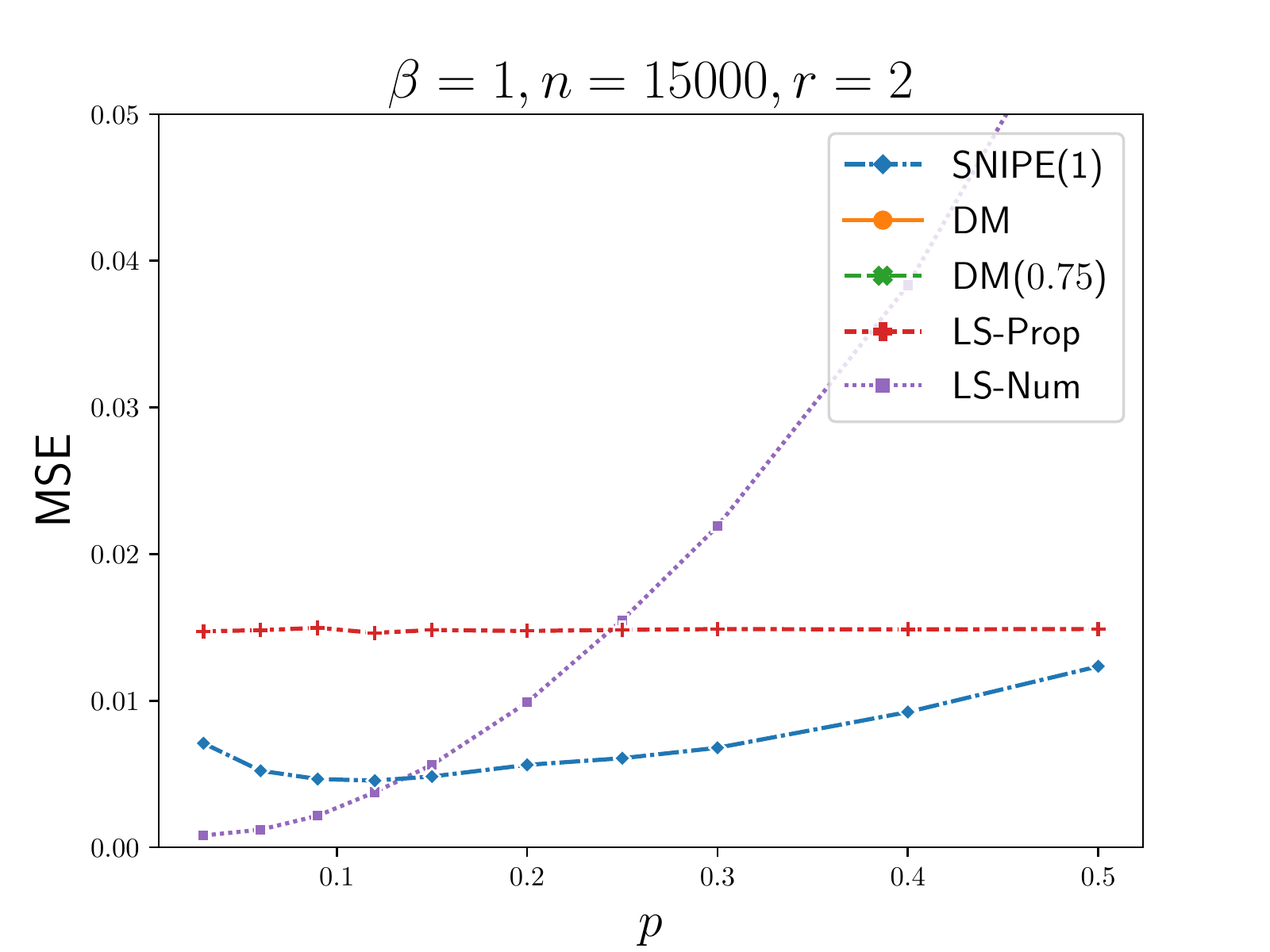}
        \includegraphics[width=\textwidth]{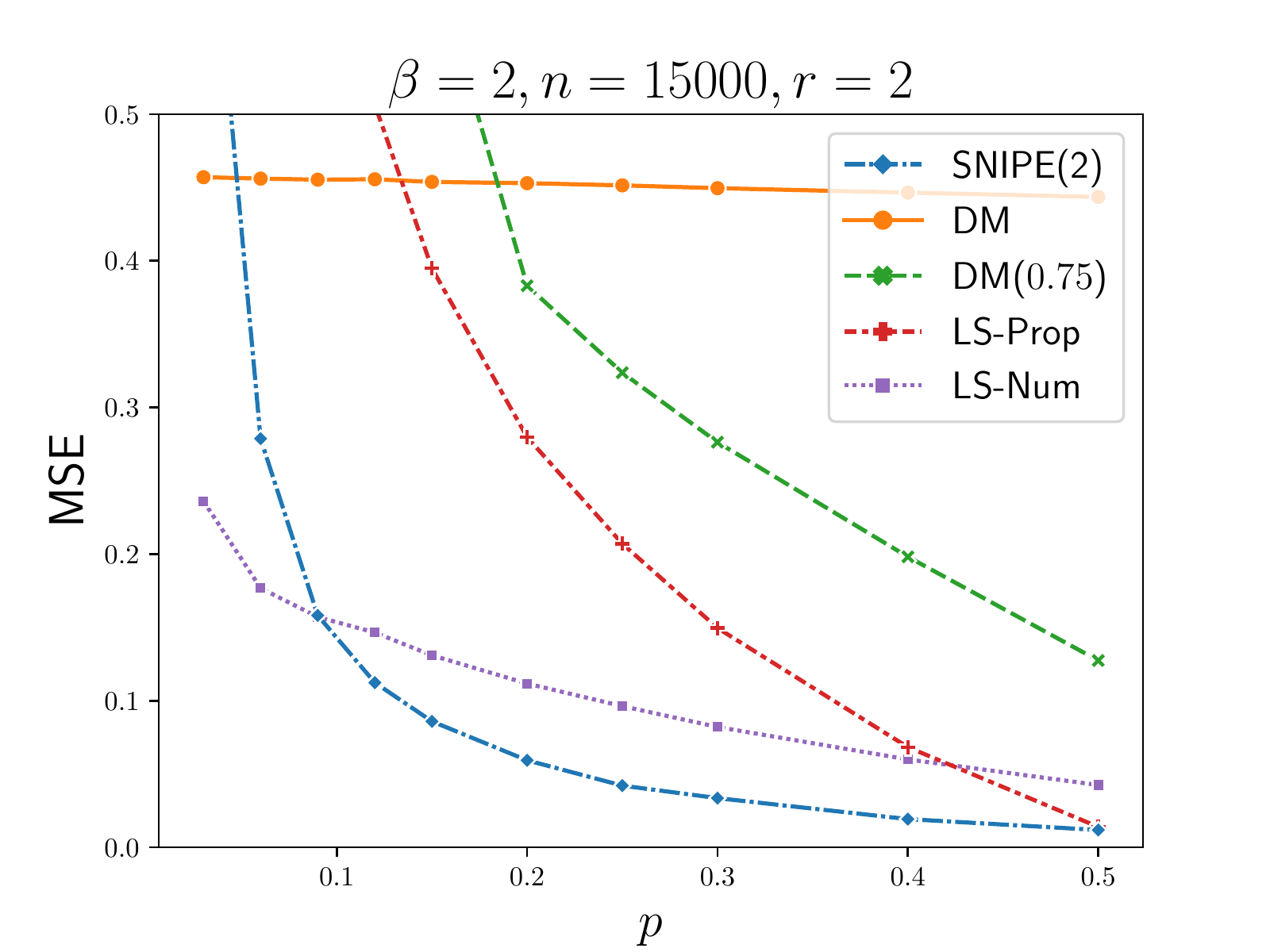}
        \caption{Varying treatment budget}  \label{fig:pER_MSE}
    \end{subfigure}
        \caption{Plots visualizing the MSE of various $\TTE$ estimators under Bernoulli design on Erd\H os-R\'enyi networks for both linear and quadratic potential outcomes models. The height of each line on a plot depicts the mean squared error when the model is normalized so that the true $\TTE$ is effectively equal to $1$. Alternatively, we can think of this as the variance of the normalized estimates. Our estimator under a $\beta$-order potential outcomes model is denoted SNIPE$(\beta)$ in the figure.} \label{fig:nonuniform_resultsER_MSE}
\end{figure}

\mcedit{\subsection{Variance Estimator Experiments}
We compare the empirical variance of SNIPE with the variance bound from Theorem \ref{thm:est_bias_var} and the variance estimator constructed from Aronow-Samii's method described in section \ref{ssec:variance_est}. As before, for each population size $n$, we sample $G$ networks from the Erd\H os-R\'enyi model described previously. 
For every configuration of parameters in the experiment, we sample $N$ treatment assignment vectors $\bz_1,\ldots,\bz_N$ from a uniform Bernoulli distribution with treatment probability $p$ to compute the TTE using each estimator.  For the variance experiments, we chose $G=10$ and $N=100$.}
\begin{table}
    \centering
\begin{tabular}{ p{3.2cm}||p{3.2cm}|p{3.2cm}|p{3.2cm}}
 \hline
 \multicolumn{4}{c}{Results for SNIPE$(1)$} \\
  & Experimental Variance &Variance Estimate&Variance Bound\\
 \hline
 \multirow{2}{4em}{$\mathbf{n \ (p=0.2, r=2)}$} &  &  & \\
 & & &\\
 $1000$ & $17.63$   & $16327.43$ & $488300.11$ \\
 $2500$ & $8.19$    & $3724.33$  & $245279.16$ \\
 $5000$ & $3.34 $   & $2270.81$  & $140688.10$ \\
 $7500$ & $2.33$    & $1408.50$  & $106981.33$ \\
 $10000$ & $1.95$   & $1437.32$  & $105354.62$ \\
 \multirow{2}{4em}{$\mathbf{p \ (n=5000, r=2)}$} & & & \\
 & & &\\
 $0.1$ & $2.93$   & $5439.83$  & $272939.86$ \\
 $0.2$ & $3.71$   & $2375.13$  & $174599.20$ \\
 $0.3$ & $4.33$   & $1259.70$  & $106516.48$ \\
 $0.4$ & $6.10$   & $1660.36$  & $122405.92$ \\
 $0.5$ & $8.06$   & $2197.06$  & $93585.34$ \\
 \multirow{2}{4em}{$\mathbf{r \ (n=5000, p=0.2)}$} & & & \\
 & & &\\
 $0.5$ & $1.06$   & $1480.23$  & $19410.59$ \\
 $1.0$ & $1.69$   & $1793.46$  & $44527.96$ \\
 $1.5$ & $2.50$   & $3802.44$  & $112805.72$ \\
 $2.0$ & $3.91$   & $2080.84$  & $142140.47$ \\
 & & &\\
 \hline
 \multicolumn{4}{c}{Results for SNIPE$(2)$} \\
 \hline
 & & &\\
  \multirow{2}{4em}{$\mathbf{n \ (p=0.2, r=2)}$} &  &  & \\
 & & &\\
 $1000$ & $255.59$   & $11046.03$ & $470718.02$ \\
 $2500$ & $92.21$    & $3873.10$  & $231156.58$ \\
 $5000$ & $46.36$    & $2962.08$  & $147815.75$ \\
 $7500$ & $29.80$    & $1495.31$  & $95530.25$ \\
 $10000$ & $21.28$   & $1463.03$  & $81113.89$ \\
 \multirow{2}{4em}{$\mathbf{p \ (n=5000, r=2)}$} & & & \\
 & & &\\
 $0.1$ & $114.95$  & $8635.44$  & $147108.70$ \\
 $0.2$ & $43.94$   & $2142.06$  & $133774.85$ \\
 $0.3$ & $24.36$   & $1373.10$  & $129746.07$ \\
 $0.4$ & $13.24$   & $1115.97$  & $128101.61$ \\
 $0.5$ & $9.19$    & $2435.29$  & $133271.42$ \\
 \multirow{2}{4em}{$\mathbf{r \ (n=5000, p=0.2)}$} & & & \\
 & & &\\
 $0.5$ & $12.99$   & $1627.89$  & $126414.68$ \\
 $1.0$ & $20.32$   & $1731.70$  & $129161.00$ \\
 $1.5$ & $30.94$   & $2054.55$  & $134921.48$ \\
 $2.0$ & $44.50$   & $2494.29$  & $144292.22$ \\
 \hline
\end{tabular}
    \caption{Table presenting the empirical variance of the SNIPE$(\beta)$ estimator for the $\TTE$ under Bernoulli$(p)$ design on Erd\H{o}s-R\'enyi networks with $\beta=1,2$. The variance bound is computed using the bound in Theorem \ref{thm:est_bias_var} and the variance estimate is computed using the Aronow-Samii estimator described in section \ref{ssec:variance_est}. The parameters $n$, $p$, and $r$ refer to the population size, the treatment probability, and the ratio of direct to indirect effects, respectively. }
    \label{tab:variance}
\end{table}
\mcedit{Table \ref{tab:variance} displays the effects of various network or estimator parameters on the experimental variance of SNIPE$(\beta)$ as well as the variance estimate and the theoretical variance bound, all under Bernoulli randomized design and all averaged over $GN$ samples of treatment vectors. As in previous experiments, we consider the effects of the population size ($n$), the treatment budget ($p$), the ratio between the network and direct effects ($r$), and the degree of the potential outcomes model ($\beta$). We list fixed values for the parameters in parentheses.}
\mcedit{The main observation we wish to draw attention to is the differences in the orders of magnitude amongst the empirical variance, the variance estimate, and the variance bound. It is clear from these results that obtaining a tighter variance estimator would be a valuable direction for future work.}
\section{Conclusions and Future Work}

We propose an estimator for the total treatment effect (TTE) under neighborhood interference and Bernoulli design when the graph is known. Our approach considers a potential outcomes model that is polynomial in the treatment vector $\bz$ with degree parameterized by $\beta$, which we assume to be much smaller than the maximum neighborhood size. This assumption is equivalent to constraining the order of interactions amongst treated neighbors to sets of size at most $\beta$. We derive theoretical bounds on the variance of our estimator under Bernoulli randomized design and show that we improve upon the variance of the Horvitz-Thompson estimator when $\beta$ is significantly lower than the maximum degree of the graph. We provide minimax lower bounds on the mean squared error of our estimator when the graph is $d-$regular and the treatment probabilities are the same for each individual. Furthermore, under additional boundedness conditions, we prove a central limit theorem for our estimator, allowing for conservative, asympotically valid confidence intervals using our proposed variance estimator. Through computational experiments, we illustrate that our estimator has lower MSE than the MSE of standard difference in means and least-squares estimators for the TTE. Our work uniquely complements the literature in that we consider how to incorporate and exploit structure in the potential outcomes model in a way that allows for a richer model class than the typical parametric model classes and does not reduce the effective treatment to a low-dimensional statistic. 

Our work presents many interesting and likely fruitful directions for future work. We summarize a few of these below.

\paragraph*{Optimized Experimental Designs} \medskip
\meedit{In this work, we studied the relationship between the complexity of the potential outcomes model (parameterized by $\beta$) and the difficulty of estimation. In this analysis, we made no structural assumptions on the network and restricted focus to independent Bernoulli experimental design. It is easy to conceive that a more careful selection of the experimental design, motivated by structural information of the causal network, could lead to improved performance of the estimator. For example, in graphs that are well-clustered, correlating the treatments within each cluster could allow some individuals to have more of their neighborhood treated, giving a better estimate of the magnitude of their treatment effect. The design philosophy of our estimator, as motivated in Section~\ref{subsec:pi_linear} through the lens of experiment replication, is not particular to Bernoulli design. As such, an enticing direction of future study would be to explore this estimator for other experimental designs and better understand the interplay between performance gains due to network structure and model structure.}

\paragraph*{Implications to Observational Studies} \medskip

\cyreplace{In this section, we consider the possible implications of our work to an unconfounded observational setting where the non-uniform treatment probabilities can be estimated and treatments are independent from each other conditioned on observed covariates. In particular, our theoretical results above hold for non-uniform Bernoulli designs where the probabilities $p_i$ can be arbitrarily set, as long as they are bounded away from zero and one. As a result, observational settings in which the treatments are independent after conditioning on observed covariates may be suitable to using a similar analysis. For example, consider a partially randomized setting that arises from non-compliance as described in \cite{DiTraglia2020}. They assume that treatment is offered to a random subset of individuals, but individuals decide whether or not to comply with the recommended treatment. They additionally posit the individualistic offer response assumption, under which an individual's decision to comply with the recommended treatment is not a function of whether or not the treatment was offered to other individuals in the population, and as such is also independent from the treatment outcomes of other individuals. As a result, the distribution over treatments can be modeled as a nonuniform Bernoulli randomized design, where $p_i$ depends on both the probability that the treatment is offered to individual $i$ as well as individual $i$'s probability of complying, which can be a function of the individual's type. If additionally we assume that $p_i$ is only a function of observable covariates, then we could plausibly estimate $p_i$ and estimate the total treatment effect using our proposed \mcdelete{pseudoinverse} estimator.}{While our stated theoretical results only hold for non-uniform Bernoulli designs, there are natural implications to the analysis of observational studies under appropriate unconfoundedness assumptions. In particular, if treatments across individuals are independent from each other conditioned on observed covariates, and if the conditional treatment probabilities could be estimated, then one could plausibly consider a plugin approach to modify our estimator for such observational data.} 
Formalizing how to extend our results to observational settings would be a fruitful and interesting direction for future work.

\paragraph*{Dealing with Model Misspecification} \medskip

Another interesting direction for future work centers around how to use our proposed class of estimators when the model parameter $\beta$ is unknown. In settings such as online social networks, it is reasonable to posit a low-degree interactions assumption on the network interference, which corresponds to adopting a potential outcomes model parameterized by a ground-truth value $\beta^{GT}$. To estimate the TTE, a researcher will select a value $\beta^{Exp}$ to use in defining their estimator. Without knowledge of the ground truth model, it is possible that $\beta^{GT} \ne \beta^{Exp}$. As this phenomenon of \textit{model misspecification} is pervasive through causal inference and more general machine learning domains, it would be useful to quantify the relationship between the degree of $\beta$-misspecification and any additional accrual of bias or variance. Another related question is whether a statistical test can be developed to aid in the correct choice of $\beta^{Exp}$ or to validate the low-degree polynomial structure of the model.

\section*{Acknowledgements}
 We gratefully acknowledge financial support from the National Science Foundation grants CCF-1948256 and CNS-1955997 and the National Science Foundation Graduate Research Fellowship grant DGE-1650441. We also thank Professor Nathan Kallus for his insightful feedback.

\bibliographystyle{plain}
\bibliography{refs}

\appendix 
\meedit{\section{The Explicit TTE Estimator for General  \texorpdfstring{$\beta$}{beta}} \label{sec:explicit_est}

Here, we derive an explicit formula for the $\TTE$ estimator under non-uniform Bernoulli design for general $\beta$. Recall (equation~\ref{eqn:TTE_hat implicit}) that the estimator has the form
\[
    \widehat{\TTE} = \tfrac{1}{n} \sum_{i=1}^{n} Y_i(\bz) \Big\langle \mathbb{E} \big[ \tbz_i \tbz_i^\intercal \big]^{-1} (\mathbf{1}_{|\cS_i^\beta|} - \mathbf{e}_1) ,  \tbz_i \Big\rangle,
\]
where each $\cS_i^\beta$ collects all subsets of $\cN_i$ with cardinality at most $\beta$ and each $\tbz_i$ is the vector of length $|S_i|$ with entries $(\tbz_i)_{\cS} = \prod_{j \in \cS} z_j$ indicating whether the entire subset $\cS$ has been assigned to treatment. Here, we index vectors and matrix entries by their corresponding sets (rather than numerical indices) for notational convenience. To begin, we focus our attention on the first argument vectors of these inner products. Entries of the matrix $\mathbb{E} \big[ \tbz_i \tbz_i^\intercal \big]$ have the form
\[
    \Big( \mathbb{E} \big[ \tbz_i \tbz_i^\intercal \big] \Big)_{\cS,\cT} 
    = \mathbb{E} \Big[ \prod_{j \in \cS} z_j \prod_{j' \in \cT} z_{j'} \Big]
    = \mathbb{E} \Big[ \prod_{j \in \cS \cup \cT} z_j \Big]
    = \prod_{j \in \cS \cup \cT} p_j.
\]
Here, the second equality uses the fact that each $z_j \in \{0,1\}$, and the third equality uses the independence of the treatment assignments. The following lemma establishes the invertibility of this matrix by giving an explicit expression for its inverse.

\begin{lemma} \label{lem:inverse_explicit}
    The matrix $\mathbb{E} \big[ \tbz_i \tbz_i^\intercal \big]$ is invertible, with entries of its inverse $A_i$ given by the formula 
     \[
        \big( A_i \big)_{\cS,\cT} = \prod_{j \in \cS} \tfrac{-1}{p_j} \prod_{k \in \cT} \tfrac{-1}{p_k} \sum_{\substack{\cU \in \cS_i^\beta \\ (\cS \cup \cT) \subseteq \cU}} \prod_{\ell \in \cU} \tfrac{p_\ell}{1-p_\ell}.
    \] 
\end{lemma}

\begin{proof}
    We'll argue that $\mathbb{E} \big[ \tbz_i \tbz_i^\intercal \big] A_i = I_{|\cS_i^\beta|}$ entry-wise. Note that both $\mathbb{E} \big[ \tbz_i \tbz_i^\intercal \big]$ and $A_i$ are symmetric matrices, so their product will be as well. First, we consider the diagonal entries of this matrix. Given any $\cS \in \cS_i^\beta$,
    \begin{align*}
        \big( \mathbb{E} \big[ \tbz_i \tbz_i^\intercal \big] A_i \big)_{\cS,\cS} 
        &= \sum_{\cT \in \cS_i^\beta} \prod_{j \in \cS \cup \cT} \!\!p_j \:\:\big( A_i \big)_{\cT, \cS} \\
        &= \sum_{\cT \in \cS_i^\beta} (-1)^{|\cS \cup \cT|} \prod_{k \in (\cS \cap \cT)} \tfrac{-1}{p_k} \sum_{\substack{\cU \in \cS_i^\beta \\ (\cS \cup \cT) \subseteq \cU}} \prod_{\ell \in \cU} \tfrac{p_\ell}{1-p_\ell} \\
        &= \sum_{\substack{\cU \in \cS_i^\beta \\ \cS \subseteq \cU}} \prod_{\ell \in \cU} \tfrac{p_\ell}{1-p_\ell} \sum_{\cT \subseteq \cU} (-1)^{|\cS \cup \cT|} \prod_{k \in (\cS \cap \cT)} \tfrac{-1}{p_k} \tag{reverse sum} \\
        &= (-1)^{|\cS|} \sum_{\substack{\cU \in \cS_i^\beta \\ \cS \subseteq \cU}} \prod_{\ell \in \cU} \tfrac{p_\ell}{1-p_\ell} \sum_{\cV \subseteq \cS} \prod_{k \in \cV} \tfrac{-1}{p_k} \underbrace{\sum_{\cW \subseteq (\cU \setminus \cS)} (-1)^{|\cW|}}_{\mathbb{I} \big( \cU \subseteq \cS \big)}  \tag{$\cV = \cT \cap \cS, \cW = \cT \setminus \cS$} \\
        &= (-1)^{|\cS|} \prod_{\ell \in \cS} \tfrac{p_\ell}{1-p_\ell} \sum_{\cV \subseteq \cS} \prod_{k \in \cV} \tfrac{-1}{p_k} \tag{only non-zero term is $\cU = \cS$} \\
        &= (-1)^{|\cS|} \prod_{\ell \in \cS} \tfrac{p_\ell}{1-p_\ell} \prod_{k \in \cS} \Big( 1 - \tfrac{1}{p_k} \Big) \tag{distributivity} \\
        &= 1.
    \end{align*}
    
    Next, we consider the off-diagonal entries. By the symmetry of $\mathbb{E} \big[ \tbz_i \tbz_i^\intercal \big] A_i$ it suffices to consider entries $(\cS', \cS)$ for which $\cS' \setminus \cS \ne \varnothing$. We have,
    
    \begin{align*}
        \big( M_i A_i \big)_{\cS',\cS} 
        &= \sum_{\cT \in \cS_i^\beta} \prod_{j' \in \cS' \cup \cT} \!\!p_{j'} \:\:\big( A_i \big)_{\cT, \cS} \\
        &= \sum_{\cT \in \cS_i^\beta} \prod_{j' \in \cS' \cup \cT} p_{j'} \prod_{j \in \cS} \tfrac{-1}{p_j} \prod_{k \in \cT} \tfrac{-1}{p_k} \sum_{\substack{\cU \in \cS_i^\beta \\ (\cS \cup \cT) \subseteq \cU}} \prod_{\ell \in \cU} \tfrac{p_\ell}{1-p_\ell} \\
        &= \prod_{j \in \cS} \tfrac{-1}{p_j} \sum_{\substack{\cU \in \cS_i^\beta \\ \cS \subseteq \cU}} \prod_{\ell \in \cU} \tfrac{p_\ell}{1-p_\ell} \sum_{\cT \subseteq \cU} \prod_{j' \in \cS' \cup \cT} p_{j'} \prod_{k \in \cT} \tfrac{-1}{p_k} \tag{reverse sum} \\
        &= \prod_{j \in \cS} \tfrac{-1}{p_j} \sum_{\substack{\cU \in \cS_i^\beta \\ \cS \subseteq \cU}} \prod_{\ell \in \cU} \tfrac{p_\ell}{1-p_\ell} \prod_{j' \in \cS'} p_{j'} \sum_{\cV \subseteq \cS'} \prod_{k \in \cV} \tfrac{-1}{p_k} \underbrace{\sum_{\cW \subseteq \cU \setminus \cS'} (-1)^{|\cW|}}_{\mathbb{I} \big( \cU \subseteq \cS' \big)} \tag{$\cV = \cT \cap \cS', \cW = \cT \setminus \cS'$} \\[-16pt]
        &= 0.
    \end{align*}
    Here, the last line follows because any non-zero term in the outer sum must correspond to some $\cU$ such that $\cS \subseteq \cU \subseteq \cS' \implies \cS \subseteq \cS'$. Our earlier assumption that $\cS' \setminus \cS \ne \varnothing$ ensures there is no such $\cU$.
\end{proof}

Using this lemma, we consider the entries in the first argument vector of each inner product. We have,
\begin{align*}
    \Big( A_i(\mathbf{1}_{|\cS_i^\beta|} - \mathbf{e}_1)  \Big)_{\cS} 
    &= \sum_{\cT \in \cS_i^\beta} \big( A_i \big)_{\cS,\cT} - \big( A_i \big)_{\cS,\varnothing}\\
    &= \prod_{j \in \cS} \tfrac{-1}{p_j} \bigg[ \sum_{\cT \in \cS_i^\beta} \prod_{k \in \cT} \tfrac{-1}{p_k} \sum_{\substack{\cU \in \cS_i^\beta \\ (\cS \cup \cT) \subseteq \cU}} \prod_{\ell \in \cU} \tfrac{p_\ell}{1-p_\ell} 
    - \sum_{\substack{\cU \in \cS_i^\beta \\ \cS \subseteq \cU}} \prod_{\ell \in \cU} \tfrac{p_\ell}{1-p_\ell} \bigg] \\
    &= \prod_{j \in \cS} \tfrac{-1}{p_j} \sum_{\substack{\cU \in \cS_i^\beta \\ \cS \subseteq \cU}} \prod_{\ell \in \cU} \tfrac{p_\ell}{1-p_\ell} \bigg[  \sum_{\cT \subseteq \cU} \prod_{k \in \cT} \tfrac{-1}{p_k} - 1 \bigg] \tag{reverse sum and factor} \\
    &= \prod_{j \in \cS} \tfrac{-1}{p_j} \sum_{\substack{\cU \in \cS_i^\beta \\ \cS \subseteq \cU}} \prod_{\ell \in \cU} \tfrac{p_\ell}{1-p_\ell} \bigg[ \prod_{k \in \cU} \tfrac{p_k-1}{p_k} - 1 \bigg] \tag{distributivity} \\
    &= \prod_{j \in \cS} \tfrac{-1}{p_j} \sum_{\substack{\cU \in \cS_i^\beta \\ \cS \subseteq \cU}} g(\cU) \prod_{\ell \in \cU} \tfrac{-1}{1-p_\ell},
\end{align*}
where we define $\displaystyle g(\cS) := \prod_{s \in \cS} (1-p_s) - \prod_{s \in \cS} (-p_s)$. Substituting back into the inner product, we calculate
\begin{align*}
    \Big\langle A_i (\mathbf{1}_{|\cS_i^\beta|} - \mathbf{e}_1),    \tbz_i \Big\rangle 
    &= \sum_{\cS \in \cS_i^\beta} \prod_{j \in \cS} \tfrac{-z_j}{p_j} \sum_{\substack{\cU \in \cS_i^\beta \\ \cS \subseteq \cU}} g(\cU) \prod_{\ell \in \cU} \tfrac{-1}{1-p_\ell} \\
    &= \sum_{\cU \in \cS_i^\beta} g(\cU) \prod_{\ell \in \cU} \tfrac{-1}{1-p_\ell} \sum_{\cS \subseteq \cU} \prod_{j \in \cS} \tfrac{-z_j}{p_j} \\
    &= \sum_{\cU \in \cS_i^\beta} g(\cU) \prod_{\ell \in \cU} \tfrac{-1}{1-p_\ell} \prod_{j \in \cU} \Big( 1 - \tfrac{z_j}{p_j} \Big) \\
    &= \sum_{\cU \in \cS_i^\beta} g(\cU) \prod_{j \in \cU} \tfrac{z_j-p_j}{p_j(1-p_j)} .
\end{align*}
Finally, we replace $\cU$ with $\cS$ to conform to earlier notation and obtain the explicit form for our estimator
\[
    \widehat{TTE} = \tfrac{1}{n} \sum_{i=1}^{n} Y_i \sum_{\cS \in \cS_i^\beta} g(\cS) \prod_{j \in \cS} \tfrac{z_j-p_j}{p_j(1-p_j)}.
\]

}
\section{Proof of Theorem~\ref{thm:est_bias_var}} \label{sec:main_proof}

\paragraph*{Unbiasedness}

The key insight that we use in our unbiasedness calculations comes from the following lemma.

\begin{lemma} \label{lem:exp_prod}
    If $\{z_j\}_{j \in [n]}$ are mutually independent with $z_j \sim \textrm{Bernoulli}(p_j)$, then for any $\cS, \cS' \subseteq [n]$,
    \[
        \E \Big[ \prod_{j \in \cS} \frac{z_j - p_j}{p_j(1-p_j)} \prod_{j' \in \cS'} z_{j'} \Big] = \Ind \big( \cS \subseteq \cS' \big) \cdot \!\!\!\prod_{j' \in \cS' \setminus \cS} p_{j'}. 
    \]
\end{lemma}

\begin{proof}
    By the mutual independence of the $\{z_j\}$, we can rewrite this expectations as a product, separating the variables into three groups.
    \[
        \E \Big[ \prod_{j \in \cS} \frac{z_j - p_j}{p_j(1-p_j)} \prod_{j' \in \cS'} z_{j'} \Big] 
        =
        \prod_{j \in \cS \setminus \cS'} \E\left[\frac{z_j - p_j}{p_j(1-p_j)}\right] \prod_{j' \in \cS' \setminus \cS} \E[z_{j'}] \prod_{j'' \in \cS \cap \cS'} \E\left[\frac{z_{j''}(z_{j''}-p_{j''})}{p_{j''}(1-p_{j''})}\right].
    \]
    Note that the expectations in the first product each simplify to $0$, so this expectation is non-zero only when $\cS \subseteq \cS'$. The expectations in the second product simplify to $p_{j'}$, and those in the third product each simplify to 1. These observations imply the lemma.
\end{proof}

The critical feature of this lemma is that this indicator function simplifies sums over arbitrary sets to sums over subsets $\cS \subseteq \cS'$. This additional structure \mereplace{permits simplification through techniques including the binomial theorem and M\"{o}bius inversion}{permits simplifications using the distributive property
\[
   \sum_{\cS \subseteq \cS'} \prod_{j \in \cS} a_j = \prod_{j \in \cS'} \big( 1 + a_j \big).
\]}

\meedit{
We leverage this fact in the following calculation. Given any $\cS' \in \cS_i^\beta$, we may simplify
\begin{align*}
    \E \Big[ \sum_{\cS \in \cS_i^\beta} g(\cS) \prod_{j \in \cS} \tfrac{z_j - p_j}{p_j(1-p_j)} \prod_{j' \in \cS'} z_{j'} \Big]
    &= \sum_{\cS \in \cS_i^\beta} g(\cS) \cdot \E \Big[ \prod_{j \in \cS} \frac{z_j - p_j}{p_j(1-p_j)} \prod_{j' \in \cS'} z_{j'} \Big] \tag{linearity} \\
    &= \sum_{\cS \subseteq \cS'} g(\cS) \!\!\prod_{j' \in \cS' \setminus \cS} p_{j'} \tag{Lemma~\ref{lem:exp_prod}} \\
    &= \prod_{j' \in \cS'} p_{j'} \sum_{\cS \subseteq \cS'} g(\cS) \!\prod_{j \in \cS} \tfrac{1}{p_j} \\
    &= \prod_{j' \in \cS'} p_{j'} \Big( \sum_{\cS \subseteq \cS'} \prod_{j \in \cS} \tfrac{1-p_j}{p_j} - \sum_{\cS \subseteq \cS'} \prod_{j \in \cS} (-1) \Big) \tag{definition of $g(\cS)$} \\
    &= \prod_{j' \in \cS'} p_{j'} \Big( \prod_{j \in \cS'} \big( 1 + \tfrac{1-p_j}{p_j} \big) - \mathbb{I}(\cS' = \varnothing) \Big) \tag{distributivity} \\
    &= 1 - \mathbb{I}(\cS' = \varnothing) \\
    &= \mathbb{I}(\cS' \ne \varnothing).
\end{align*}

Applying the linearity of expectation and the previous result, we calculate
\begin{align*}
    \E \Big[ \widehat{\TTE} \Big] 
    &= \tfrac{1}{n} \sum_{i=1}^{n} \E \Big[ Y_i \sum_{\cS \in \cS_i^\beta} g(\cS) \prod_{j \in \cS} \tfrac{z_j-p_j}{p_j(1-p_j)} \Big] \\
    &= \tfrac{1}{n} \sum_{i=1}^{n} \sum_{\cS' \in \cS_i^\beta} c_{i,\cS'} \: \E \Big[ \sum_{\cS \in \cS_i^\beta} g(\cS) \prod_{j \in \cS} \tfrac{z_j-p_j}{p_j(1-p_j)} \prod_{j' \in \cS'} z_{j'} \Big] \\
    &= \tfrac{1}{n} \sum_{i=1}^{n} \sum_{\substack{\cS' \in \cS_i^\beta \\ \cS' \ne \varnothing}} c_{i,\cS'} \\
    &= \TTE.
\end{align*}

}

\paragraph*{Variance Bound}

To bound the variance of this estimator, we make use of the following lemma to bound the magnitude of each $g(\cS)$ coefficient. 

\begin{lemma} \label{lem:gs_bound}
    For any $\cS \subseteq [n]$, $|g(\cS)| \leq 1$.
\end{lemma}

\begin{proof}
    First, note that $|g(\emptyset)| = 0 \leq 1$. Now, for any non-empty set $\cS$, let $i \in \cS$. Then,
    \begin{align*}
        \big| g(\cS) \big| &= \Big|\prod_{s \in \cS} (1-p_s) - \prod_{s \in \cS} (-p_s)\Big| \\
        &= \Big|(1-p_i) \prod_{s \in \cS \setminus \{i\}} (1-p_s) + p_i \prod_{s \in \cS \setminus \{i\}} (-p_s) \Big| \\
        &\leq (1-p_i) \prod_{s \in \cS \setminus \{i\}} (1-p_s) + p_i \prod_{s \in \cS \setminus \{i\}} p_s \tag{triangle inequality} \\
        &\leq 1-p_i + p_i \\
        &= 1.
    \end{align*}         
\end{proof}

This next lemma is used to bound the covariance terms that appear in our final calculation.

\begin{lemma} \label{lem:cov_bound}
    Suppose that $\{z_j\}_{j \in [n]}$ are mutually independent, with $z_j \sim \emph{Bernoulli}(p_j)$. Then, for any $\cS, \cS', \cT, \cT' \subseteq [n]$,
    \[
        0 \leq \emph{Cov} \Big[ \prod_{j \in \cS} \frac{z_j - p_j}{p_j(1-p_j)} \prod_{j' \in \cS'} \!\!z_{j'} \:,\: \prod_{k \in \cT} \frac{z_k - p_k}{p_k(1-p_k)} \prod_{k' \in \cT'} \!\!z_{k'} \Big]  \leq \Ind(\cS \triangle \cT \subseteq \cS' \cup \cT') \cdot \mereplace{(p (1-p))^{-\beta}}{\Big(\tfrac{1}{p(1-p)}\Big)^{|\cS \cap \cT|}},
    \]
    where $\cS \triangle \cT = (\cS \cup \cT) \setminus (\cS \cap \cT)$ indicates the symmetric difference of $\cS$ and $\cT$.
\end{lemma}

\begin{proof}
    We reason separately about the two terms in the covariance expansion. By Lemma~\ref{lem:exp_prod},
    \begin{equation} \label{eq:cov_term2}
        \E \Big[ \prod_{j \in \cS} \frac{z_j - p_j}{p_j(1-p_j)} \prod_{j' \in \cS'} \!\!z_{j'} \Big] \E \Big[ \prod_{k \in \cT} \frac{z_k - p_k}{p_k(1-p_k)} \prod_{k' \in \cT'} \!\!z_{k'} \Big]
        = \Ind \Big( \begin{matrix} \cS \subseteq \cS', \\ \cT \subseteq \cT' \end{matrix} \Big) \prod_{j' \in \cS'\setminus\cS} \!\!p_{j'} \prod_{k' \in \cT'\setminus\cT} \!\!p_{k'}.
    \end{equation}
    Next, we reason about the expectation of the product term. Since the $z_j$ are Bernoulli random variables, we can combine the products over $\cS'$ and $\cT'$, giving
    \begin{equation} \label{eq:cov_term1}
        \E \Big[ \prod_{j \in \cS} \frac{z_j - p_j}{p_j(1-p_j)} \prod_{k \in \cT} \frac{z_k - p_k}{p_k(1-p_k)} \prod_{j' \in \cS' \cup \cT'} z_{j'}\Big]. 
    \end{equation}
    We partition the elements of $\cS \cup \cS' \cup \cT \cup \cT'$ based on which of the products they are present in:
    
    \begin{enumerate}
        \item $j \in \cS \cap \cT \cap (\cS' \cup \cT')$: $j$ contributes a factor of $\E \big[ \frac{z_j^3-2z_j^2p_j+z_jp_j^2}{p_j^2(1-p_j)^2} \big] = \frac{1}{p_j}$.
        \item $j \in \cS \cap \cT \setminus (\cS' \cup \cT')$: $j$ contributes a factor of $\E \big[ \frac{z_j^2-2z_jp_j+p_j^2}{p_j^2(1-p_j)^2} \big] = \frac{1}{p_j(1-p_j)}$. 
        \item $j \in \cS \cap (\cS' \cup \cT') \setminus \cT$: $j$ contributes a factor of $\E \big[ \frac{z_j^2-z_jp_j}{p_j-p_j^2} \big] = 1$. 
        \item $j \in \cT \cap (\cS' \cup \cT') \setminus \cS$: $j$ contributes a factor of $\E \big[ \frac{z_j^2-z_jp_j}{p_j-p_j^2} \big] = 1$.  
        \item $j \in \cS \setminus \cT \setminus (\cS' \cup \cT')$: $j$ contributes a factor of $\E \big[ \frac{z_j-p_j}{p_j-p_j^2} \big] = 0$. 
        \item $j \in \cT \setminus \cS \setminus (\cS' \cup \cT')$: $j$ contributes a factor of $\E \big[ \frac{z_j-p_j}{p_j-p_j^2} \big] = 0$. 
        \item $j \in (\cS' \cup \cT') \setminus \cS \setminus \cT$: $j$ contributes a factor of $\E \big[ z_j \big] = p_j$. 
    \end{enumerate}
    Cases 5 and 6 ensure that \eqref{eq:cov_term1} is non-zero only when $\cS \subseteq (\cT \cup \cS' \cup \cT')$ and  $\cT \subseteq (\cS \cup \cS' \cup \cT')$, or equivalently when $\cS \triangle \cT \subseteq \cS' \cup \cT'$. This condition is necessary for \eqref{eq:cov_term2} to be non-zero. In addition, note that each $j$ from case 7 contributing a factor of $p_j$ to \eqref{eq:cov_term1} also contributes at least one factor of $p_j$ to \eqref{eq:cov_term2}. The remaining $j$ from other cases contribute a factor of at least 1. Notably, both \eqref{eq:cov_term1} and \eqref{eq:cov_term2} are non-negative, with \eqref{eq:cov_term1} dominating \eqref{eq:cov_term2}, so that the covariance is always bounded below by zero, and upper bounded by \eqref{eq:cov_term1}. \mereplace{\eqref{eq:cov_term1} will be largest when $|\cS \cap \cT| = \beta$ and are disjoint from $\cS' \cup \cT'$ such that all individuals in $\cS \cap \cT$ fall into case 2, allowing us to bound the covariance by $(p(1-p))^{-\beta}$.}{In \eqref{eq:cov_term1}, we can upper bound the contribution of each $j \in \cS \cap \cT$ by $(p(1-p))^{-1}$; note that our definition of $p$ ensures that $p_j(1-p_j) \geq p(1-p)$ for each $j$. We upper bound the contribution of each other $j$ by 1 which establishes the stated bound on the covariance}
\end{proof}

We are ready to bound the variance. If $\cN_i \cap \cN_{i'} = \emptyset$, then $Y_i(\bz) w_i(\bz)$ and $Y_{i'}(\bz) w_{i'}(\bz)$ are functions of disjoint sets of independent variables. Thus, $\Cov \big[ Y_i(\bz) w_i(\bz),Y_{i'}(\bz) w_{i'}(\bz) \big] = 0$. We let $\cM_i$ denote the set of individuals $i'$ such that $\cN_i \cap \cN_{i'} \neq \emptyset$, i.e. all individuals $i'$ that share an in-neighbor with individual $i$. Note that $|\cM_i| \leq \din \dout$. Applying the bilinearity of covariance and the triangle inequality, we have
\begin{align*}
    \Var[\widehat{TTE}] 
    &\leq \tfrac{1}{n^2}\sum_{i=1}^n \sum_{i' \in \cM_i} \sum_{\cS' \in \cS_i^\beta} |c_{i,\cS'}| \sum_{\cT' \in \cS_{i'}^\beta} |c_{i',\cT'}|
    \sum_{\cS \in \cS_i^\beta} |g(\cS)|
    \sum_{\cT \in \cS_{i'}^\beta} |g(\cT)| \\
    &\hspace{40pt} 
    \bigg| \Cov \Big[ \prod_{j \in \cS} \tfrac{z_j - p_j}{p_j(1-p_j)} \prod_{j' \in \cS'} \!\!z_{j'} \:,\: \prod_{k \in \cT} \tfrac{z_k - p_k}{p_k(1-p_k)} \prod_{k' \in \cT'} \!\!z_{k'} \Big] \bigg|.
\end{align*}
Plugging in our bounds from Lemmas~\ref{lem:gs_bound} and~\ref{lem:cov_bound}, we can simplify this bound:
\meedit{
\[
    \Var[\widehat{TTE}] 
    \leq \tfrac{1}{n^2}\sum_{i=1}^n \sum_{i' \in \cM_i} \sum_{\cS' \in \cS_i^\beta} |c_{i,\cS'}| \sum_{\cT' \in \cS_{i'}^\beta} |c_{i',\cT'}| \sum_{\cS \in \cS_i^\beta} \sum_{\cT \in \cS_{i'}^\beta} \Ind(\cS \triangle \cT \subseteq \cS' \cup \cT') \cdot \Big(\tfrac{1}{p(1-p)}\Big)^{|\cS \cap \cT|} \\
\]
Via the change of variables $\cU = \cS \cap \cT$, $\cS'' = \cS \setminus \cU$, $\cT'' = \cT \setminus \cU$, we may rewrite this
\[
    \tfrac{1}{n^2}\sum_{i=1}^n \sum_{i' \in \cM_i} \sum_{\cS' \in \cS_i^\beta} |c_{i,\cS'}| \sum_{\cT' \in \cS_{i'}^\beta} |c_{i',\cT'}| \sum_{\cU \in \cS_i^\beta} \Big(\tfrac{1}{p(1-p)}\Big)^{|\cU|} \# \bigg\{ (\cS'',\cT'') \in (\cS' \cup \cT')^2 \colon \begin{matrix} \cS'' \cap \cT'' = \varnothing \\ |\cS''|,|\cT''| \leq \beta - |\cU| \end{matrix} \bigg\} 
\]
\begin{align*}
    &\leq \tfrac{1}{n^2}\sum_{i=1}^n \sum_{i' \in \cM_i} \sum_{\cS' \in \cS_i^\beta} |c_{i,\cS'}| \sum_{\cT' \in \cS_{i'}^\beta} |c_{i',\cT'}| \sum_{\cU \in \cS_i^\beta} \Big(\tfrac{1}{p(1-p)}\Big)^{|\cU|} (2\beta)^{2\beta - 2|\cU|} \\
    &\leq \tfrac{1}{n^2} (2\beta)^{2\beta} \sum_{i=1}^n \sum_{i' \in \cM_i} \sum_{\cS' \in \cS_i^\beta} |c_{i,\cS'}| \sum_{\cT' \in \cS_{i'}^\beta} |c_{i',\cT'}| \sum_{\cU \in \cS_i^\beta} \Big(\tfrac{1}{4\beta^2 p(1-p)}\Big)^{|\cU|} \\
    &\leq \tfrac{\din \dout {Y_{\max}}^2}{n} (2\beta)^{2\beta} \sum_{k=0}^{\beta} \binom{\din}{k} \Big(\tfrac{1}{4\beta^2 p(1-p)}\Big)^{k} \\
    &\leq \tfrac{\din \dout {Y_{\max}}^2}{n} \Big( \tfrac{e\din}{\beta} \cdot \max(4\beta^2, \tfrac{1}{p(1-p)}) \Big)^\beta
\end{align*}
Here, the final inequality makes use of the bound $\sum_{k=0}^{\beta} \binom{d}{k} \leq \big( \tfrac{ed}{\beta} \big)^\beta$. Note that when $\beta = 1$, this bound simplifies to 
\[
    \frac{e \; \din^2 \; \dout \; {Y_{\max}}^2}{n p (1-p)}.
\]
}

\meedit{\section{Proof of Theorem \ref{thm:lower_bound}}\label{sec:lower_bd_pf}

We use a variation of LeCam's method, which allows us to recast the hardness of TTE estimation through the lens of hypothesis testing. We consider the following setting. 
\begin{itemize}
    \item The causal network is a $d$-regular directed graph (so $\din = \dout = d$ for all nodes) on $n$ nodes.
    \item The coefficients $\{ c_{i,\cS} \}$ are drawn from one of two possible Gaussian distributions, $\Gamma_0$ and $\Gamma_1$. The coefficients are mutually independent under both distributions with marginal probabilities
    \[
        \Gamma_0 \colon \quad c_{i,\cS} \sim \begin{cases}
            0 & |\cS| < \beta, \\
            N \Big( \delta \binom{d}{\beta}^{-1}, \binom{d}{\beta}^{-1} \Big) & |\cS| = \beta,
        \end{cases}
        \hspace{40pt}
        \Gamma_1 \colon \quad c_{i,\cS} \sim \begin{cases}
            0 & |\cS| < \beta, \\
            N \Big( -\delta \binom{d}{\beta}^{-1}, \binom{d}{\beta}^{-1} \Big) & |\cS| = \beta,
        \end{cases}
    \]
    where $\delta$ is a parameter that we will fix later.
    \item All units have a uniform treatment probability $p$. 
\end{itemize}
Using the mutual independence assumption, we see that under $\Gamma_0$, each $Y_i(\textbf{1})  \sim N \big( \delta, 1 \big)$ and $\TTE \sim N \big(\delta, \tfrac{1}{n} \big)$. Under $\Gamma_1$, each $Y_i(\textbf{1}) \sim N \big( -\delta, 1 \big)$ and $\TTE \sim N \big(-\delta, \tfrac{1}{n} \big)$.

We wish to compute a lower bound on the mean squared error of any estimator for $\TTE$ in this setting. To begin this calculation, we have
\begin{align}
    \inf_{\widehat{\TTE}} \sup_{\mathbf{c}} \E\limits_{\bz} \Big[ \big( \widehat{\TTE} - \TTE \big)^2 \Big| \;\mathbf{c}\; \Big] 
    &\geq \inf_{\widehat{\TTE}} \sup_{\mathbf{c}} \tfrac{\delta^2}{100} \Pr \big( \,|\widehat{\TTE} - \TTE| \geq \tfrac{\delta}{10} \,\big|\,\mathbf{c}\, \big) \notag \\
    &\geq \tfrac{\delta^2}{100} \inf_{\widehat{\TTE}} \: \max_{\Gamma \in \{\Gamma_0, \Gamma_1\}} \: \underset{\mathbf{c} \sim \Gamma}{\E} \Big[ \Pr \big( \,|\widehat{\TTE} - \TTE| \geq \tfrac{\delta}{10} \,\big|\,\mathbf{c}\, \big) \Big] \label{eq:infmax}.
\end{align}
Here, the first inequality lower bounds the conditional expectation by $\tfrac{\delta^2}{100}$ whenever $|\widehat{\TTE} - \TTE| \geq \tfrac{\delta}{10}$ and by $0$ whenever $|\widehat{\TTE} - \TTE| < \tfrac{\delta}{10}$. The second inequality replaces the supremum over all possible $\mathbf{c}$ with a maximum over two possible distributions over $\mathbf{c}$. 

Now, consider designing a hypothesis test $\Psi$ to distinguish these models, i.e. a test for $\Ind \big( \E_{\mathbf{c}}[\TTE] > 0 \big)$. Each estimator $\widehat{\TTE}$ gives rise to a decision rule $\widehat{\Psi} = \Ind \big( \widehat{\TTE} > 0 \big)$. If $\widehat{\Psi}$ is incorrect, then one of the following two scenarios must have occurred:
\begin{enumerate}
    \item $\big|\TTE - \E_{\mathbf{c}}[\TTE] \big| \geq \tfrac{9\delta}{10}$ \hfill by a Gaussian tail bound this has probability $< \exp \big( \tfrac{ -81n \delta^2}{200} \big)$
    \item $\big|\widehat{\TTE} - \TTE \big| \geq \tfrac{\delta}{10}$.
\end{enumerate}
When neither of these conditions is true, then $\big| \widehat{\TTE} - \E_{\mathbf{c}}[\TTE] \big| < \delta$, so $\widehat{\TTE}$ and $\E_{\mathbf{c}}[\TTE]$ have the same sign. Applying a union bound, we have
\[
    \Pr \Big( \widehat{\Psi} \ne \Ind \big( \E_{\mathbf{c}}[\TTE] > 0 \big) \Big) < \exp \big( \tfrac{ -81n \delta^2}{200} \big) + \Pr \big( \,|\widehat{\TTE} - \TTE| \geq \tfrac{\delta}{10} \big)
\]
Rearranging this inequality and plugging into \eqref{eq:infmax}, we may continue the simplification
\begin{align*}
    &\geq \tfrac{\delta^2}{100} \inf_{\widehat{\TTE}} \: \max_{\Gamma \in \{\Gamma_0, \Gamma_0\}} \: \underset{\mathbf{c} \sim \Gamma}{\E} \Big[  \Pr \big( \widehat{\Psi} \ne \Ind(\E[\TTE]>0) \,\big|\,\mathbf{c}\, \big) - \tfrac{\delta^2}{100} \exp \big( \tfrac{ -81n \delta^2}{200} \big) \\
    &\geq \tfrac{\delta^2}{100} \inf_{\Psi} \: \max_{\Gamma \in \{\Gamma_0, \Gamma_0\}} \: \underset{\mathbf{c} \sim \Gamma}{\E} \Big[  \Pr \big( \Psi \ne \Ind(\E[\TTE]>0) \,\big|\,\mathbf{c}\, \big) \Big] - \tfrac{\delta^2}{100} \exp \big( \tfrac{ -81n \delta^2}{200} \big) \\
    &\geq \tfrac{\delta^2}{200} \inf_{\Psi} \: \Big( \underset{\mathbf{c} \sim \Gamma_0}{\E} \big[ \Pr ( \Psi \ne 0 )  \big] + \underset{\mathbf{c} \sim \Gamma_1}{\E} \big[ \Pr ( \Psi \ne 1 ) \big] \Big) - \tfrac{\delta^2}{100} \exp \big( \tfrac{ -81n \delta^2}{200} \big) \\
    &= \tfrac{\delta^2}{200} \Big( 1 - \| P_0 - P_1 \|_{\textrm{TV}} \Big) - \tfrac{\delta^2}{100} \exp \big( \tfrac{ -81n \delta^2}{200} \big) \\
    &\geq \tfrac{\delta^2}{200} \Big( 1 - \sqrt{1 - \exp\big( -D_{\textrm{KL}}(P_0 \| P_1)\big)} \; \Big) - \tfrac{\delta^2}{100} \exp \big( \tfrac{ -81n \delta^2}{200} \big)
\end{align*}
Here, the second line follows because we have expanded the support of the infimum to all hypothesis tests, not just those that make use of an estimator $\widehat{TTE}$. The third line lower bounds the maximum over the distributions by an average. The $P_i$ in the fourth line represent the joint distribution over $(\bz, \mathbf{c})$ when $\mathbf{c}$ is drawn from $\Gamma_i$. The last line is an application of the Bretagnolle-Huber inequality.

Next, we derive an upper bound for this KL-Divergence. Applying the definition, we have
\[
    D_{\textrm{KL}} ( P_0 \| P_1 ) 
    = \E_{P_0} \Big[ \log \Big( \tfrac{P_0(\bY,\bz)}{P_1(\bY,\bz)} \Big) \Big] 
    = \E_{P_0} \Big[ \log \Big( \tfrac{P_0(\bY | \bz)}{P_1(\bY | \bz)} \Big) \Big]. 
\]
Here, the second equality uses the fact that the treatment assignments are independent from the random model coefficients. Now, conditioned on the treatment assignment $\bz$, the outcomes $\mathbf{Y}$ are distributed according to a Gaussian with independent coordinates. If we let $\cT(\bz) = \{ i \in [n] \colon z_i = 1 \}$ denote the set of treated individuals under $\bz$, then the marginal distribution of $Y_i(\bz)$ conditioned on $\bz$ is 
\[
    Y_i(\bz) \sim N \bigg( \pm \Big( \begin{smallmatrix} |\cN_i \cap \cT(\bz)| \\ \beta \end{smallmatrix} \Big) \Big( \begin{smallmatrix} d \\ \beta \end{smallmatrix} \Big)^{-1} \delta, \Big( \begin{smallmatrix} |\cN_i \cap \cT| \\ \beta \end{smallmatrix} \Big) \Big( \begin{smallmatrix} d \\ \beta \end{smallmatrix} \Big)^{-1} \bigg),
\]
where the positive expectation comes from $P_0$ and the negative expectation comes from $P_1$. Plugging in the density of the Gaussian into our KL-Divergence formula, we find that
\begin{align*}
    D&_{\textrm{KL}} ( P_0 \| P_1 ) \\ 
    &= \E_{P_0} \bigg[ \tfrac{-1}{2} \big( \begin{smallmatrix} d \\ \beta \end{smallmatrix} \big) \sum_{i=1}^{n} \Big( \begin{smallmatrix} |\cN_i \cap \cT(\bz)| \\ \beta \end{smallmatrix} \Big)^{-1} \Big[ \Big( Y_i(\bz) - \Big( \begin{smallmatrix} |\cN_i \cap \cT(\bz)| \\ \beta \end{smallmatrix} \Big) \big( \begin{smallmatrix} d \\ \beta \end{smallmatrix} \big)^{-1} \delta \Big)^2 - \Big( Y_i(\bz) + \Big( \begin{smallmatrix} |\cN_i \cap \cT(\bz)| \\ \beta \end{smallmatrix} \Big) \big( \begin{smallmatrix} d \\ \beta \end{smallmatrix} \big)^{-1} \delta \Big)^2 \Big] \bigg] \\
    &= 2\delta \sum_{i=1}^{n} \E_{P_0}[ Y_i(\bz)] \\
    &= 2\delta^2 \sum_{i=1}^{n} \big( \begin{smallmatrix} d \\ \beta \end{smallmatrix} \big)^{-1} \E_{\bz} \Big[ \Big( \begin{smallmatrix} |\cN_i \cap \cT(\bz)| \\ \beta \end{smallmatrix} \Big) \Big] \\
    &= 2\delta^2 \sum_{i=1}^{n} \big( \begin{smallmatrix} d \\ \beta \end{smallmatrix} \big)^{-1} \sum_{\substack{\cS \subseteq \cN_i \\ \cS = \beta}} \E_{\bz} \Big[ \Ind \big( \cS \subseteq \cT(\bz) \big) \Big] \\
    &= 2n\delta^2 p^\beta 
\end{align*}

Plugging into our earlier results, we find that
\[
    \inf_{\widehat{\TTE}} \sup_{\mathbf{c}} \E\limits_{\bz} \Big[ \big( \widehat{\TTE} - \TTE \big)^2 \Big| \;\mathbf{c}\; \Big] \geq \tfrac{\delta^2}{200} \Big( 1 - \sqrt{1 - \exp\big( -2n\delta^2 p^\beta \big)} - 2 \exp \big( \tfrac{ -81n \delta^2}{200} \big) \Big).
\]
Taking $\delta^2 = \tfrac{8}{3np^\beta}$, we obtain the upper bound
\[
    \tfrac{1}{75np^\beta} \Big( 1 - \sqrt{1 - \exp( \tfrac{-16}{3})} - 2 \exp \big( \tfrac{-27}{25p^\beta} \big) \Big).
\]
This is $\Omega\Big( \tfrac{1}{np^{\beta}} \Big)$ as long as $p^\beta < 0.16$. 

}
\section{Proof of Theorem~\ref{thm:clt}} \label{sec:CLT_proof}

\begin{proof}[Proof of Theorem~\ref{thm:clt}]
    We apply Theorem 3.6 from \cite{ross2011fundamentals} to the following defined random variables,
    \[ X_i := \tfrac{1}{n}\Big(\Yobs{i}w_i(\bz) - \E\big[\Yobs{i}w_i(\bz)\big]\Big), \quad \nu^2 := \Var\Big(\sum_{i \in [n]} X_i\Big), \quad \text{and} \quad W := \tfrac{1}{\nu} \sum_{i\in [n]}X_i, \]
    where $w_i(\bz) = \sum_{\substack{\cS \subseteq \cN_i \\ |\cS| \leq \beta}}g(\cS) \prod_{j \in \cS}\big(\frac{z_j}{p_j}-\frac{1-z_j}{1-p_j}\big)$.
    By construction, it follows that
    \[
        W = \tfrac{1}{n \nu} \sum_{i \in [n]} \Big( \Yobs{i}w_i(\bz) - \E\big[\Yobs{i}w_i(\bz)\big]\Big) 
        = \tfrac{1}{\nu}  \big(\widehat{\TTE} - \TTE\big),
    \]
    such that $\widehat{\TTE} = W\nu + \TTE$. The proof follows from verifying the conditions used in Theorem 3.6 of \cite{ross2011fundamentals}, computing appropriate bounds for the moments of $X_i$, and using the fact that $\nu^2 = \Omega(1/n)$ \mcedit{by Assumption \ref{assp:nondegeneracy} and the variance bound in Theorem \ref{thm:est_bias_var}}.
   
    First, we note that $\E[X_i] = 0$ by construction.
    To upper bound $\E[X_i^4]$, note that that for all $i$, we have 
    \begin{equation}
    \label{eqn:w_iBound}
        |w_i(\bz)| = \Big|\sum_{\substack{\cS \subseteq \cN_i \\ |\cS| \leq \beta}}g(\cS) \prod_{j \in \cS}\Big(\frac{z_j}{p_j}-\frac{1-z_j}{1-p_j}\Big)\Big| 
        \leq \Big|\sum_{\substack{\cS \subseteq \cN_i \\ |\cS| \leq \beta}} \frac{1}{p^{|\cS|}}\Big| \leq \Big(\frac{d}{p}\Big)^\beta.
    \end{equation}
    The second inequality in \eqref{eqn:w_iBound} follows from Lemma \ref{lem:gs_bound}, which upper bounds $|g(\cS)| \leq 1$. Furthermore we use the assumption that for all $i$, $p_i \in [p,1-p]$ and hence $|(z_j/p_j - (1-z_j)/(1-p_j))| \leq 1/p$. The final inequality in \eqref{eqn:w_iBound} follows from the fact that $|\cS| \leq \beta$ and the number of subsets of $\cN_i$ for any $i$ is bounded above by $d^\beta$.
    
    Recognizing that $\E[\Yobs{i} w_i(\bz)] = \sum_{\substack{\cS' \subseteq \cN_i \\ 1 \leq |\cS'|\leq \beta}} c_{i,\cS'} \leq Y_{\max},$ we obtain a bound $|X_i| \leq \frac{1}{n}\big(Y_{\max} (\frac{d}{p})^{\beta} + Y_{\max}\big)$ using the triangle inequality and the bound on $|w_i(\bz)|$ from \eqref{eqn:w_iBound}. From this, we can bound
    \begin{align}
    \E[X_i^4] &\leq \frac{1}{n^4}\E\bigg[\bigg(Y_{\max} \Big(\frac{d}{p}\Big)^{\beta} + \ Y_{\max}\bigg)^4 
    \bigg] \nn \\
    &= \frac{Y_{\max}^4}{n^4}\left[\Big(\frac{d}{p}\Big)^{4\beta} + 4\Big(\frac{d}{p}\Big)^{3\beta} + 6\Big(\frac{d}{p}\Big)^{2\beta} + 4\Big(\frac{d}{p}\Big)^{\beta} + 1\right]\mcedit{.} \nn\\
    &\quad \hspace{40mm}
    \end{align} 
    Since $d/p > 1$ and $\beta \geq 1$, we can then bound 
    \begin{equation}
       \E[X_i^4] = O\left(\frac{1}{n^4} \bigg(\frac{d}{p}\bigg)^{4\beta} Y_{\max}^4\right). \label{eqn:4thMomentBd}
    \end{equation}
    
    In the same way, we can bound 
    \begin{align}
        \E[|X_i|^3] &\leq \frac{1}{n^3} \E\left[\bigg(Y_{\max} \Big(\frac{d}{p}\Big)^{\beta} + \ Y_{\max}\bigg)^3 \right]\nn \\
        &= \frac{Y_{\max}^3}{n^3}\left[\Big(\frac{d}{p}\Big)^{3\beta}+3\Big(\frac{d}{p}\Big)^{2\beta}+3\Big(\frac{d}{p}\Big)^{\beta}+1\right].\nn
    \end{align} 
    Since $d/p > 1$ and $\beta \geq 1$, we obtain the bound
    \begin{equation}
        \E[|X_i|^3] = O\left(\frac{1}{n^3} \bigg(\frac{d}{p}\bigg)^{3\beta}  Y_{\max}^3\right). \label{eqn:3rdMomentBd}
    \end{equation}
    
    Let $\cM_i$ denote the set of individuals $i'$ such that $\cN_i \cap \cN_{i'} \neq \emptyset$, i.e. all individuals $i'$ that share an in-neighbors with individual $i$. The set $\cM_i$ characterizes the \textit{dependency neighborhood} of $i$, as $X_i$ and $X_j$ are dependent if and only if there is a shared neighbor $k$ such that $X_i$ and $X_j$ both depend on $z_k$. It follows that the maximum size of any dependency neighborhood $D = \max_{i\in[n]}|\cM_i| \leq \din \dout \leq \dmax^2$. Then, by plugging in the bounds for $D, \E[|X_i|^3]$ and $\E[X_i^4]$ into Theorem 3.6 of \cite{ross2011fundamentals} results in the following bound
    \begin{equation}
    \label{eqn:WasserBound}
        d_\text{W}(W,Z) \leq O\left(\frac{\dmax^4}{\nu^3} \frac{Y_{\max}^3 \dmax^{3\beta}}{n^2 p^{3\beta}} + \frac{\dmax^3}{\nu^2} \sqrt{ \frac{Y_{\max}^4 \dmax^{4\beta}}{n^3 p^{4\beta}}}\right).
    \end{equation}
    where recall that $Z$ is a standard normal random variable.
    \mcreplace{Our final step is to lower bound $\nu^2$ and argue that under the boundedness asumptions, the Wasserstein distance between $W$ and $Z$ goes to $0$ as $n \to \infty$.
    
    It follows by the law of total variance that
    
\begin{align*}
    \nu^2 &= \Var\big(\widehat{\TTE}\big) \\
    &= \E_{\bz}\left[ \Var_{\epsilon}\left[ \left.\frac{1}{n}\sum_{i=1}^n \Yobs{i}w_i(\bz) ~\right|~ \bz\right] \right] + \Var_{\bz}\left[ \E_{\epsilon}\left[ \left. \frac{1}{n}\sum_{i=1}^n \Yobs{i}w_i(\bz) ~\right|~ \bz\right] \right] \\
    &= \E\left[ \frac{\sigma^2}{n^2}\sum_{i=1}^n w_i(\bz)^2 \right] + \Var\left[ \frac{1}{n}\sum_{i=1}^n Y_i(\bz) w_i(\bz) \right]
\end{align*}
Recall that  $g(\cS) = \prod_{j \in \cS} (1-p_j) - \prod_{j \in \cS} (-p_j)$ for $\cS \subseteq \cN_i$, and for $\cS, \cS' \subset \cN_i$,
\[\E\left[\prod_{j \in \cS}\Big(\frac{z_j}{p_j}-\frac{1-z_j}{1-p_j}\Big) \prod_{j' \in \cS'}\Big(\frac{z_{j'}}{p_{j'}}-\frac{1-z_{j'}}{1-p_{j'}}\Big)\right] = \Ind(\cS = \cS') \prod_{j \in \cS} \frac{1}{p_j (1-p_j)}.\]
The first term in the variance can be lower bounded by 
\begin{align*}
\E\left[ \frac{\sigma^2}{n^2}\sum_{i=1}^n w_i(\bz)^2 \right]
    &= \frac{\sigma^2}{n^2}\sum_{i=1}^n \E\left[ \left(\sum_{\substack{\cS \subseteq \cN_i \\ |\cS| \leq \beta}}g(\cS) \prod_{j \in \cS}\Big(\frac{z_j}{p_j}-\frac{1-z_j}{1-p_j}\Big)\right)^2 \right] \\  
    &= \frac{\sigma^2}{n^2}\sum_{i=1}^n \sum_{\substack{\cS \subseteq \cN_i \\ |\cS| \leq \beta}}g(\cS)^2 \prod_{j \in \cS} \frac{1}{p_j (1-p_j)} \\
    &= \frac{\sigma^2}{n^2}\sum_{i=1}^n \sum_{\substack{\cS \subseteq \cN_i \\ |\cS| \leq \beta}} \left(\prod_{j \in \cS} (1-p_j)^2 + \prod_{j \in \cS} p_j^2 - 2 \prod_{j \in \cS} (1-p_j) (-p_j) \right) \frac{1}{\prod_{j \in \cS} p_j (1-p_j)} \\
    &= \frac{\sigma^2}{n^2}\sum_{i=1}^n \sum_{\substack{\cS \subseteq \cN_i \\ |\cS| \leq \beta}} \left(\prod_{j \in \cS} \frac{1-p_j}{p_j} + \prod_{j \in \cS} \frac{p_j}{1-p_j} - 2 (-1)^{|\cS|} \right) \\
    &\geq \frac{2 \sigma^2}{n^2}\sum_{i=1}^n \sum_{k = 1}^{\min(\lceil \beta/2\rceil, |\cN_i|)} \binom{|\cN_i|}{2k-1} \\
    &\geq \frac{2 \sigma^2}{n},
\end{align*}
where the last inequality uses that $|\cN_i| \geq 1$ since it must at least contain $i$ itself, and the second to last inequality uses the fact that for $|\cS|$ odd, $\left(\prod_{j \in \cS} \frac{1-p_j}{p_j} + \prod_{j \in \cS} \frac{p_j}{1-p_j} - 2 (-1)^{|\cS|} \right) \geq 2$.

The second term in the variance can be expanded to
\begin{align*}
    \Var\left[ \frac{1}{n}\sum_{i=1}^n Y_i(\bz) w_i(\bz) \right]
    &= \frac{1}{n^2} \sum_{\substack{i,\cS' \subseteq \cN_i \\ |\cS'| \leq \beta}} c_{i,\cS'}  \sum_{\substack{i',\cT' \subseteq \cN_{i'} \\ |\cT'| \leq \beta}} c_{i',\cT'} \sum_{\substack{\cS \subseteq \cN_i \\ |\cS| \leq \beta}} g(\cS) \sum_{\substack{\cT \subseteq \cN_{i'} \\ |\cT| \leq \beta}} g(\cT) \\
    &\hspace{3em}\times\Cov\left[ \prod_{j \in \cS} \Big( \frac{z_j}{p_j} - \frac{1-z_j}{1-p_j} \Big) \prod_{j' \in \cS'} z_{j'} \:,\: \prod_{k \in \cT} \Big( \frac{z_k}{p_k} - \frac{1-z_k}{1-p_k} \Big) \prod_{k' \in \cT'} z_{k'}\right].
\end{align*}
By Lemma \ref{lem:cov_bound}, the covariance terms are nonnegative. By the assumption that $p_i \leq \frac12$ for all $i$, it follows that $g(\cS)$ is always nonnegative. Let $\cM_i := \{i' ~\text{s.t.}~ \cN_i \cap \cN_{i'} \neq \emptyset\}$ denote the set of individuals that share in-neighbors with individual $i$. Additionally, by the assumption that the coefficients $c_{i,\cS}$ are all the same sign and that minimally $|c_{i,\emptyset}| \geq \rho > 0$, then it follows that 
\begin{align*}
    \Var\left[ \frac{1}{n}\sum_{i=1}^n Y_i(\bz) w_i(\bz) \right]
    &\geq \frac{1}{n^2} \sum_{i} c_{i,\emptyset}  \sum_{i' \in \cM_i} c_{i',\emptyset} \sum_{\substack{\cS \subseteq \cN_i \cap \cN_{i'} \\ |\cS| \leq \beta}} g(\cS)^2
    \Var\left[ \prod_{j \in \cS} \Big( \frac{z_j}{p_j} - \frac{1-z_j}{1-p_j} \Big)\right], \\
    &= \frac{1}{n^2} \sum_{i} c_{i,\emptyset}  \sum_{i' \in \cM_i} c_{i',\emptyset} \sum_{\substack{\cS \subseteq \cN_i \cap \cN_{i'} \\ |\cS| \leq \beta}} g(\cS)^2
    \prod_{j \in \cS} \frac{1}{p_j (1-p_j)},\\
    &\overset{(a)}{\geq} \frac{1}{n^2} \sum_{i} c_{i,\emptyset}^2 \sum_{\substack{\cS \subseteq \cN_i \\ |\cS| \leq \beta}} g(\cS)^2
    \prod_{j \in \cS} \frac{1}{p_j (1-p_j)},\\
    &\overset{(b)}{\geq} \frac{1}{n^2} \sum_{i} c_{i,\emptyset}^2 \sum_{j \in \cN_i} \frac{1}{p_j (1-p_j)},\\
    &\overset{(c)}{\geq} \frac{4\rho^2}{n},
\end{align*}
where (a) uses the fact that $\cM_i$ minimally includes individual $i$ itself, and (b) limits the second sum to the singleton sets, for which $g(\{j\}) = 1$, and (c) uses the fact that $1/p_j(1-p_j) \geq 4$ and $|\cN_i|\geq 1$ as it must contain $i$ itself.

Plugging in our lower bound on $\nu^2$ into \eqref{eqn:WasserBound}}{Since Assumption \ref{assp:nondegeneracy} implies that $\nu^2 \geq O(1/n)$}, it follows that
\mcreplace
  {\begin{align*}
        d_\text{W}(W,Z) 
        &= O\left(\frac{(Y_{\max}^3 + \sigma^3)\dmax^{3\beta+4}}{n^{1/2} (\sigma^2 + \rho^2)^{3/2} p^{3\beta}} +  \frac{(Y_{\max}^4 + \sigma^4)^{1/2}\dmax^{2\beta+3}}{n^{1/2} (\sigma^2 + \rho^2)  p^{2\beta}}\right).
    \end{align*}}
    {
    \begin{align*}
         d_\text{W}(W,Z)
        &= O\left(\frac{Y_{\max}^3\dmax^{3\beta+4}}{n^{1/2} p^{3\beta}} +  \frac{Y_{\max}^2\dmax^{2\beta+3}}{n^{1/2}  p^{2\beta}}\right).
    \end{align*}
    }
By boundedness of $Y_{\max}, \dmax, \beta,$ and as $p = \omega(n^{-1/4\beta})$, the Wasserstein distance between $W$ and $Z \sim \textrm{N}(0,1)$ goes to $0$ as $n\to \infty$.
As $\widehat{\TTE} = W\nu + \TTE$, it follows that the distribution of $\widehat{\TTE}$ converges to a normal with mean $\TTE$ and variance $\nu^2$.
\end{proof}

\meedit{\section{Other Estimands} \label{sec:proof_other_est}



\paragraph*{Average Treatment Effect} \medskip

The \textit{average treatment effect} (ATE) measures the average effect that one's own treatment has on their outcome (assuming no one else is treated).
\begin{equation*}
    \textrm{ATE} := \tfrac{1}{n} \sum_{i=1}^{n} \Big( Y_i(\mathbf{e_i}) - Y_i(\mathbf{0}) \Big) = \tfrac{1}{n} \sum_{i=1}^{n} c_{i,\{i\}}.
\end{equation*}
The estimator for ATE takes the form
\[
    \widehat{\textrm{ATE}} = \tfrac{1}{n} \sum_{i=1}^{n} Y_i(\bz) \sum_{\cS \in \cS_i^\beta} (A_i)_{\cS,\{i\}} \prod_{j \in \cS} z_j,
\]
where $A_i$ is the inverse of $\E[ \tbz_i \tbz_i^\intercal]$ as defined in \Cref{lem:inverse_explicit}. Plugging in the explicit form of these matrix entries, we have:
\begin{align*}
    \widehat{\textrm{ATE}} &= \tfrac{1}{n} \sum_{i=1}^{n} Y_i(\bz) \sum_{\cS \in \cS_i^\beta} \tfrac{-1}{p_i} \prod_{j \in \cS} \tfrac{-z_j}{p_j} \sum_{\substack{\cU \in \cS_i^\beta \\ (\cS \cup \{i\}) \subseteq \cU}} \prod_{\ell \in \cU} \tfrac{p_\ell}{1-p_\ell} \\
    &= \tfrac{1}{n} \sum_{i=1}^{n} Y_i(\bz) \cdot \tfrac{-1}{p_i} \sum_{\substack{\cU \in \cS_i^\beta \\ i \in \cU}} \prod_{\ell \in \cU} \tfrac{p_\ell}{1-p_\ell} \sum_{\cS \subseteq \cU} \prod_{j \in \cS} \tfrac{-z_j}{p_j}  \tag{reverse inner sums} \\
    &= \tfrac{1}{n} \sum_{i=1}^{n} Y_i(\bz) \cdot \tfrac{-1}{p_i} \sum_{\substack{\cU \in \cS_i^\beta \\ i \in \cU}} \prod_{\ell \in \cU} \tfrac{p_\ell}{1-p_\ell} \prod_{j \in \cU} \tfrac{p_j-z_j}{p_j} \\
    &= \tfrac{1}{n} \sum_{i=1}^{n} Y_i(\bz) \cdot \tfrac{-1}{p_i} \sum_{\substack{\cU \in \cS_i^\beta \\ i \in \cU}} \prod_{j \in \cU} \tfrac{p_j-z_j}{1-p_j}.
\end{align*}

\paragraph*{Conditional Average Treatment Effect} \medskip

Given any subdemographic of the population $\mathcal{D} \subseteq [n]$, the \textit{Conditional average treatment effects} (CATE) of $\mathcal{D}$ is the average effect to an individual of the demographic that is treated in isolation
\[
    \textrm{CATE}(\mathcal{D}) := \tfrac{1}{|\mathcal{D}|} \sum_{i \in \mathcal{D}} \Big( Y_i(\mathbf{e_i}) - Y_i(\mathbf{0}) \Big) = \tfrac{1}{|\mathcal{D}|} \sum_{i \in \mathcal{D}} c_{i,\{i\}}.
\]
By the same calculation as above, we estimate the conditional average treatment effect
\[
    \widehat{\textrm{CATE}(\mathcal{D})}
    = \tfrac{1}{|\mathcal{D}|} \sum_{i \in \mathcal{D}} \widehat{c_{i,\{i\}}} 
    = \tfrac{1}{|\mathcal{D}|} \sum_{i \in \mathcal{D}} Y_i(\bz) \cdot \tfrac{-1}{p_i} \sum_{\substack{\cU \in \cS_i^\beta \\ i \in \cU}} \prod_{j \in \cU} \tfrac{p_j-z_j}{1-p_j}.
\]

\paragraph*{Size-Dependent Treatment Effects} \medskip

Although not standard in the literature, as its significance is largely brought about by the low-degree polynomial structure of our potential outcomes model, another family of causal estimands can be used to understand the magnitude of the treatment effects that individuals experience as a result of different sized subsets of their neighborhood being treated. We define the $\alpha$-treatment effect,
\[
    \textrm{TE}(\alpha) := \tfrac{1}{n} \sum_{i=1}^{n} \sum_{\substack{\cS' \subseteq \cN_i \\ |\cS'| = \alpha}} c_{i,\cS'},
\]
for some parameter $\alpha \leq \beta$. That is, the $\alpha$-treatment effect measures only those effects that subsets $\cS'$ of size $\alpha$ have on the outcome of each individual and averages the cumulative effect over the individuals. For example, even though the polynomial degree of a model could be large, the causal effects associated to higher order interactions could be small such that the potential outcomes could be well approximated with a linear model ($\beta = 1$). In such an event, one would expect that \textrm{TE}(1) would be close to \TTE, and \textrm{TE}($\alpha$) would be significantly smaller in magnitude for $\alpha > 1$. As $\textrm{TE}(\alpha)$ is again a linear combination of the model parameters, we can obtain an unbiased estimate using the \mcreplace{pseudoinverse estimator}{same framework as before}. In this case, the estimator takes the form
\begin{align*}
    \widehat{\textrm{TE}(\alpha)} 
    &= \tfrac{1}{n} \sum_{i=1}^{n} Y_i(\bz) \sum_{\cS \in \cS_i^\beta} \sum_{\substack{\cT \subseteq \cN_i \\ |\cT| = \alpha}} (A_i)_{\cS,\cT} \prod_{j \in \cS} z_j \\
    &= \tfrac{1}{n} \sum_{i=1}^{n} Y_i(\bz) \sum_{\cS \in \cS_i^\beta} \sum_{\substack{\cT \subseteq \cN_i \\ |\cT| = \alpha}} \prod_{j \in \cS} \tfrac{-z_j}{p_j} \prod_{k \in \cT} \tfrac{-1}{p_k} \sum_{\substack{\cU \in \cS_i^\beta \\ (\cS \cup \cT) \subseteq \cU}} \prod_{\ell \in \cU} \tfrac{p_\ell}{1-p_\ell} \\
    &= \tfrac{1}{n} \sum_{i=1}^{n} Y_i(\bz) \sum_{\substack{\cT \subseteq \cN_i \\ |\cT| = \alpha}} \prod_{k \in \cT} \tfrac{-1}{p_k} \sum_{\substack{\cU \in \cS_i^\beta \\ \cT \subseteq \cU}} \prod_{\ell \in \cU} \tfrac{p_\ell}{1-p_\ell} \sum_{\cS \subseteq \cU} \prod_{j \in \cS} \tfrac{-z_j}{p_j} \tag{reorder sums} \\
    &= \tfrac{1}{n} \sum_{i=1}^{n} Y_i(\bz) \sum_{\substack{\cT \subseteq \cN_i \\ |\cT| = \alpha}} \prod_{k \in \cT} \tfrac{-1}{p_k} \sum_{\substack{\cU \in \cS_i^\beta \\ \cT \subseteq \cU}} \prod_{j \in \cU} \tfrac{p_j-z_j}{1-p_j} \tag{distributivity}.
\end{align*}}

\end{document}